\documentclass[smallcondensed,envcountsect]{svjour3}

\usepackage[vlined,figure,noline,noend,linesnumbered]{algorithm2e}
\SetKwInput{Algorithm}{Algorithm}
\SetKwInput{Function}{Function}
\SetEndCharOfAlgoLine{}
\newcommand\mycommfont[1]{\scriptsize\ttfamily}

\usepackage[english]{babel}

\usepackage{amsmath,amsthm,amssymb}
\usepackage{hyperref}
\usepackage[font={small,bf},labelfont=bf]{caption}
\usepackage{amsfonts}
\usepackage{url}
\usepackage{subcaption}
\usepackage{graphicx}
\usepackage{wrapfig}
\usepackage{cite}
\usepackage{xcolor}

\newtheorem{thm}{Theorem}[section]
\numberwithin{equation}{section}

\newcommand{\SLH}{$SLH$}
\newcommand{\TR}{$TR$}
\newcommand{\TA}{$TA$}

\newcommand{\figref}[1]{Figure~\ref{#1}}
\newcommand{\figsref}[1]{Figures~\ref{#1}}

\newcommand{\myindent}[1]{
\newline\makebox[#1cm]{}
}

\begin{document}

\title{Location histogram privacy by sensitive location hiding and target histogram avoidance/resemblance (extended version)}
\titlerunning{Sensitive location hiding and target histogram avoidance/resemblance}

\author{\makebox[20em][l]{Grigorios Loukides \and George Theodorakopoulos}}
\authorrunning{G. Loukides and G. Theodorakopoulos}

\institute{Grigorios Loukides \at
              Dept. of Informatics, King's College London, UK \\
              \email{grigorios.loukides@kcl.ac.uk}           
           \and
           George Theodorakopoulos \at
              School of Computer Science and Informatics, Cardiff University, UK\\
              \email{TheodorakopoulosG@cardiff.ac.uk}
}

\date{}
\maketitle
\vspace{-7mm}
\begin{abstract}
A location histogram is comprised of the number of times a user has visited locations as they move in an area of interest, and it is often obtained from the user in the context of applications such as recommendation and advertising. However, a location histogram that leaves a user's computer or device may threaten privacy when it contains visits to locations that the user does not want to disclose (\emph{sensitive locations}), or when it can be used to profile the user in a way that leads to price discrimination and unsolicited advertising (e.g. as ``wealthy'' or ``minority member''). Our work introduces two privacy notions to protect a location histogram from these threats:
\emph{sensitive location hiding}, which aims at concealing all visits to sensitive locations, and \emph{target avoidance/resemblance}, which aims at concealing the similarity/dissimilarity of the user's histogram to a \emph{target histogram} that corresponds to an undesired/desired profile. We formulate an optimization problem around each notion: Sensitive Location Hiding ($SLH$), which seeks to construct a histogram that is as similar as possible to the user's histogram but associates all visits with nonsensitive locations, and Target Avoidance/Resemblance (\TA{}/\TR{}), which seeks to construct a histogram that is as dissimilar/similar as possible to a given target histogram but remains useful for getting a good response from the application that analyzes the histogram. We develop an optimal algorithm for each notion, which operates on a notion-specific search space graph and finds a shortest or longest path in the graph that corresponds to a solution histogram. In addition, we develop a greedy heuristic for the \TA{}/\TR{} problem, which operates directly on a user's histogram. Our experiments demonstrate that all algorithms are effective at preserving the distribution of locations in a histogram and the quality of location recommendation. They also demonstrate that the heuristic produces near-optimal solutions while being orders of magnitude faster than the optimal algorithm for \TA{}/\TR{}.
\end{abstract}

\vspace{-3mm}
\section{Introduction}\label{sec:introduction}
\vspace{-2mm}

A location histogram is a statistical summary of a user's whereabouts, comprised of the number of times a user has visited each location in an area of interest.
Location histograms are often obtained from users, in the context of applications including recommendation \cite{lbsrecomaaai,lbsicdm15,geomf}, advertising \cite{lbaexpertsystems,lbaicis}, and location pattern discovery \cite{lpd}. For example, a recommender application typically employs a set of location histograms each corresponding to 
a different user (i.e., a user-location matrix) as a training set, and it aims at recommending locations that a user may be interested in visiting based on the user's histogram \cite{lbsrecomaaai}. Location histograms are also often visualized or analyzed directly \cite{sdmhistogram}.

However, a location histogram that leaves a user's computer or device may pose a threat to the user's privacy. This happens when the histogram contains visits to \emph{sensitive} locations that the user does not want to disclose, because they are associated with confidential information (e.g. a temple is associated with a religion, and the headquarters of a political organization with certain political beliefs), or when the histogram can be used to profile the user (e.g. as ``wealthy'' or ``minority member'') leading to price discrimination  \cite{mikians,pricedisc2} and unsolicited advertising \cite{bettini}. For example, if the histogram reveals that a user frequently visits expensive restaurants, a targeted-advertisement application may display to the user advertisements about products and services that are priced higher than normal \cite{mikians,pricedisc2}.

In this work, we introduce two novel notions of histogram privacy, sensitive location hiding and target avoidance/resemblance, for protecting against the disclosure of sensitive locations and user profiling, respectively. {Sensitive location hiding aims at concealing all visits to user-specified sensitive locations, by producing a \emph{sanitized} histogram, in which the frequencies associated with the sensitive locations are equal to zero. 
This protects a user from an adversary who receives the sanitized histogram, knows the set of locations considered to be sensitive, and tries to infer which of these sensitive locations were visited by the user.}
{By enforcing the notion of sensitive location hiding, users are able to disseminate their location histogram in order to benefit from location-based services, such as location recommendation, while being protected from the inference of their sensitive locations and the aforementioned consequences such inference may have.}  

Target avoidance aims at concealing the fact that the user's histogram is similar to an undesirable histogram that, if disseminated, would lead to undesired user profiling. For example, a user may wish to make their histogram dissimilar to a target histogram of a typical wealthy person, containing frequent visits to expensive restaurants, to avoid price discrimination \cite{mikians}. {As another example, a user's location histogram may allow the inference of the user's political affiliation, religious beliefs, and sexual orientation, which may lead to emotional distress, harassment or even persecution. Thus, a user would wish to avoid disseminating a histogram that is similar to histograms that can lead to such undesirable inferences. This protects from adversaries who use the sanitized histogram and the target histogram of a person with an undesirable profile, to infer that the user's histogram resembles the latter histogram.} 

Target resemblance is a variant of target avoidance, in which the user expressly \emph{wishes} to make their histogram similar to the target histogram representing a desirable profile. {For example, the desirable target histogram for a tourist could be that of a local resident in order to avoid discriminatory practices towards tourists (e.g., price discrimination). 
As another example, consider a company that engages in secret discriminatory hiring practices by preferentially hiring members of a particular demographic group. There are cases where companies have been shown to discriminate based on sexual orientation when hiring~\cite{weichselbaumer2003}. In these cases, a person who wishes to be hired will want to make their histogram resemble that of an heterosexual person, so as to avoid discriminatory treatment. 
The target histogram may be specified by the users themselves, or selected with the help of domain experts (see Section~\ref{subsec:TR_definition}). Enforcing target resemblance protects from adversaries who use the sanitized histogram and the target histogram of a person with a desirable profile, to infer that the user's histogram does not resemble the latter histogram.}

{Comparing target avoidance and target resemblance, we see that in both cases the adversary aims to infer whether or not the sanitized histogram resembles a given target histogram. The difference is that, in target avoidance, the user wants the adversary to conclude that there is no resemblance, whereas in target resemblance the user wants the opposite.}

{Our privacy notions can be achieved by histogram sanitization}, i.e., by changing the frequencies of location visits in the histogram.
However, sanitization incurs a quality (utility) loss, which must be controlled to ensure that the user obtains a good response from the application which uses their sanitized histogram.
To achieve this balance between privacy and quality, we define an optimization problem around each privacy notion:
the Sensitive Location Hiding ($SLH$) problem, which seeks to construct a sanitized histogram with minimum quality loss, and
the Target Avoidance/Resemblance (\TA{}/\TR{}) problem, which {seeks to avoid/resemble the target to a level at least equal to a user-provided privacy parameter}, while ensuring that the quality loss does not exceed a user-provided quality parameter. {If it is impossible to satisfy both the privacy and the quality requirements, then the problem has no solution.}

{Neither notion can be achieved} by existing methods for histogram sanitization. The aim of existing methods is to either (I) prevent the inference of the exact frequencies of the histogram (i.e., the number of visits to one or more locations)  \cite{acs2012differentially,xu2013differentially, kellaris, sdmhistogram,sortaki,qardaji,li2016improving}, or (II) make a user's histogram indistinguishable from a set of histograms belonging to other users \cite{fawaz2014location, xu2006utility, ghinita2007fast}. {Their aim is neither to hide sensitive locations, nor to avoid/resemble a target histogram. The privacy notions we introduce in the paper are important to achieve in real applications, as we discuss in Examples~\ref{ex1}~and~\ref{ex2} in Section~\ref{sec:overview}.}

Therefore, we develop new methods {for achieving the $SLH$ and the \TA{}/\TR{} notions}: 
(I) An optimal algorithm for $SLH$, called $LHO$ (Location Hiding Optimal). (II) An optimal algorithm for \TR{}, called $RO$ (Resemblance Optimal). (III) A greedy heuristic for \TR{}, called $RH$ (Resemblance Heuristic).
Because \TA{} and \TR{} are similar, we focus on \TR{} and discuss \TA{} briefly. 

Our methods are both effective and efficient, as demonstrated by experiments using two real datasets derived from the Foursquare location-based social network \cite{dataseturl}, which together contain approximately 3400 histograms.
In terms of effectiveness, {all algorithms achieve the corresponding notions, or announce that it is impossible to achieve them, and they are additionally able to preserve}:
(I) the distribution of locations in a histogram, which is useful in applications such as aggregate query answering and classification \cite{sdmhistogram,odessa}, and
(II) the quality of location recommendation based on Collaborative Filtering \cite{Melville}. 
In addition, the heuristic produces near-optimal solutions (up to $1.5\%$ worse than the optimal), with respect to preserving distribution similarity. 
In terms of efficiency, all algorithms scale well with the histogram parameters, requiring from less than $1$ second (the $LHO$ algorithm) to $5$ minutes (the $RO$ algorithm). In addition, the $RH$ heuristic is more efficient than the optimal algorithm by at least two orders of magnitude. 

We note that our notions are framed in the context of location histograms but can be applied to any histogram. For example, they could be applied to a histogram comprised of webpage visits. The resultant sanitized histogram would then conceal visits to webpages that a user does not want to disclose, or it would resemble/avoid a target histogram for protecting the user from targeted advertising based on their webpage visits.

\noindent {\bf Organization}~We provide an overview of {and motivation for} our approach in Section~\ref{sec:overview};
we introduce formal notation, and we formalize the privacy notions, the adversary models, and the optimization problems we solve in Section~\ref{sec:problem};
we describe our algorithms and our heuristics in Section~\ref{sec:solution}; we evaluate our approach in Section~\ref{sec:evaluation}; we discuss related work in Section~\ref{sec:related_work}; we conclude the paper in Section~\ref{sec:conclusion}.

\vspace{-5mm}
\section{Overview {and motivation} of our approach}\label{sec:overview}
\vspace{-2mm}

{This section provides examples to motivate the need for sensitive location hiding and target resemblance and also provides a high-level overview of the optimization problems and methods for solving them.} 

\vspace{-5mm}
\subsection{Sensitive Location Hiding}
\vspace{-2mm}

{Given a set of sensitive locations, a histogram satisfies the notion of sensitive location hiding when the  frequency of each of its sensitive locations is zero.  Clearly, one simple strategy to achieve this notion is by setting the frequency of each sensitive location of a given histogram to zero. However, 
this strategy may have a substantial negative impact on the quality (utility) of the histogram in location histogram applications.  
This is because it reduces the \emph{size} (sum of frequencies) of the histogram. A size reduction should be avoided because some important statistics depend on the size of the histogram. An example of such statistics is the fraction of all users' visits to a particular location in a city (i.e., the ratio between the sum of the frequency of the location over all users' histograms and the sum of the sizes of these histograms), which is a simple indicator of the popularity of the location. Another example is the average number of visits to a location (i.e., the ratio between the size of the user's histogram and the number of locations in the histogram), which is used in location recommendation \cite{Melville,bre}.}

{A different strategy that achieves the sensitive location hiding notion, while preserving the size of the  histogram is to redistribute the frequency counts of the sensitive locations to non-sensitive ones. However, the redistribution needs to be performed in a way that  preserves the quality (utility) of the histogram in location histogram applications. The impact of each possible redistribution on quality must be quantified, and the selected redistribution strategy must be the one with the lower impact. We quantify the impact of a redistribution strategy with a quality distance function, similarly to most works on histogram sanitization~\cite{acs2012differentially, sdmhistogram, li2016improving,xu2013differentially}. This function offers generality, because different functions can be chosen for different applications.}

{The above discussion motivates the formulation of the \emph{Sensitive Location Hiding} ($SLH$) optimization problem: 
Given a histogram $H$, a set of sensitive locations, and a quality distance function, produce a sanitized histogram $H'$ such that the frequency of each sensitive location in $H'$ is $0$, $H'$ is as similar as possible to $H$, and $H'$ has the same size as $H$. Similarity is measured with the quality distance function.}

{In Section~\ref{subsec:SLH_definition} we give a formal definition of the $SLH$ optimization problem, discuss the adversary model it provides protection against}, and show that the problem is \emph{weakly} NP-hard~\cite{papadimitrioubook}. In addition, we discuss a variation of the problem which relaxes the size requirement and can be easily dealt with by our algorithms.

To illustrate the $SLH$ notion and the $SLH$ problem, we now provide Example~\ref{ex1}, which is inspired from approaches on privacy-preserving recommendation \cite{shenrecomicdm,tse18}.
However, the $SLH$ notion and problem are not tied to recommendation and cannot be handled with existing approaches.

\begin{figure*}[htbp]\hspace{-1mm}
	\begin{subfigure}[b]{.19\textwidth}\centering
		\includegraphics[scale=0.3,clip]{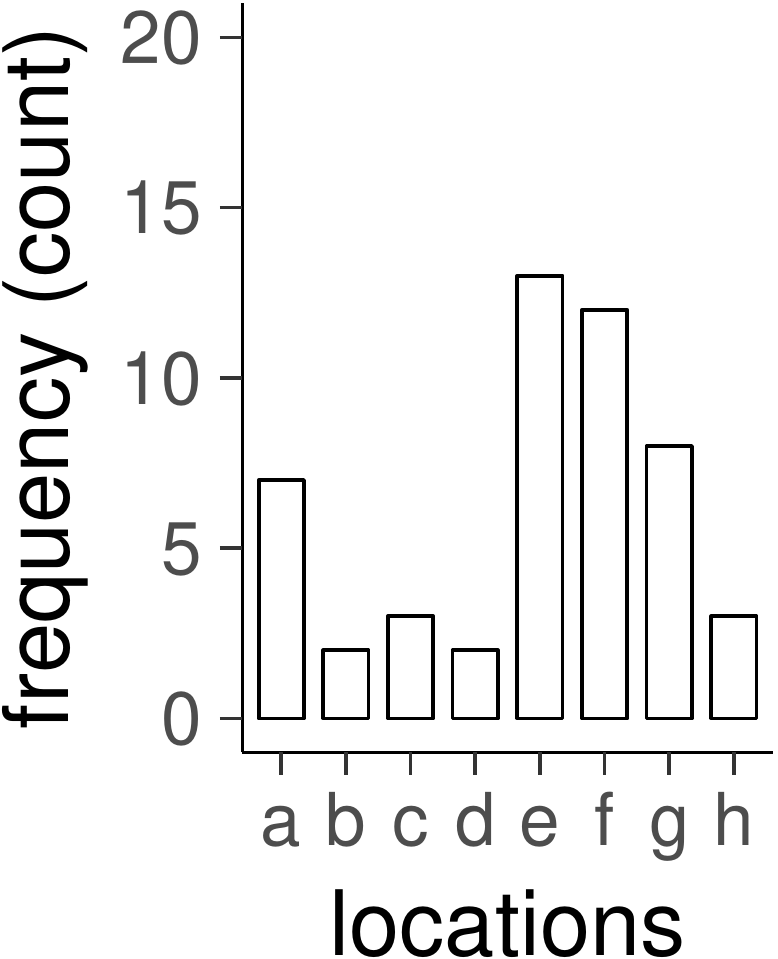}
		\caption{$H$}\label{ex1a}
	\end{subfigure}\hspace{+2mm}
	\begin{subfigure}[b]{.18\textwidth}\centering
		\includegraphics[scale=0.3,clip]{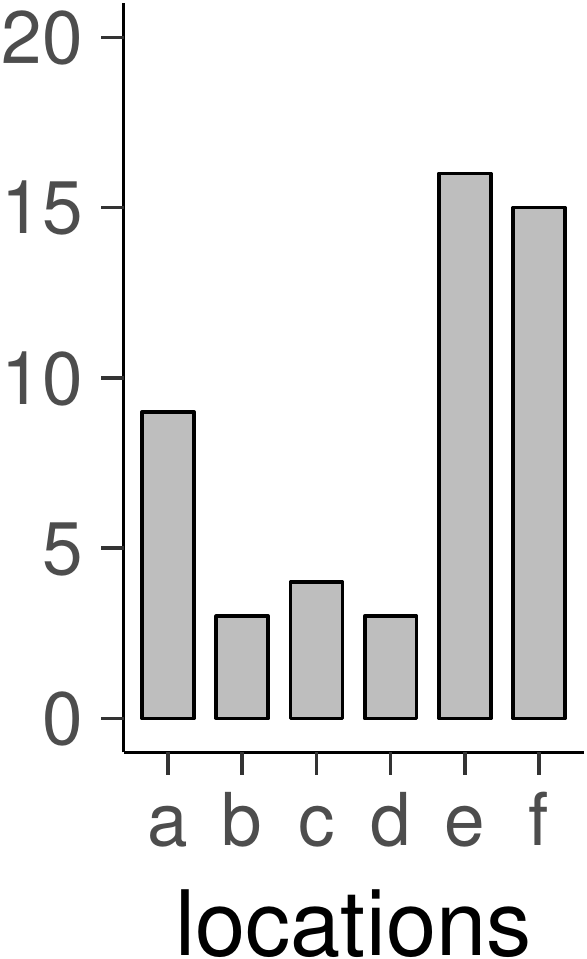}
		\caption{$H'$ by $LHO$}\label{ex1b}
	\end{subfigure}\hspace{+1mm}
	\begin{subfigure}[b]{.18\textwidth}\centering
		\includegraphics[scale=0.3,clip]{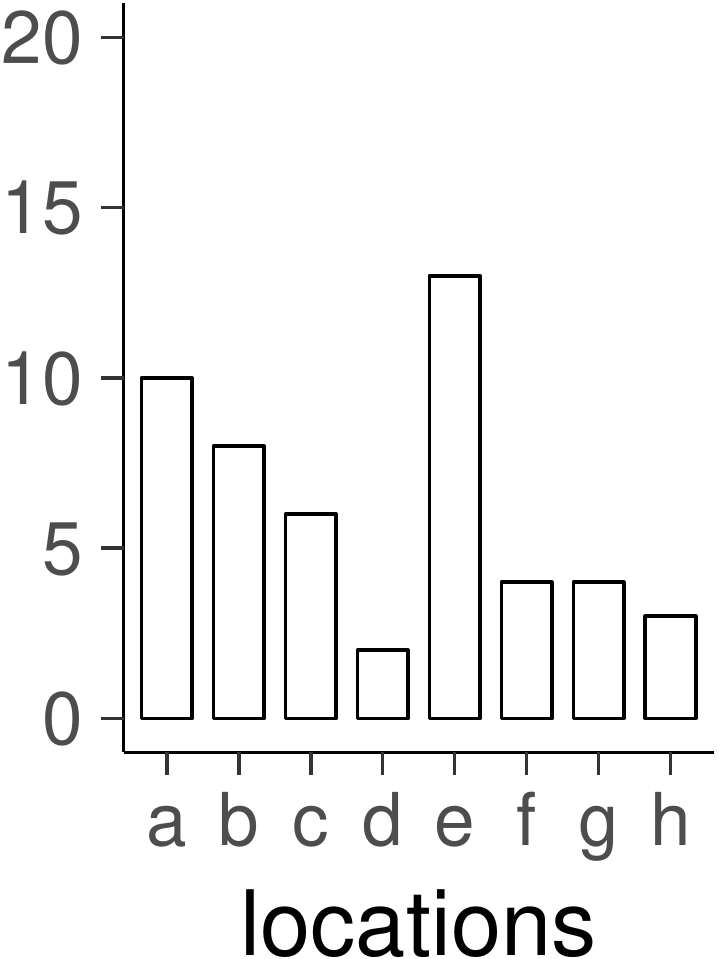}
		\caption{$H''$ in \TR{}}\label{ex1c}
	\end{subfigure}\hspace{+1mm}
	\begin{subfigure}[b]{.18\textwidth}\centering
		\includegraphics[scale=0.3,clip]{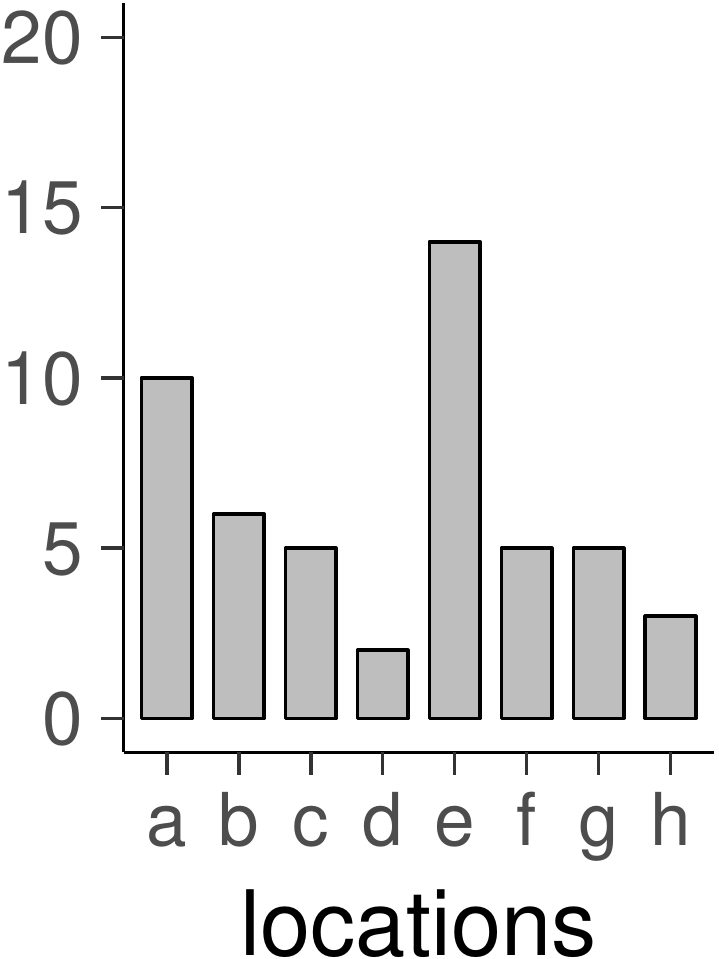}
		\caption{$H'_{RO}$ by $RO$}\label{ex1d}
	\end{subfigure}
	\hspace{+1mm}
	\begin{subfigure}[b]{.19\textwidth}\centering
		\includegraphics[scale=0.3,clip]{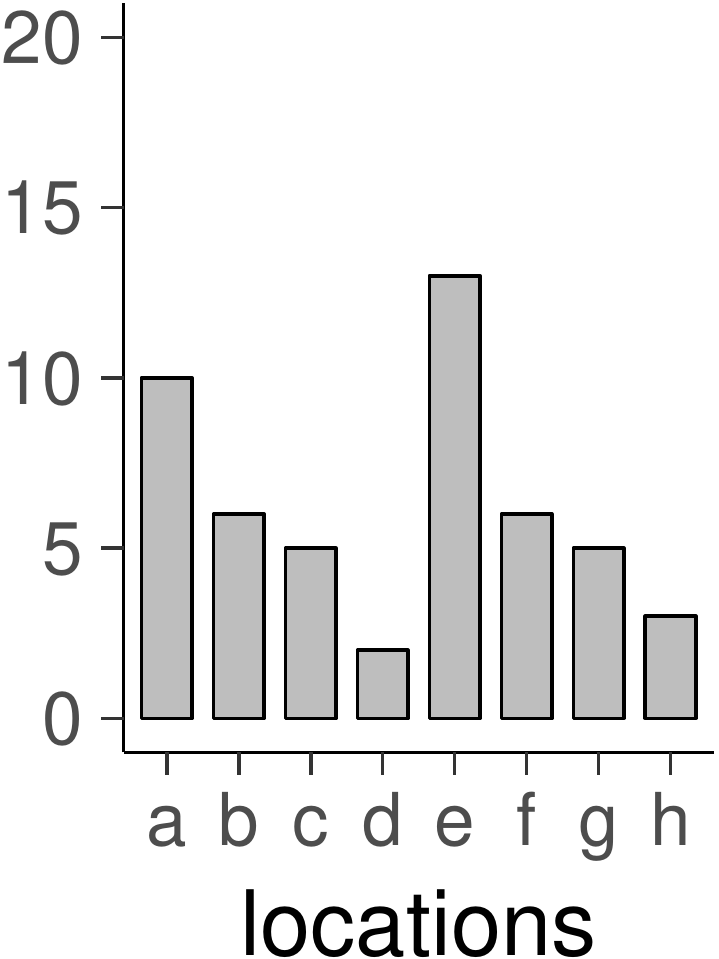}
		\caption{$H'_{RH}$ by $RH$}\label{ex1e}
	\end{subfigure}
	\caption{(a) Location histogram. (b) Sanitized histogram produced by the $LHO$ algorithm when $g$ and $h$ are sensitive locations and the quality distance function is Jensen-Shannon divergence \cite{jsdivergence} (only positive counts are illustrated). (c) Target histogram used in the \TR{} problem. (d) Sanitized histogram produced by the $RO$ algorithm when applied with the target histogram in \figref{ex1c}. (e) Sanitized histogram produced by the $RH$ heuristic when applied with the target histogram in \figref{ex1c}}
\end{figure*}

\vspace{-1mm}
\begin{example}[Illustration of the $SLH$ notion and $SLH$ problem] 
	An application provides location recommendations to users by analyzing their location profiles. To obtain a recommended location, a user
	must send $50$ location visits to the application in the form of a location histogram. To compute the recommended location, the application uses common mining tasks, such as discovering frequent location patterns in the user's histogram and finding similar histograms to it
	\cite{zhang2014}. The location histogram $H$ of a user Alice is shown in \figref{ex1a}.
	The histogram contains the number of times Alice visited each of the locations $a$ to $h$. Alice is not willing to provide $H$ to the application, because the last two locations in $H$, $g$ and $h$, are sensitive, but she still wishes to receive a ``good'' recommended location by the application.
	Therefore, Alice solves the $SLH$ problem and obtains the sanitized histogram $H'$ shown in \figref{ex1b}. The sanitized histogram  preserves privacy, because it does not contain the sensitive locations. It can be sent to the application to receive a fairly accurate recommendation, because it contains $50$ visits to nonsensitive locations (the visits to sensitive locations are zero and not shown) and is as ``similar'' as possible to $H$, to the extent permitted by the privacy requirement. \qedsymbol
	\label{ex1}
\end{example}

\begin{figure}[htbp]\centering
	\includegraphics[width=9.5cm,keepaspectratio]{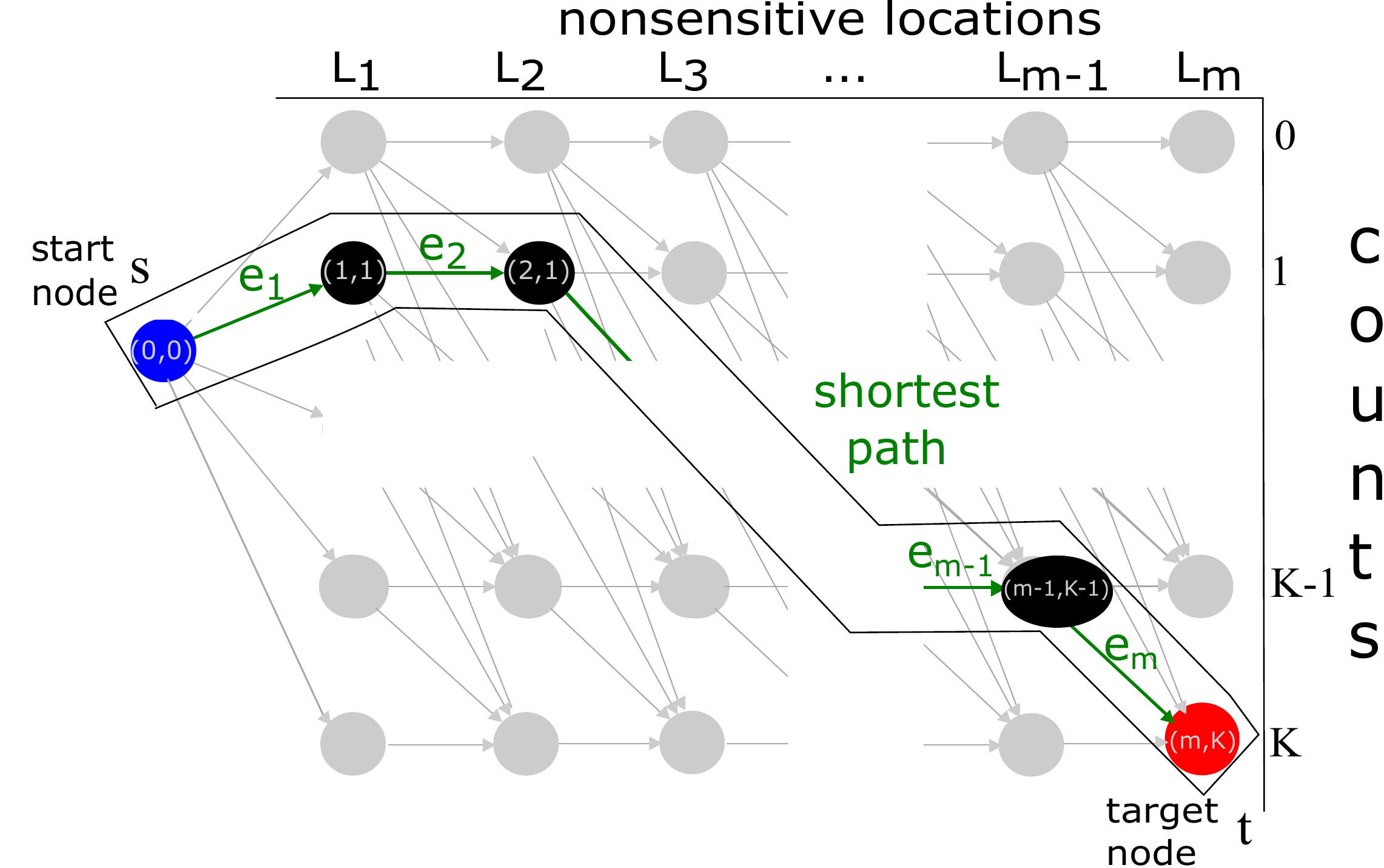}
	\caption{Graph constructed by the Location Hiding Optimal ($LHO$) algorithm.
	The shortest path from $s$ to $t$ corresponds to the optimal solution of the $SLH$ problem.
	Each node $(i,j)$, $i\in[1,m]$, $j\in[0,K]$, in the path denotes the redistribution of $j$ sensitive location visits into the nonsensitive bins $1, \ldots, i$.
	The path corresponds to the optimal way of redistributing  all $K$ sensitive location visits into all $m$ nonsensitive bins.
	The weight of the edge $((i,j),(i+1,j+k))$ denotes the impact on quality caused by redistributing $k$ sensitive location visits into the nonsensitive bin $i+1$, and the sum of the edge weights of this path $e_1 + \cdots + e_m$ quantifies the quality distance between the optimal solution and $H$}\label{LHOabstractgraph}
\end{figure}

To optimally solve the $SLH$ problem, the $LHO$ algorithm finds the exact number of sensitive location visits that need to be redistributed into each \emph{nonsensitive bin} (bin corresponding to a nonsensitive location), so that all sensitive location visits
are redistributed and quality is optimally preserved, with respect to the quality distance function. That is, the algorithm determines the frequency of each nonsensitive location of the sanitized histogram $H'$, so that $H'$ has the same size with the given histogram $H$ and is as similar as possible to it, with respect to the quality distance function.  However, it is computationally prohibitive to directly compute the quality of each possible redistribution of the sensitive location visits into the nonsensitive bins and then select the optimal solution.
This follows from the fact that there are $O\left(\binom{K+m-1}{m-1}\right)$ ways to redistribute $K$ sensitive location visits into $m$ nonsensitive bins
(each way corresponds to a \emph{weak composition} of $K$ \cite{bona2011walk}).
Therefore, $LHO$ solves the problem by modeling it as a shortest path problem between two specific nodes, $s$ and $t$, of a directed acyclic graph (DAG) (see \figref{LHOabstractgraph}).
The node $s$ is labeled $(0,0)$, and each other node is labeled $(i,j)$, where $i\in[1,m]$ corresponds to a nonsensitive location $L_i$ and $j\in[0,K]$ corresponds to the number of sensitive location visits that will be redistributed into the nonsensitive bins $1, \ldots, i$ of the sanitized histogram $H'$.
For example, the label $(m,K)$ of the node $t$ denotes the redistribution of all $K$ sensitive location visits to all $m$ nonsensitive bins of $H'$.
We may refer to a node using its label.
The graph contains an edge from each node $(i,j)$ to each node $(i+1,j+k)$ with $k\in[0, K-j]$, where $k$ denotes the number of sensitive location visits that are redistributed into the nonsensitive bin $i+1$.
For example, the edge $((i,j), (i+1,j+k)) = ((1,1), (2,1))$ denotes that $k=0$ visits are redistributed into the nonsensitive bin $i+1=2$.
Each edge $((i,j), (i+1,j+k))$ has a weight that quantifies the impact on quality caused by the redistribution of $k$ sensitive location visits into the nonsensitive bin $i+1$.
Every path from $s$ to $t$ corresponds to a feasible solution of the $SLH$ problem.
This is because the nodes in the path uniquely determine how all $K$ sensitive location visits will be redistributed into all $m$ nonsensitive bins of $H'$ (see property (I) in Section~\ref{LHOalgorithmsectin}).
In addition, the length (sum of edge weights) of the path is equal to the quality distance between the corresponding solution $H'$ and $H$ (see property (II) in Section~\ref{LHOalgorithmsectin}).
Thus, the shortest path from $s$ to $t$ corresponds to a histogram $H'$ that is as similar as possible to $H$ and therefore it is the optimal solution of the $SLH$ optimization problem.
For example, applying the $LHO$ algorithm  to the histogram of \figref{ex1a}, when the locations $g$ and $h$ are sensitive and the quality distance function is Jensen-Shannon divergence (see Section~\ref{sec:problem:ql}), produces the sanitized histogram in \figref{ex1b}.
Note that the visits to $g$ and $h$ are redistributed into all nonsensitive bins, so that the sanitized histogram is as similar as possible to the histogram of \figref{ex1a}. {A formal description of the $LHO$ algorithm, as well as the analysis of the algorithm is provided in Section~\ref{LHOalgorithmsectin}.} 

\vspace{-6mm}
\subsection{Target Resemblance}
\vspace{-2mm}
{Given a target histogram, a histogram satisfies the notion of target resemblance when it is 
similar enough to the target. A privacy distance function quantifies 
similarity, and a privacy parameter quantifies the threshold for determining whether the two histograms are similar enough.} 

{Clearly, any histogram can be easily modified to be arbitrarily similar 
to a given target histogram, by simply redistributing all 
its frequency counts so that they are exactly equal to the counts in the target histogram. However, 
as in the case of $SLH$, 
a simplistic redistribution can deteriorate quality unacceptably. The modification to the histogram must balance between resemblance to the target histogram and similarity to the original histogram.  A quality distance function quantifies the quality loss caused by the modification, and a quality parameter quantifies the threshold for determining whether the loss is acceptable or not.}

{The above discussion motivates the formulation of the \emph{Target Resemblance} (\TR{}) optimization problem: 
Given a histogram $H$, a target histogram $H''$, a quality distance function and a quality parameter $\epsilon$, a privacy distance function and a privacy parameter $c$, produce a sanitized histogram $H'$ such that its quality distance from $H$ is at most $\epsilon$, its privacy distance from $H''$ is minimized, and its size is the same as $H$. If the resulting privacy distance of $H'$ from $H''$ is larger than $c$, then there is no solution.}


{In Section~\ref{subsec:TR_definition} we give a formal definition of the \TR{} problem, discuss  the
adversary  model  it  provides  protection  against,} and we show that it is \emph{weakly} NP-hard. In addition, we discuss a variation that relaxes the size requirement and can be easily dealt with by our algorithms. To illustrate the \TR{} privacy notion and optimization problem, we provide Example~\ref{ex2}.

\vspace{-2mm}
\begin{example}[Illustration of the $TR$ notion and problem, continuing from Example~\ref{ex1}] \figref{ex1a} shows the location histogram $H$ of a user, Bob, who wants to use the location recommendation application.
	Bob is not willing to provide $H$ to the application, because he is concerned about price discrimination, as a result of frequent visits to locations $f$ (``airport'') and $g$ (``5-star hotel'').
	To achieve his purpose, Bob can solve the Target Resemblance (\TR{}) problem to generate a histogram that resembles the target histogram $H''$ in \figref{ex1c}.
	$H''$ reflects a budget-conscious person, because in $H''$ the frequencies of locations $a$ (``train station''), $b$ (``$2$-star hotel''), and $c$ (``$3$-star hotel'') are relatively high, whereas the frequencies of $f$ and $g$ are relatively low.
	Hence, $H''$ is likely to attract lower-priced recommendations than $H$ would, and it is definitely more likely to prevent price discrimination \cite{mikians,pricedisc2}.
	The resemblance to $H''$ is satisfied by generating a \emph{sanitized} histogram $H'_{RO}$ ($RO$ for ``Resemblance Optimal'') that minimizes a \emph{privacy distance function} between the sanitized histogram and $H''$.
	In parallel, Bob still wishes to receive  a ``good'' recommended location by the application.
	This quality requirement is satisfied by limiting the dissimilarity between $H$ and the sanitized histogram $H'_{RO}$ to a maximum of $\epsilon=0.05$, as measured by a \emph{quality distance function}, so that the sanitization preserves the similarity between $H$ and other users' histograms, which helps compute a ``good'' recommended location \cite{Melville}.
	After solving the \TR{} problem, Bob obtains the sanitized histogram $H'_{RO}$ in \figref{ex1d}, which is almost identical to the target $H''$. \qedsymbol
	\label{ex2}
\end{example}

To optimally solve the \TR{} problem, the Resemblance Optimal ($RO$) algorithm finds the exact number of location visits that need to be added into, or removed from, each bin of a histogram $H$, so that the resultant sanitized histogram $H'$ is as similar as possible to the target histogram $H''$, and no more dissimilar from $H$ than what is allowed by the quality threshold $\epsilon$.
Again, the large number of potential solutions, given by $O\left(\binom{N+n-1}{n-1}\right)$, where $N$ is the size of $H$ and $n$ is its length, prohibits directly computing the quality of each possible solution and selecting the optimal solution.
Therefore, $RO$ solves the problem by modeling it as a constrained shortest path problem in a DAG (see \figref{ROabstractgraph}).
The graph contains a path $(u^0_0,u^{N_1}_1,\ldots,u^{N_n}_n)$ for each allocation of $N = N_n$ counts to the $n$ bins of the histogram (i.e., each allocation corresponds to a possible solution to the Target Resemblance problem, ignoring the quality constraint), where a node $u_i^{N_i}$  corresponds to allocating $N_i$ counts to bins 1 up to and including $i$.
The length of a path is equal to the dissimilarity of the corresponding allocation to the target histogram $H''$, whereas the cost of the path is equal to the quality loss as compared to the user's histogram $H$.
The algorithm finds the shortest path among those whose cost does not exceed the quality threshold $\epsilon$.
As the graph is a DAG, to find the optimal solution it suffices to explore it in Breadth-First Search order.
First, we compute constrained shortest paths to all nodes that correspond to bin 1: $u_1^{N_1}, N_1 = 0, \ldots, N$; then, we extend these paths to all nodes that correspond to bin 2: $u_2^{N_2}, N_2 = 0, \ldots, N$ and we prune them if they violate the quality constraint; we continue all the way to $u_{n-1}^{N_{n-1}}, N_{n-1} = 0, \ldots, N$ and finally to the node $u_n^{N_n}, N_n = N$.
The shortest path to that final node corresponds to the optimal valid allocation of $N$ counts to bins $1, \ldots, n$.

\begin{figure}[!ht]\centering
	\includegraphics[scale=0.32]{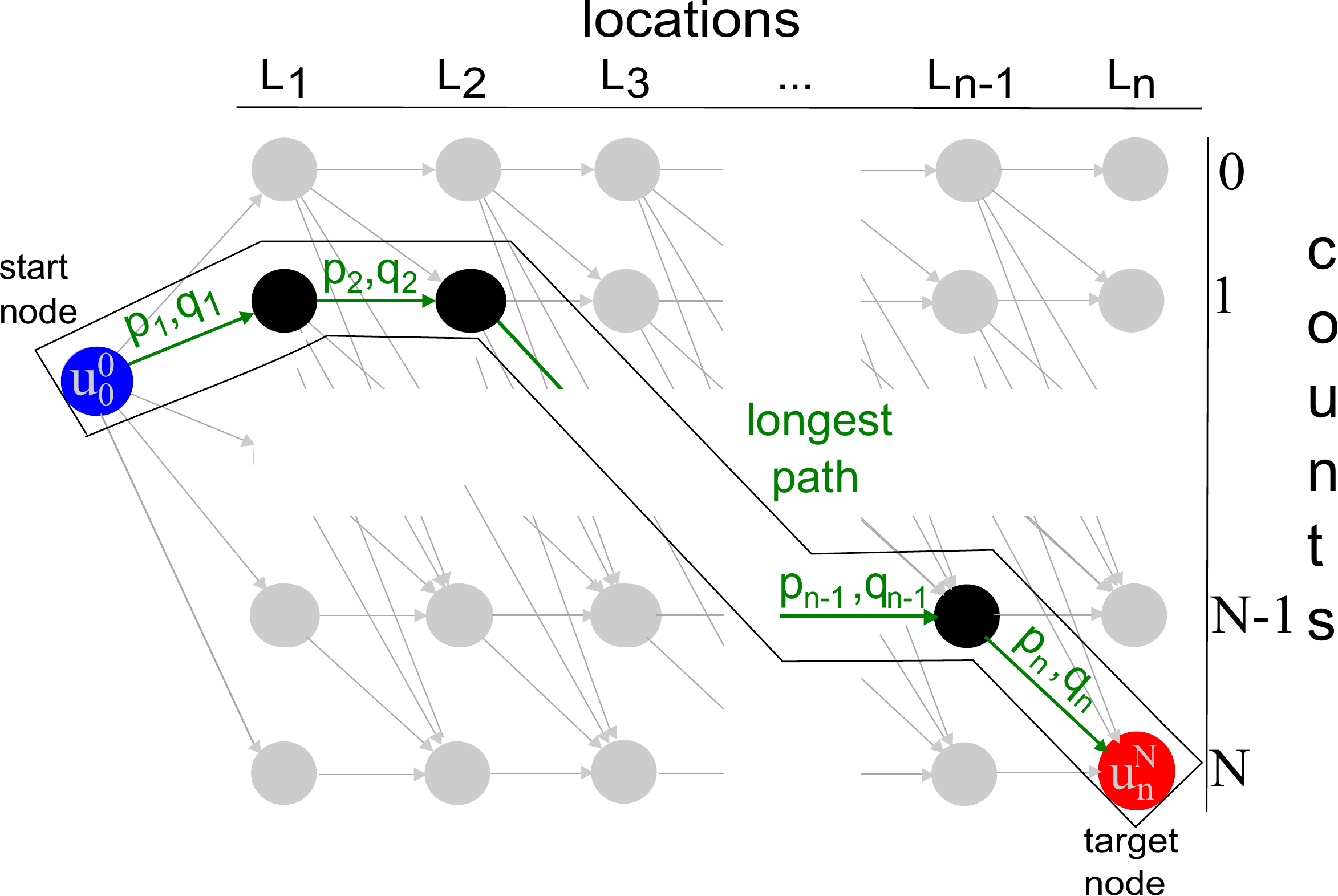}
	\caption{Graph constructed by the Resemblance Optimal ($RO$) algorithm.
		The shortest path from $u^0_0$ to $u^N_n$ with cost at most $\epsilon$ corresponds to the optimal solution of the \TR{} problem.
		The nodes of this path correspond to the optimal way of allocating counts to all bins.
		The edge weights of this path are (privacy, quality loss) pairs $(p_1,q_1), \ldots, (p_n,q_n)$. The two weights of an edge $u_i^{N_i}$ to $u_{i+1}^{N_i + k}$ are the privacy and quality effects of allocating exactly $k$ counts to bin $i+1$ of the solution histogram.
		The sum $\sum_{i\in[1,n]}p_i$ quantifies the dissimilarity between the optimal solution and the target histogram (smaller is better) and the sum $\sum_{i\in[1,n]}q_i$ quantifies the total quality loss, which should be at most $\epsilon$}\label{ROabstractgraph}
\end{figure}

When solution optimality is not necessary, the \TR{} problem can be solved more efficiently by the $RH$ heuristic. $RH$ differs from the $RO$ algorithm in that it restricts the set of bins in the histogram $H$ whose number of location visits can increase or decrease. Specifically, it works in a greedy fashion, iteratively ``moving'' frequency counts from \emph{source bins} to \emph{destination bins}. The source bins have higher frequency in $H$ than in the target histogram $H''$, whereas the destination bins have lower frequency in $H$ than in $H''$. Thus, moving counts from source to destination bins makes the sanitized histogram more and more similar to the target histogram, but it incurs a quality loss due to changes in frequency counts. Therefore, to control the loss of quality, moves are performed for as long as the quality distance of the resultant sanitized histogram from $H$ does not exceed the quality threshold. Example~\ref{ex3} below illustrates the $RO$ algorithm and the $RH$ heuristic.

\begin{example}[Illustration of the $RO$ algorithm and $RH$ heuristic, continuing from Example \ref{ex2}]
	Bob applies $RO$ to his histogram $H$ in \figref{ex1a}, using the target histogram $H''$ in \figref{ex1c}, the quality threshold $\epsilon=0.05$, and the Jensen-Shannon divergence to measure dissimilarity from $H''$ and from $H$.
	The algorithm produces the sanitized histogram $H'_{RO}$ in \figref{ex1d}, which is as similar to $H''$ as allowed by the specified threshold.
	Similarly, Bob applies $RH$ and obtains the sanitized histogram $H'_{RH}$ in \figref{ex1e}.
	Comparing $H'_{RO}$ and $H'_{RH}$ to $H''$, we observe that $H'_{RO}$ is very similar to $H''$, while $H'_{RH}$ is slightly less similar (e.g. the frequencies of $f$ and $g$ are equal in $H''$ and $H'_{RO}$, while they are not equal in $H''$ and $H'_{RH}$).
	However, $H'_{RH}$ is still useful for getting a good recommendation from the application, because the quality loss (dissimilarity to $H$) does not exceed $\epsilon$. \qedsymbol
	\label{ex3}
\end{example}

\subsection{Target Avoidance}
Given a target histogram, a histogram satisfies the notion of target avoidance when it is \emph{dissimilar} enough from the target. A privacy distance function quantifies similarity, and a privacy parameter quantifies the threshold for determining whether the two histograms are dissimilar enough.

Similarly to $SLH$ and $TR$, any modification to the histogram must balance between dissimilarity to the target histogram, so as to achieve target avoidance, and similarity to the original histogram, so as to preserve quality.

For illustration and motivation, we refer again to Example~\ref{ex2}. In that example, Bob can alternatively solve the Target Avoidance problem to directly avoid price discrimination.
In so doing, the aim would be to avoid, not resemble, the target histogram, and so it would need to have as many visits to each of the locations $f$ and $g$ as $H$. For example, $H''$ could be identical to $H$.

The Target Avoidance optimization problem is formally defined in Section~\ref{subsec:TA_definition}, in which we also describe precisely the adversary model.

The optimal algorithm $AO$ and heuristic $AH$ for solving the \TA{} optimization problem are very similar to the ones described for \TR{} above. The only essential difference is that \TA{} algorithms compute a longest path in the graph instead of a shortest path as in \TR{}.

\section{Background, problem definitions, and adversary models}\label{sec:problem}
\vspace{-2mm}

In this section, we define some preliminary concepts and then we formally define the $SLH$, \TA{}, and \TR{} optimization problems.
A summary of the most important notation we introduce is in Table~\ref{tbl:notation}. 

\begin{table}[th]\small
	\caption{Acronyms and Notation}
	\centering
	\begin{tabular}{c l}
		\hline
		\textbf{Acronym} & \textbf{Meaning} \\
		\hline
		$SLH$ & Sensitive Location Hiding (input: $H, L', d_q()$, output: $H'$)\\
		$LHO$ & Location Hiding Optimal (optimal algorithm for SLH)\\
		$TR$ & Target Resemblance (input: $H, H'', d_p(), d_q(), \epsilon$, output: $H'$)\\
		$RO$ & Resemblance Optimal (optimal algorithm for TR)\\
		$RH$ & Resemblance Heuristic (heuristic for TR)\\
		$TA$ & Target Avoidance (input: $H, H'', d_p(), d_q(), \epsilon$, output: $H'$)\\
		$AO$ & Avoidance Optimal (optimal algorithm for TA)\\
		$AH$ & Avoidance Heuristic (heuristic for TA)\\
		$JS$ & Jensen-Shannon divergence \\
		~    & (distance function used in the evaluation)\\
		\hline
		\textbf{Notation} & \textbf{Meaning} \\
		\hline
		$L = \{L_1, \ldots, L_{|L|}\}$ & Set of semantic locations that a user can visit\\
		$L' \subseteq L$ & Set of sensitive locations (for SLH)\\
		$\mathcal{H}_{n,N}$ & Set of all histograms of length $n, n \leq |L|,$ and size $N$\\
		$H = (f(L_1),\ldots, f(L_n))$ & Histogram of nonnegative frequencies (counts)\\
		~                             & of visits to locations $L_1, \ldots, L_n \in L, n \leq |L|$\\
		$H'$ & Sanitized histogram\\
		$H''$ & Target histogram (for TA and TR)\\
		$\epsilon$ & Maximum quality loss threshold (for TA and TR)\\
		$d_p(H', H'')$ & Privacy distance function (for TA and TR)\\
		$d_q(H, H')$ & Quality distance function (measures quality loss)\\[2pt]
		\hline
	\end{tabular}
	\label{tbl:notation}
\end{table}

\vspace{-3mm}
\subsection{Preliminaries}
\vspace{-2mm}

We consider an area of interest, modeled as a finite set of semantic locations $L=\{L_1, \ldots, L_{|L|}\}$ of cardinality $|L|$, where a location $L_i$, $i \in [1,|L|]$, is e.g. ``Italian Restaurant,'' ``Cinema,'' or ``Museum.''

We also consider a user who moves in this area. The user's \emph{histogram} is a vector of integer frequencies $H=(f(L_1), \ldots, f(L_n))$, where $n\leq |L|$ is the \emph{length} of the histogram. Each location $L_i, i\in[1,n]$, has a frequency $f(L_i)>0$, when $L_i$ was visited by the user, or $f(L_i)=0$ otherwise. We may refer to frequencies as counts.

We use $H[i]$ to refer to the $i$-th element, or bin, of $H$, and $N$, or \emph{size}, to refer to the $L_1$-norm $|H|_1=\sum_{i\in[1,n]}H[i]$ of $H$. We use $\mathcal{H}_{n,N}$ to denote the set of all histograms of length $n$ and size $N$.

Having compiled $H$, the user wishes to submit it to a location-based application. Before submitting it, the user transforms it into a \emph{sanitized} histogram $H'$ (in a way to be made concrete in Problems \ref{SLHdef} and \ref{TRdef} below) and then submits $H'$ to the application. Next, the application returns a response to the user. Depending on the sanitization required, $H'$ may contain zero frequency counts for some locations, or it may contain nonzero frequency counts for locations that the user never visited. If the user wishes, we can easily guarantee that $H'$ will not contain nonzero frequency counts for locations that the user never visited, by assigning an infinite cost $d_q(H[i],H'[i])$ for each such location $L_i$. 

\vspace{-2mm}
\subsubsection{Quality loss}\label{sec:problem:ql}
\vspace{-2mm}

Since the user submits $H'$, which is in general different from $H$, there will be a negative impact on the quality of the application response. The resulting loss in quality is measured by a \emph{quality distance} function $d_q(H,H')$. For every pair $H, H'$, we require that $d_q(H, H') \geq 0$, and that $H = H'$ implies $d_q(H, H') = 0$.
In addition, we require $d_q$ to decompose as a sum over bins, i.e. there must be a function $q$ such that $d_q(H,H') = \sum_{i \in [1,n]} q(H[i],H'[i])$. Most distances used in data mining applications in which distances between histograms/vectors must be preserved (e.g., Jensen-Shannon divergence ($JS$-divergence) \cite{jsdivergence}, Jeffrey's divergence \cite{qardaji},  $L_2$-distance (Euclidean distance) and Squared Euclidean distance \cite{dmbook}, Variational distance \cite{jsdivergence}, Pearson $\chi^2$ distance \cite{jsdivergence},  and Neyman $\chi^2$ distance \cite{jsdivergence})  decompose as a sum over bins. 

We use $JS$-divergence as the objective function $d_q$ in our experiments (see Section \ref{sec:evaluation}). $JS$-divergence is a standard measure for quantifying distances between probability distributions, which is often used in histogram/vector classification \cite{classifdivergence} and clustering \cite{clustdivergence}. Given two histograms $H_1, H_2$, the JS-divergence between them is defined as
\begin{equation*}
\resizebox{.95\hsize}{!}{
	$JS(H_1,H_2)=\frac{1}{2\cdot N}\cdot \displaystyle\sum_{i\in[1,n]}\left(H_1[i]\cdot \log_2\left(\frac{2\cdot H_1[i]}{H_1[i]+H_2[i]}\right) + H_2[i]\cdot \log_2\left(\frac{2\cdot H_2[i]}{H_1[i]+H_2[i]}\right)\right)$}.
\end{equation*}

\vspace{-2mm}

\noindent with the convention $0\cdot log_2(0)=0$. $JS$-divergence is bounded in $[0,1]$ \cite{jsdivergence}, and $JS(H_1,H_2)=0$ implies no quality loss. As explained in \cite{li2009modeling}, $JS$-divergence can also be easily extended to capture semantic similarity requirements (e.g. An Italian Restaurant is more similar to a French Restaurant than to an American Cinema), when this is needed in applications. The extended measure, called \emph{smoothed} $JS$-divergence, requires preprocessing the histogram by kernel smoothing and then applying $JS$-divergence to the preprocessed histogram. Incorporating smoothed $JS$-divergence in our methods is straightforward and left for future work.





\vspace{-3mm}
\subsection{The Sensitive Location Hiding problem: adversary model and formal definition} \label{subsec:SLH_definition}
\vspace{-2mm}

{As discussed in the Introduction and in Section~\ref{sec:overview}, the Sensitive Location Hiding ($SLH$) privacy notion aims to conceal all visits to sensitive locations.
We formulate the adversary model and the desired privacy property for the $SLH$ notion as follows.}

{The adversary knows: (I) the sanitized histogram $H'$ that the user submits, (II) the set of all possible sensitive locations $L'$, and (III) the fact that, if $H'$ is fake, then it must have been produced by the $LHO$ algorithm in our paper. The adversary has no other background knowledge. 
The adversary succeeds if, based on their knowledge, they manage to determine whether or not the user visited one or more of the sensitive locations in $L'$.} 

{The desired privacy property is the negation of the adversary's success criterion. That is, the adversary must not be able to infer, from the sanitized histogram, that the user has visited any of the sensitive locations.}  

We formally define the corresponding optimization problem as follows:
\begin{problem}[Sensitive Location Hiding (\SLH{})]\label{SLHdef}
	Given a histogram $H \in \mathcal{H}_{n,N}$, a subset $L'\subseteq L$ of \emph{sensitive} locations, and a quality distance function $d_q()$, construct a sanitized histogram $H' \in \mathcal{H}_{n,N}$ that
	\begin{equation*}
	\begin{aligned}
	& \underset{H' \in \mathcal{H}_{n,N}}{\text{minimizes}}
	& & d_q(H, H') \\
	& \text{subject to}
	& & H'[i] = 0 \text{ for each location }L_i\in L'.
	\end{aligned}
	\end{equation*}
\end{problem}

Intuitively, the \SLH{} problem requires constructing a sanitized histogram by redistributing the counts of the sensitive locations of $H$ into bins that correspond to nonsensitive locations, in the best possible way according to $d_q$. The sensitive locations are specified by the user based on their preferences. 

In the \SLH{} problem formulation, we follow the user-centric (or personalized) approach to privacy that is employed in \cite{Abul2018,damiani2010probe,shokri2017privacy,augir2016privacy}.
This approach requires the users to specify their own privacy preferences, so that these preferences are best reflected in the produced solutions. However, not all users may possess knowledge allowing them to identify certain locations in their histograms as sensitive. Yet, such users often know that a class of locations are sensitive, or they do not want to be associated with a class of locations \cite{Xiao:2006:PPP:1142473.1142500, Loukides2013}. 
For instance, several users may not want to be associated with visits to any type of clinic or adult entertainment location. In this case, users may employ a taxonomy\footnote{A taxonomy (also known as hierarchy) is a tree structure, in which the root corresponds to the most general location ``any'', there is a leaf node for each sensitive location, and the internal nodes of the taxonomy correspond to aggregate/coarse sensitive locations.} to identify classes of sensitive locations, which requires less detailed knowledge. This method is inspired by   \cite{Xiao:2006:PPP:1142473.1142500, Loukides2013} 
and simply requires a user to select one or more nodes in the taxonomy. If a node $u$ that is not a leaf is selected, then all locations corresponding to leaves in the subtree rooted at $u$ will be considered as sensitive. 
If the selected node $u$ is a leaf, then its corresponding location will be considered as sensitive. Such taxonomies already exist for location-based data, and they can also be automatically constructed based on machine learning techniques \cite{velardi}. For example, in the Foursquare taxonomy (see Section \ref{sec:evaluation}), there is an aggregate category (internal node) ``Medical center'' which contains more specific categories (leaves)  ``Hospital,'' ``Rehab center,'' etc. 

\begin{thm}\label{thm:SLH_NP_Hard}
	The \SLH{} problem is weakly NP-hard.
\end{thm}
\begin{proof}
\vspace{-2mm}
	See Appendix (Section~\ref{subsec:SLH_NP_hard}).
\end{proof}
\vspace{-2mm}

Clearly, the $SLH$ problem seeks to produce a sanitized histogram $H'$ with the same size as $H$. As discussed in the Introduction, this allows preserving
statistics that depend on the size of the histogram, which are important in location based applications, such as location recommendation. However, it is also possible
to require the sanitized histogram $H'$ to have a given size instead (e.g., when an application requires a histogram to have a certain number of location counts, or in pathological cases where redistribution leads to undesirable/implausible histograms). This leads to a variation of the $SLH$ problem, referred to as $SLH_{r}$, which requires redistributing $r\geq 0$ counts of sensitive locations into the bins corresponding to nonsensitive locations. Note the following choices for $r$ in $SLH_{r}$: (I) For $r=0$, the $SLH_{r}$ problem requires constructing a sanitized histogram where each sensitive location has count $0$ and each nonsensitive location has a count equal to that of its count in $H$. Such a histogram is trivial to produce, by simply replacing the count of each sensitive location with $0$. (II) For $r=\sum_{L_i\in L'}f(L_i)$ (i.e., equal to the total count of sensitive locations), $SLH_{r}$ becomes equivalent to the $SLH$ problem. (III) For
$r>\sum_{L_i\in L'}f(L_i)$, the $SLH_{r}$ problem requires constructing a sanitized histogram with larger size than $H$. As we will explain in Section~\ref{LHOalgorithmsectin}, it is straightforward to optimally solve $SLH_{r}$ based on our $LHO$ algorithm.

{\subsubsection{Solutions to the $SLH$ optimization problem satisfy the desired privacy property}
The adversary cannot distinguish between a user A who has only visited nonsensitive locations and thus submits a non-sanitized histogram $H_A$, and a user B who has visited some sensitive locations and the algorithm has produced a sanitized histogram $H_B'$ that is identical to $H_A$. This is because every possible sanitized histogram that the $LHO$ algorithm can output is a valid histogram that could have legitimately been produced by a user. Note that, if there are histograms that cannot be produced by a legitimate user, $LHO$ can be trivially adapted to never output such histograms. This adaptation is easy because all histograms are encoded as paths in a graph, so illegitimate histograms are also 
paths in the graph, referred to as illegitimate paths, and these histograms 
can be avoided by simply changing the shortest-path finding algorithm to an algorithm that finds a shortest path which is not contained in a given subset of illegitimate paths~\cite{pugliese}.} 
\vspace{-2mm}
\subsection{The Target Resemblance problem: adversary model and formal definition} \label{subsec:TR_definition}
\vspace{-2mm}
{As discussed in the Introduction and in Section~\ref{sec:overview}, for the Target Resemblance ({\TR}) privacy notion the user specifies a target histogram $H''$ to resemble, a quality parameter $\epsilon$ and a privacy parameter $c$. The objective of the {\TR} optimization problem is to create a sanitized histogram $H'$ that is as similar as possible to $H''$, subject to the quality constraint $d_q(H,H')\leq \epsilon$.}
{The privacy distance function that quantifies the notion of similarity is denoted $d_p(H',H'')$. If $d_p(H',H'') > c$, then $H'$ is not acceptable, because it is not similar enough to the target.} 

The function $d_p(H',H'')$ is nonnegative and it must decompose as a sum over bins, i.e. there must be a function $p$ such that $d_p(H', H'') = \sum_{i \in [1, n]} p(H'[i],H''[i])$, using zeros to fill in missing location counts. In \TR{}, privacy is maximum when $H'=H''$ ($d_p(H',H'') = 0$), because there is no better resemblance than being identical. Any function with these properties would be suitable as $d_p$ (e.g., $JS$-divergence, or $L_2$-distance).
We use $JS$-divergence as $d_p$ in our experiments (see Section~\ref{sec:evaluation}).

{We can formulate the adversary model and the desired privacy property for this problem as follows:
The adversary knows (I) the histogram $H'$ that the user submits, (II) a target histogram $H''$, 
(III) a privacy distance function $d_p()$, and (IV) a privacy parameter $c$.} 

{Upon receiving $H'$, 
the adversary compares it to the target $H''$ in order to profile the user. For example, if an adversary wants to determine whether the user is a member of a particular ethnic/religious/social group, the target histogram is the histogram of a typical member of that group.} 
{Formally, the adversary makes this determination by comparing $d_p(H', H'')$ to $c$, i.e., by comparing the privacy distance between the user's submitted histogram $H'$ and $H''$ to the privacy parameter $c$. If $d_p(H', H'') \leq c$, the adversary concludes that the user is a member of the group, otherwise they conclude that the user is not a member of the group. The adversary has no other background knowledge. In particular, the adversary does not know whether the user submitted their true histogram or the user submitted a modified histogram aiming to resemble a particular target histogram.} 
{The adversary succeeds if they conclude that the user is not a member of the group, i.e.  $d_p(H', H'') > c$.}

{The desired privacy property is the negation of the adversary's success criterion. In $TR$, the desired privacy property is $d_p(H', H'') \leq c$.} 

{We formally define the corresponding optimization problem as follows:}
\begin{problem}[Target Resemblance ({\TR})]\label{TRdef}
	Given two histograms $H, H'' \in \mathcal{H}_{n,N}$, a privacy distance function $d_p()$, {a privacy parameter $c$}, a quality distance function $d_q()$, and maximum quality loss threshold $\epsilon \geq 0$, construct a sanitized histogram $H' \in \mathcal{H}_{n,N}$ that
	\begin{equation*}
	\begin{aligned}
	& \underset{H' \in \mathcal{H}_{n,N}}{\text{minimizes}}
	& & d_p(H', H'') \\
	& \text{subject to}
	& & d_q(H, H') \leq \epsilon.
	\end{aligned}
	\end{equation*}
	{If the resulting $H'$ is such that $d_p(H', H'') > c$, then it is impossible to achieve both the desired privacy property and the desired quality constraint.}
\end{problem}
\vspace{+3mm}

Intuitively, the {\TR} problem requires constructing a sanitized histogram $H'$ of the same length and size with $H$ and $H''$ that offers the best possible privacy by being as similar as possible to the target histogram $H''$ according to $d_p$, while incurring a quality loss at most $\epsilon$ according to $d_q$.  

The function $d_q$ is selected by the location-based application provider (recipient of the sanitized histogram) and is provided to the user together with an intuitive explanation of what different values of $d_q()$ mean for quality. For example, $d_q() \geq 0.8$ means ``very low quality'', $0.6 \leq d_q() \leq 0.8$ means ``low quality'' etc., where ``quality'' refers to the quality of the application response (e.g. recommendation) that the user receives. Then, in the spirit of user-centric (or personalized) privacy \cite{shokri2012protecting,lingliu}, the user uses the above explanation by the provider to choose a value of $\epsilon$ that corresponds to his/her tolerance for quality loss.

The problem requires the user to specify the target histogram $H''$. However, some users may not possess sufficient knowledge to perform this task, even though they want to resemble a person with certain characteristics (e.g., a wealthy person).  
In these cases, $H''$ can be constructed as follows. The user chooses a target probability distribution $h''$ from a repository of probability distributions that are constructed by domain experts and labeled accordingly 
(e.g., a distribution corresponding to a ``wealthy'' profile, a ``tourist'' profile, a  ``healthy person'' profile \cite{hasan2015location,vissers2014crying}), in the same way that experts compile e.g., adblock filters
(lists of URLs to block) or lists of virus signatures for antivirus software. To choose one of these profiles, the user looks for a label that they want to resemble. This setup is very similar to other papers in the literature \cite{damiani2010probe, Abul2018}. 

Note that the distribution $h''$ may be defined on a different set of locations from the user's histogram $H$, in which case both are expanded to cover all locations in either $h''$ or $H$, with zero values for the new locations. Then, each entry $h''[i]$ is multiplied by the size $N$ of the user's histogram $H$ to create the target histogram $H''[i]$ (see Section~\ref{TRalgsection}). So, in effect, $H''$ is the expected histogram by a hypothetical user that picks $N$ locations from the distribution $h''$.
By the above construction, $H''$ and $H$ are of the same length $n$ and size $N$, but note that $H''$ may not have integer counts, because $H''[i] = N \cdot h''[i]$ is not necessarily an integer. Strictly speaking, this violates the requirement of histograms to have integer counts, but that is not a problem for our methods, because the privacy distance functions do not need integer arguments. However, we do require the histogram $H'$ that the algorithms output to have integer counts.

\vspace{-1mm}
\begin{thm}\label{thm:TR_NP_Hard}
	The {\TR} problem is weakly NP-hard.
\end{thm}
\vspace{-1mm}

\begin{proof}
\vspace{-2mm}
See Appendix (Section \ref{subsec:TR_NP_hard}).
\end{proof}
\vspace{-2mm}

Clearly, the {\TR} problem requires constructing a sanitized histogram $H'$ with the same size as $H$ and $H''$. That is, it assumes that the desirable target histogram $H''$ has the same counts as $H$, but these counts are distributed differently from $H$. However, it is also possible to relax this assumption. This leads to a variation of the {\TR} problem, referred to as {\TR}$_{|H''|_1}$, which instead requires the sanitized histogram $H'$ to only have the same size as $H''$, while it can be different from the size of $H$. 
It is straightforward to optimally (resp., heuristically) solve {\TR}$_{|H''|_1}$ based on our $RO$ algorithm (resp., based on our $RH$ heuristic) (see Section~\ref{TRalgsection}).

\vspace{-3mm}
{\subsubsection{Solutions to the $TR$ optimization problem satisfy the desired privacy property}} 

The $TR$ problem tries to minimize $d_p(H', H'')$, while satisfying the quality constraint $d_q(H, H') \leq \epsilon$.
Of course, a particular choice of $\epsilon$ affects privacy. If $\epsilon$ is low, an algorithm for the {\TR} problem may output an $H'$ that is the same or very similar to $H$, because all histograms that satisfy the specified quality constraint are close to $H$. Then, the user has to decide whether this $H'$ is safe to release. 

{Given the privacy parameter $c$, it is not safe to release $H'$ when $d_p(H',H'') > c$. If $d_p(H',H'') > c$, the user will decide not to release any histogram at all. Alternatively, the user may want to re-run the algorithm with a larger $\epsilon$, i.e. to sacrifice more quality in order to achieve the privacy requirement.}

The user's decision may depend on the intuitive meaning of the function used for $d_p$. For example, if $d_p$ is Pearson $\chi^2$ and the target $H''$ models a ``wealthy'' user, then $d_p(H',H'')$ quantifies how much more likely it is that $H'$ has been produced by a user who follows the ``wealthy'' profile compared to any other profile\footnote{We refer the interested reader to statistics textbooks for more details\cite{leboudec2010performance}}. Thus, if this likelihood ratio  {exceeds $c$}, then the user may not want to release that $H'$. 

It is also trivial to exclude solutions with {$d_p(H',H'')>c$} by modifying our methods to disregard such solutions and terminate if no solution exists. In conclusion, the user either submits a histogram that satisfies the privacy property, or nothing at all.


\vspace{+2mm}
\subsection{The Target Avoidance problem} \label{subsec:TA_definition}
\vspace{-2mm}

As mentioned above, Target Avoidance ({\TA}) is a variant of the Target Resemblance ({\TR}) problem, which we briefly discuss below. 
\begin{problem}[Target Avoidance ({\TA})]\label{TAdef}
	Given two histograms $H, H'' \in \mathcal{H}_{n,N}$, a privacy distance function $d_p()$, {a privacy parameter $c$}, a quality distance function $d_q()$, and maximum quality loss threshold $\epsilon \geq 0$, construct a sanitized histogram $H' \in \mathcal{H}_{n,N}$ that
	\begin{equation*}
	\begin{aligned}
	& \underset{H' \in \mathcal{H}_{n,N}}{\text{maximizes}}
	& & d_p(H', H'') \\
	& \text{subject to}
	& & d_q(H, H') \leq \epsilon.
	\end{aligned}
	\end{equation*}
	{If the resulting $H'$ is such that $d_p(H', H'') < c$, then it is impossible to achieve both the desired privacy property and the desired quality constraint.}
\end{problem}
\vspace{+2mm}

Intuitively, the {\TA} problem requires constructing a sanitized histogram $H'$ of the same length and size with $H$ and $H''$. The sanitized histogram must offer the best possible privacy by being as dissimilar as possible to the target histogram $H''$ according to $d_p$, while incurring a quality loss at most $\epsilon$ according to $d_q$. The threshold $\epsilon$ and target histogram $H''$ are specified by the user based on their preferences. For example, the user can set $H''$ to $H$, in order to avoid $H$ itself, or to a part of $H$ that contains the locations that characterize an undesirable profile (e.g., frequent visits to airports) or are frequented by a certain ethnic minority (which may help infer that an individual belongs to the minority). The user could also choose $H''$ with the help of domain experts, as in the $TR$ problem.

{In terms of an adversary model, the adversary has the same knowledge as in {\TR} and they  succeed if $d_p(H', H'') < c$. If the algorithm does not produce an $H''$ such that $d_p(H', H'') \geq c$, then the user can either not submit any histogram at all, or the user may want to re-run the algorithm with a larger $\epsilon$. The proof that the {\TA} problem leads to a solution satisfying the desired privacy property is similar to that for {\TR} (omitted).}

\begin{thm}\label{thm:TA_NP_Hard}
	The {\TA} problem is weakly NP-hard.
\end{thm}
\begin{proof}
	See Appendix (Section \ref{subsec:TA_NP_hard})
\end{proof}

The {\TA} problem is very similar to the {\TR} problem. This is established through a reduction from {\TA} to {\TR} that is given in Appendix (Section~\ref{subsec:TA_TR_reduction}).
There is also a variation of {\TA}, referred to as {\TA}$_{|H''|_1}$, which requires the sanitized histogram $H'$ to have the same size as $H''$, but not necessarily as $H$. Again, our methods can easily deal with this variation.

\vspace{-3mm}

\section{Algorithms}\label{sec:solution}

Since the {\SLH}, {\TR} and {\TA} problems are weakly NP-hard, it is possible to design
pseudopolynomial algorithms\footnote{Pseudopolynomial algorithms run in polynomial time in the numerical value of the input \cite{papadimitrioubook}.}
to optimally solve them. We present optimal algorithms based on shortest/longest path problems for the {\SLH}, {\TA}, and {\TR} problems. In addition, we present heuristic algorithms for the {\TR} and {\TA} problem.
The heuristics find solutions of comparable quality to those of the optimal algorithms but are more efficient by two orders of magnitude.
Furthermore, we explain how our methods can deal with the variation $SLH_{r}$ and $TA/TR_{|H''|_1}$ of the $SLH$ and $TA/TR$ problem, respectively.

\subsection{$LHO$: An optimal algorithm for {\SLH}}\label{LHOalgorithmsectin}

This section presents Location Hiding Optimal ($LHO$), which optimally solves the {\SLH} problem.  Before presenting $LHO$, as motivation, we consider a simple algorithm which distributes the counts of the  sensitive location(s) to the non-sensitive bin(s) proportionally to the counts of the non-sensitive bins. Thus, it aims to construct an $H'$ by initializing it to $H$ and then increasing the count of each non-sensitive bin $H'[i]$ by $x[i]=H[i]\cdot \frac{\sum_{i\in L'}H[i]}{\sum_{i \in L\setminus L'}H[i]}=H[i]\cdot \frac{K}{N-K}$, while assigning $0$ to each sensitive bin. While intuitive, this algorithm fails to construct an $H'$, for a given histogram $H$ and distance function $d_q$, when $x[i]$ is not an integer, and also it may lead to solutions with large $d_q(H,H')$ (i.e., low data utility), as it does not take into account the input distance function $d_q$.

We now discuss the $LHO$ algorithm. Without loss of generality, we assume that the nonsensitive locations correspond to the first $n-|L'|$ bins of the original histogram
$H=(f(L_1),\ldots,f(L_n))$, while the remaining $|L'|$ bins correspond to the sensitive locations.
The total count of sensitive locations in $H$ is $K=\sum_{L_i\in L'}f(L_i)$. $LHO$ must move (redistribute) these counts into the nonsensitive bins, while minimizing the quality error $d_q()$.

The $LHO$ algorithm is founded on the following observation: The optimal way of redistributing counts to each nonsensitive bin of $H'$ corresponds to a shortest path between two specific nodes of a \emph{search space graph} $G_{LHO}(V,E)$, where $V$ and $E$ is the set of nodes and set of edges of $G_{LHO}$, respectively.
In the following, we discuss the construction of $G_{LHO}$ and the correspondence between this shortest path and the solution to the {\SLH} problem.
Then, we discuss the $LHO$ algorithm.

\begin{figure}[ht!]\centering
	\includegraphics[scale=0.35]{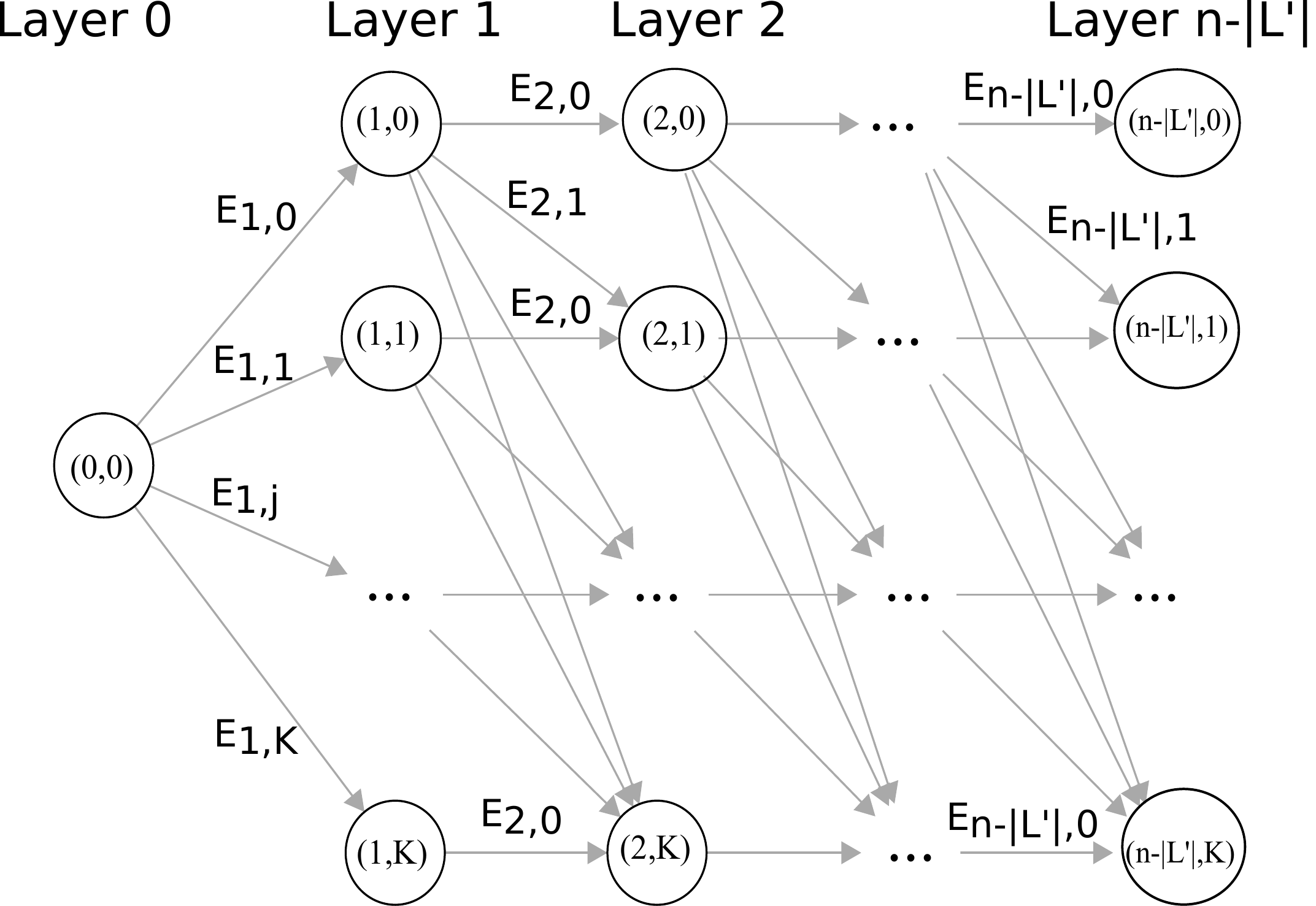}
	\caption{Search space graph $G_{LHO}$ for the Sensitive Location Hiding problem. Layer $0$ contains the node $(0,0)$ and each of the layers $1$ to $n-|L'|$ contains $K+1$ nodes.
		Each edge connects nodes of consecutive layers and has a weight equal to the
		error $E_{i+1,k}$, where $i+1$ is the layer of the end node of the edge and $k$ is the count of sensitive locations.
		$E_{i+1, k}$ represents the impact of redistributing (i.e., adding) $k$ counts into the $i+1$ bin of the sanitized
		histogram $H'$, which is initialized to the original histogram $H$. That is, $E_{i+1,k}=q(H[i+1], H'[i+1]+k)$.
		The missing nodes and edges are denoted with ``$\ldots$''}\label{LHOsearchgraph}
\end{figure}

$G_{LHO}$ is a multipartite directed acyclic graph (DAG) (see \figref{LHOsearchgraph}) such that:
\begin{itemize}
	\item 
	It contains  $n-|L'|+1$ layers of nodes.
	Layer $0$ comprises a single node, and layers $1, ..., n-|L'|$ comprise $K+1$ nodes each. Each layer $1, ..., n-|L'|$ corresponds to a nonsensitive bin.
	\item The single node in layer $0$ is labeled $(0,0)$, and each node in layer $i\in [1,n-|L'|]$ is labeled $(i,j)$, where $j$ denotes the redistribution (i.e., addition) of $j$ counts to bins $1$ up to and including $i$ of the sanitized histogram. We may refer to nodes of $G_{LHO}$ using their labels.
	\item There is an edge $((i,j), (i+1, j+k))$ from node $(i,j)$ to node $(i+1, j+k)$, for each $i\in [0, n-|L'|-1]$, where $k\geq 0$, $j+k\leq K$.
	That is, each node labeled $(i,j)$ is connected to every node in the following layer $i+1$ that corresponds to a count of at least $j$.
	\item Each edge $((i,j), (i+1, j+k))$
	is associated with a weight equal to the error
	$E_{i+1, k} =q(H[i+1], H[i+1] + k)$. The error $E_{i+1, k}$ quantifies the impact on quality that is incurred by redistributing (i.e., adding) $k$ counts into bin $i+1$.
\end{itemize}

Let $P$ be a path comprised of nodes
$(0,0)$, $(1,k_1)$, $\ldots$, $(n-|L'|,k_{n-|L'|})$ of $G_{LHO}$. The properties below easily follow from the construction of $G_{LHO}$:
\begin{itemize}
	\item[(I)] The path $P$ corresponds to an addition of $k_{i}-k_{i-1}$ counts to the $i$-th bin of the histogram, for each $i\in[1, n-|L'|]$, where $k_0=0$.
	\item[(II)] The \emph{length} of $P$ is equal to the total weight $E_{1, k_1}+\ldots, E_{n-|L'|,k_{n-|L'|}}$ of the edges in $P$. This total weight is the total quality loss incurred by the allocation corresponding to $P$.
\end{itemize}
Thus, the path $P$ corresponds to a sanitized histogram $H'$ whose first $n-|L'|$ bins have counts $H[i] + (k_i - k_{i-1}), i = 1, \ldots, n-|L'|$,
and $d_q(H, H')$ is equal to $\sum_{i=1}^{n-|L'|}q(H[i], H[i] + k_i - k_{i-1})$.
For example, the path comprised of nodes $(0,0), (1,K), \ldots, (n-|L'|,K)$ in \figref{LHOsearchgraph} corresponds to a sanitized histogram $H'$ in which all $K$ sensitive counts have been moved to the first bin. The quality loss in this case is just $d_q(H, H') = q(H[1], H[1]+K)$, as all other bins have the same counts in both $H$ and $H'$.

Conversely, each possible allocation of the $K$ sensitive counts into nonsensitive bins corresponds to a path between the nodes $(0,0)$ and $(n-|L'|, K)$ of $G_{LHO}$, which represents a feasible solution to the {\SLH} problem.
Therefore, the shortest path between the nodes $(0,0)$ and $(n-|L'|, K))$ of $G_{LHO}$ (i.e., the path with the minimum length $E_{1, k_1}+\ldots, E_{n-|L'|,K}$; ties are broken arbitrarily)
represents a sanitized histogram $H'=(H[1]+(k_1-k_0), \ldots, H[n-|L'|]+(K-k_{n-|L'|-1}), 0, ..., 0)$, which is the optimal solution to {\SLH}.
This is because $H'$ has minimum $d_q(H,H')$, the same size with $H$, and a zero count for each sensitive location.

We now present the pseudocode of the $LHO$ algorithm. In step \ref{LHO1}, the algorithm constructs the search space graph $G_{LHO}$. In step \ref{LHO2},
the algorithm finds a shortest path between the nodes $(0,0)$ and $(n-|L'|,K)$ of $G_{LHO}$. In step \ref{LHO3}, the sanitized histogram $H'$
corresponding to the shortest path (i.e., the optimal solution to the {\SLH} problem) is created and, last, in step \ref{LHO4}, $H'$ is  returned.
\hspace{-1mm}
\begin{algorithm}\small
	\Algorithm{LHO (Location Hiding Optimal)}
	\KwIn{Histogram $H$, set of sensitive locations $L'$, quality distance function $d_q$}
	\KwOut{Sanitized histogram $H'$}
	Construct the search space graph $G_{LHO}$ \label{LHO1}\\
	$((0,0),(1,k_1),\ldots,(n-|L'|,K))\gets$ the shortest path from node $(0,0)$ to node \myindent{4.3} $(n-|L'|,K)$ in $G_{LHO}$\label{LHO2}\\
	$H'\gets (H[1]+k_1, H[2]+(k_2-k_1), \ldots, H[n-|L'|]+(K-k_{n-|L'|-1}),0, \ldots, 0)$\label{LHO3}\\
	\Return $H'$ \label{LHO4}\\
\end{algorithm}

\begin{figure}[ht!]
	\includegraphics[scale=0.34]{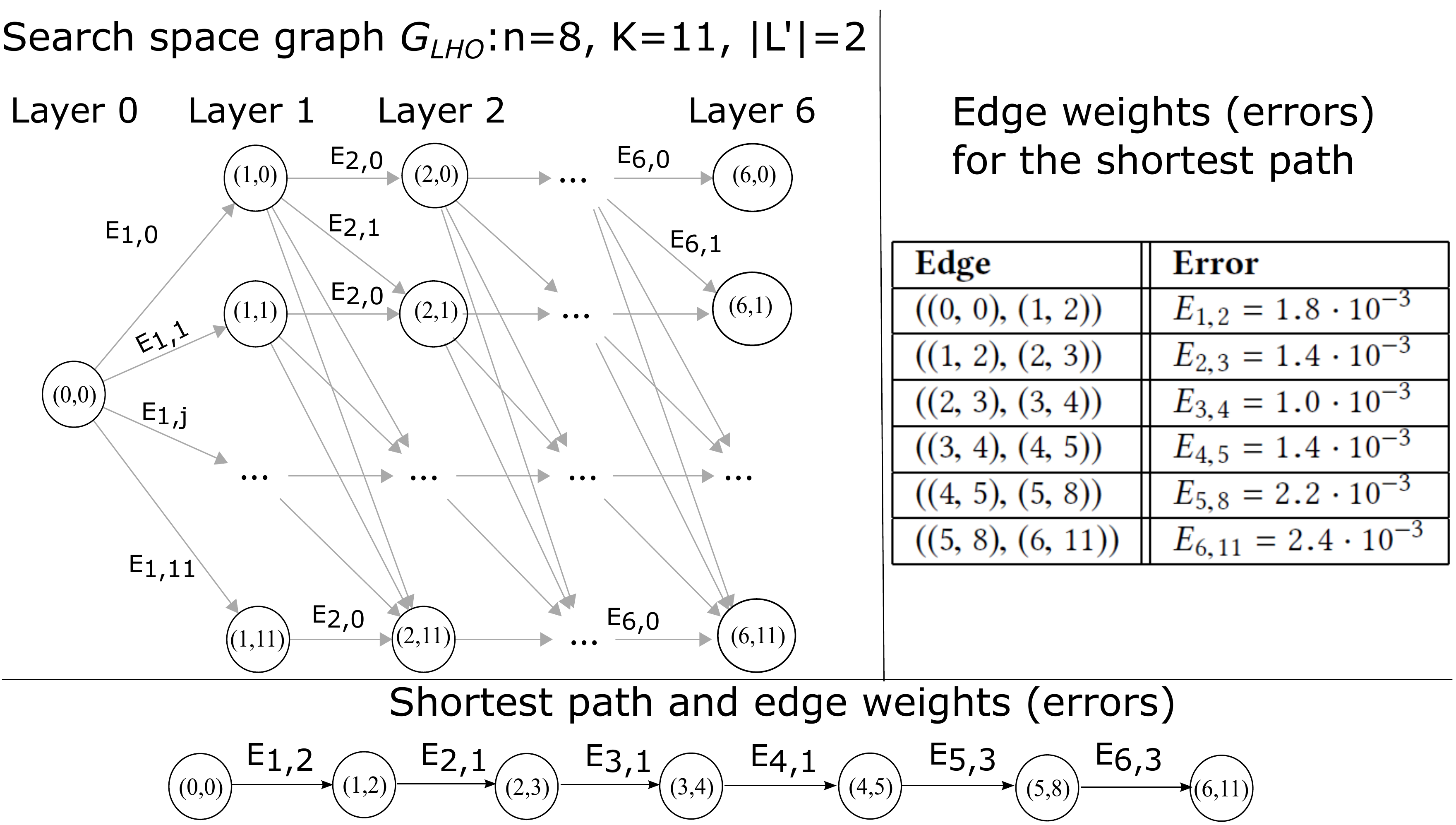}
	\caption{Search space graph $G_{LHO}$ for Example \ref{RunExampleLHO} (the missing nodes and edges are denoted with ``$\ldots$''), and shortest path along with its corresponding  weights}\label{RunExampleLHOfig}
\end{figure}
\begin{example}
	$LHO$ is applied to the histogram $H=(7,2,3,2,13,12,8,3)$ in Figure \ref{ex1a}. The set of sensitive locations $L'$ contains the locations $g$ and $h$ with counts $8$ and $3$, respectively, and the quality distance function $d_q$ is $JS$-divergence. In step \ref{LHO1}, the algorithm constructs the search space graph in Figure \ref{RunExampleLHOfig}. The graph has $n-|L'|+1=7$ layers of nodes, where $n=8$ is the length of $H$ and $|L'|=2$ is the number of sensitive locations. Layer $0$ contains the node $(0,0)$ and each other layer contains $K+1=12$ nodes, where $K=11$ is the total count of sensitive locations in $H$. Each node in layers $1, \ldots, 6$ is labeled $(i,j)$; $i$ denotes the layer of the node and  $j$ denotes the counts of sensitive locations that are redistributed into bins $1, \ldots, i$. For example, the node $(6,11)$ denotes that all $11$ counts of the sensitive locations are added into bins $1,\ldots, 6$.
	In addition, there is an edge with weight $E_{i+1,k}$ between each node $(i,j)$ and every node $(i+1, j+k)$, for each $k\in [0, K-j]$. The weight $E_{i+1,k}$ quantifies the increase to $JS$-divergence incurred by redistributing (i.e., adding) $k$ counts of sensitive locations into bin $i+1$.
	For example, the node $(0,0)$ is connected to the nodes $(1,0), \ldots, (1,11)$, and the edge $((0,0), (1,2))$ has weight $E_{1,2}\approx 1.8\cdot 10^{-3}$, because adding $2$ counts into the first bin increases $JS$-divergence by approximately $1.8\cdot 10^{-3}$. In step \ref{LHO2}, $LHO$ finds the shortest path from the node $(0,0)$ to the node $(6,11)$, shown in Figure \ref{RunExampleLHOfig},
	and in step \ref{LHO3} it constructs the sanitized histogram $H'=(9,3,4,3,16,15,0,0)$ (see Figure \ref{ex1b}) that corresponds to the shortest path. Note that $j$ in the label $(i,j)$ of each node in the shortest path corresponds to the counts of sensitive locations that are added into bins $1, \ldots, i$ in $H'$. Last, in step \ref{LHO4}, $H'$ is returned. \label{RunExampleLHO} \qedsymbol
\end{example}

The time complexity of the $LHO$ algorithm is $O\left((n-|L'|)\cdot K^2+S\right)$, where 
$(n-|L'|)\cdot K^2$ is the cost of constructing $G_{LHO}$ (step \ref{LHO1}) and
$S$ is the cost of finding the shortest path (step \ref{LHO2}).
Constructing $G_{LHO}$ takes $O\left((n-|L'|) \cdot K^2\right)$ time, because $G_{LHO}$ contains $O\left((n-|L'|)\cdot K\right)$ nodes and $O\left(K+1+(n-|L'|-1)\cdot \binom{K}{2}\right)$\newline$=
O\left((n-|L|)\cdot K^2\right)$ edges, and the computation of
each edge weight $E_{i+1,k}$ takes $O(1)$ time, because it is computed by accessing a single pair of bins from $H$ and $H'$.
The cost $S$ is determined by the shortest path algorithm. For example, it is
$O\left((n-|L|)\cdot K^2\cdot log((n-|L'|)\cdot K^2)\right)$ for Dijkstra's algorithm with binary heap \cite{sedgewickbook}.

Last, we note how the variation $SLH_r$ of the $SLH$ problem (see Section \ref{subsec:SLH_definition}) can be optimally solved with $LHO$.
This is possible by simply setting the parameter $K$ in $LHO$ to $r$ (i.e., constructing a search space graph whose layers $1, \ldots, n-|L'|$ comprise $r+1$ nodes each, and then finding a shortest path from $(0,0)$ to $(n-|L'|,r)$ and the histogram $H'$ corresponding to the path).

\subsection{Optimal algorithms for Target Resemblance and Target Avoidance}\label{TRalgsection}
Although the two problems of Target Resemblance ({\TR}) and Target Avoidance ({\TA}) are different in their privacy motivation, they are mathematically very similar (see Sections~\ref{subsec:TR_definition}~and~\ref{subsec:TA_definition}). In this section, we model and solve the {\TR}
problem as a constrained shortest path problem on a specially constructed search space graph $G_{TR}$. It follows immediately that the {\TA}
problem can be seen as a \emph{longest} path problem on the same graph. Because the graph is a directed acyclic graph (DAG), computing longest and shortest paths has the same complexity \cite{sedgewickbook}: By visiting the graph nodes in Breadth-First Search order, we can simply keep track of the shortest (or longest) path to each node. We can even solve the two problems in one pass. To keep the presentation simple, we focus on the {\TR} problem, which we solve optimally by the Resemblance Optimal ($RO$) algorithm.

In the following, we discuss the construction of $G_{TR}$ and then provide the pseudocode of $RO$.

\begin{figure}[!ht]
	\centering
	\includegraphics[scale=0.37]{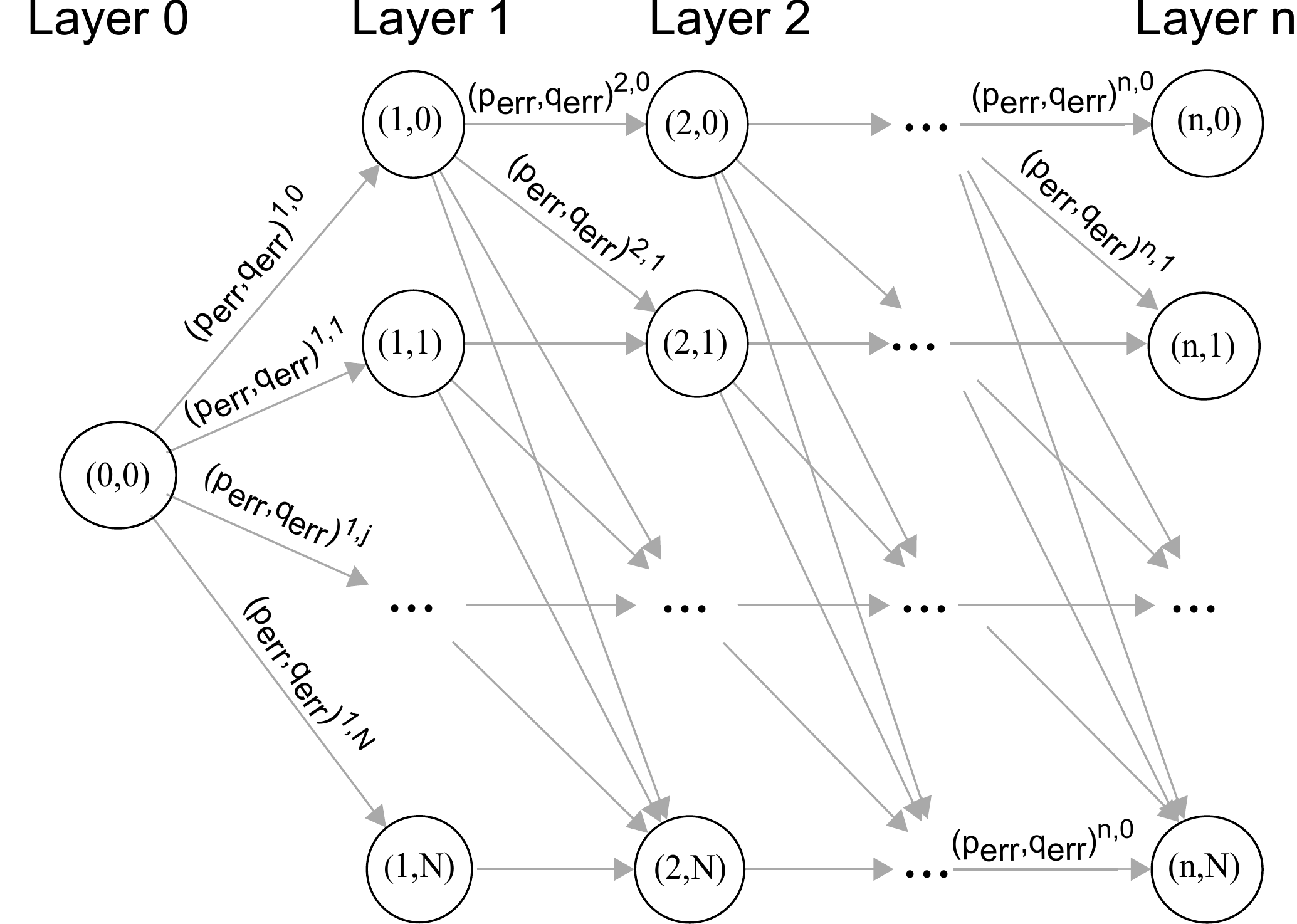}\\
	\caption{Search space graph $G_{TR}$ for the Target Resemblance problem. Layer $0$ is an auxiliary layer that just contains the node $(0,0)$. Layer $i = 1, \ldots, n$ corresponds to bin $i$ of the sanitized histogram, and node $(i, j)$ corresponds to allocating $j = 1, \ldots, N$ counts to bins 1 up to and including $i$. A path from $(0,0)$ to $(n, N)$ completely defines an allocation of $N$ counts to $n$ bins. The weight of the edge from $(i,j)$ to $(i+1, j+k)$ is the privacy and quality error of allocating exactly $k$ counts to bin $i+1$. As these errors are additive, the admissible paths are those whose total $q$-length is less than the threshold $\epsilon$. Among them, the $p$-shortest path from $(0,0)$ to $(n, N)$ corresponds to the optimal solution to {\TR}, because it also has $q$-length at most $\epsilon$} \label{fig:AO_graph_R}
\end{figure}

From the histogram $H$ and the distance functions $d_p$ and $d_q$, we construct a multipartite DAG $G_{TR}=(V,E)$, as follows (see also \figref{fig:AO_graph_R}):
\begin{itemize}
	\item There are $n\cdot (N + 1) + 1$ nodes in $V$, where $n$ and $N$ are the length and the size of the histogram $H$, respectively.
	\item The nodes are arranged in layers $0, 1, ..., n$, with layer $0$ having a single node and layers $1,...,n$ having $N+1$ nodes each. Layer $i \in [1, n]$ corresponds to bin $i$ (location $L_i$) in the histogram. Node $j \in [0, N]$ in layer $i$ corresponds to the allocation of a total of $j$ frequency counts to histogram bins $1$ up to and including $i$.
	\item The single node in layer $0$ is labeled $(0,0)$, and each node $j$ in every other layer $i$ is labeled $(i, j)$. We may refer to nodes of $G_{TR}$ using their labels.
	\item The edges in $E$ go from each node $(i,j)$ to each node $(i+1, j+k), k\geq 0, j+k \leq N$, i.e.  to each node in the following layer that has a frequency count at least equal to $j$.
	\item The weight of an edge from $(i,j)$ to $(i+1, j+k)$ is the pair $ (p_{err}, q_{err})^{i+1,k}$ of the privacy and quality errors of allocating exactly $k$ counts to bin $i+1$ of the sanitized histogram: $p_{err} = p(k,H''[i+1]), q_{err} = q(H[i+1], k)$. The $p$-length of a path is the sum of its $p_{err}$ weights. We will refer to the path with the minimum $p$-length as the $p$-shortest path. Similarly, the $q$-length of a path is the sum of its $q_{err}$ weights.
\end{itemize}

At this point, note two important differences between the edge weights of $G_{TR}$ and $G_{LHO}$ (Section \ref{LHOalgorithmsectin}): First, and most obvious, the edge weights in $G_{TR}$ are pairs of (privacy error, quality error), whereas in {\SLH} the weights are quality errors. Second, in $G_{TR}$ the weight of edge from $(i,j)$ to $(i+1, j+k)$ corresponds to setting $H'[i+1]$ exactly equal to $k$, whereas in $G_{LHO}$ that edge weight would correspond to setting $H'[i+1]$ equal to $H[i+1] + k$.

From the construction of $G_{TR}$, it follows that there is a $1-1$ correspondence between a sanitized histogram $H' \in \mathcal{H}_{n,N}$ and a path from $(0,0)$ to $(n,N)$ in $G_{TR}$. 
Therefore, to solve the {\TR} problem, we need to find the path from $(0, 0)$ to $(n, N)$ with minimum $p$-length among the paths whose $q$-length is at most $\epsilon$. Then, it is straightforward to construct the histogram from the path.

We now provide the pseudocode of the $RO$ algorithm. We assume that the pre-processing needed to construct $H''$ from $h''$ is done before the actual algorithm runs, and also $H$ and $H''$ have been expanded to be defined on the same set of locations, if needed (see Section~\ref{subsec:TR_definition} for details on $h''$). 
Also, for the moment, we assume that $d_q$ takes nonnegative integer values.
\scalebox{1}{}{
	\begin{algorithm}\footnotesize
		\Algorithm{RO (Resemblance Optimal)}
		\KwIn{Histogram $H$, target histogram $H''$, privacy distance function $d_p$, quality distance function $d_q$, maximum quality loss threshold $\epsilon$}
		\KwOut{Sanitized histogram $H'$}
		Construct the graph $G_{TR}$ \label{RO1}\\
		\ForEach{node $v$ in $G_{TR}$}{\label{RO2}
			\If{the label of $v$ is $(0,0)$}{\label{RO3}
				Associate the node $v$ with a vector $\mathcal{V}_v$ s.t. $\mathcal{V}_v[k]=0$, for each integer $k\in[0,\epsilon]$ \label{RO4}\\
			}
			\Else{\label{RO5}
				Associate the node $v$ with a vector $\mathcal{V}_v$ s. t. $\mathcal{V}_v[k]=\infty$, for each integer $k\in[0,\epsilon]$ \label{RO6}\\
			}
		}
		\ForEach{node $v$ in $G_{TR}$ in increasing lexicographic order starting from node $(1,0)$}{\label{RO7}
			\ForEach{element $k$ of $\mathcal{V}_v$}{\label{RO8}
				$\mathcal{V}_v[k]=\min_{(u,v) \in E, q_{err}(u,v) \leq k}\{ \mathcal{V}_u[k - q_{err}(u,v)] + p_{err}(u,v) \}$\label{RO9}\\
			}
		}
		$((0,0),(1,j_1),\ldots,(n,N))\gets$ the shortest path from node $(0,0)$ to node $(n,N)$ in \myindent{3.3} $G_{TR}$. Its $p$-length is equal to the minimum element of $\mathcal{V}_{v}$ \myindent{3.3}  for node $v = (n,N)$.\label{RO10}\\
		$H'\gets (j_1, j_2-j_1, \ldots, N-j_{N-1})$\label{RO11}\\
		\Return $H'$ \label{RO12}\\
	\end{algorithm}
}

In step \ref{RO1}, $RO$ constructs the graph $G_{TR}$. In steps \ref{RO2} to \ref{RO6}, the algorithm iterates over each node $v$ of the graph and associates with it a vector $\mathcal{V}_v$, indexed by all possible values of $d_q$. The elements of $\mathcal{V}_v$ are initialized to $0$ for node $(0,0)$, and to $\infty$ for any other node of $G_{TR}$. Next, in steps \ref{RO7} to \ref{RO9}, $RO$ iterates over the nodes of $G_{TR}$ in increasing lexicographic order, starting from node $(1,0)$, and for each node $v$ it updates all the elements of $\mathcal{V}_v$. Each element $\mathcal{V}_v[k]$ is updated using the following dynamic programming equation:
\begin{equation}\label{eq:vk}
\mathcal{V}_v[k] = \min_{(u,v) \in E, q_{err}(u,v) \leq k}\{ \mathcal{V}_u[k - q_{err}(u,v)] + p_{err}(u,v) \}.
\end{equation}
The element $\mathcal{V}_v[k]$ is equal to the $p$-length of the $p$-shortest path from $(0,0)$ to node $v$ with $q$-length exactly equal to $k$. Thus, as explained above, this path is a feasible solution to the {\TR} problem, and so the optimal solution to {\TR} is the $p$-shortest path from $(0,0)$ to $(n,N)$ (i.e., the path corresponding to the minimum element of the vector $\mathcal{V}_{(n, N)}$). The nodes of this path are found in step \ref{RO10} and its corresponding histogram $H'$ is constructed in step \ref{RO11}.

We now consider the general case in which the values of $q$-length of a path from $(0, 0)$ to $(n, N)$ are not necessarily integer.
We first show that the $q$-length of this path is polynomial in $N$, in Theorem \ref{thm:general_q_case} below. Then, we show that the number of values of $q$-length for all paths is polynomial in $N$, which implies that these values are not too many to store in the vectors $\mathcal{V}_u$.

\begin{thm}
	The $q$-length of a path from $(0, 0)$ to $(n, N)$ can only take a polynomial (in $N$) number of values.\label{thm:general_q_case}
\end{thm}
\begin{proof}
	The number of possible allocations of $N$ elements to $n$ bins, where some of the bins may be left empty, is equal to the number of $n$-tuples of non-negative integers $f_1,...,f_N$ that sum to $N$. Such tuples are called \emph{weak compositions} of $N$ into $n$ parts (\emph{weak}, because zeros are allowed), and their total number is
	\begin{equation}
	\binom{N + n - 1}{n - 1} = \frac{(N+1)...(N+n-1)}{(n-1)!},\label{qlenpol}
	\end{equation}
	which is a polynomial in $N$ \cite{bona2011walk}.
\end{proof}

Eq. \ref{qlenpol} gives all possible $q$-lengths for a path from $(0, 0)$ to $(n, N)$. We also need to keep \emph{intermediate} $q$-length values in the vectors $\mathcal{V}_v$, i.e. $q$-lengths for paths from $(0, 0)$ to nodes in any layer $1, ..., n$. But each intermediate allocation has at most as many $q$-length values as the final one, because an intermediate allocation allocates at most $N$ elements to at most $n$ bins. As there are $n$ stages of intermediate allocations, we have in total at most $n\cdot \binom{N + n - 1}{n - 1}$ values, i.e. still a polynomial in $N$.

\begin{example}
	$RO$ is applied to the histogram $H=(7,2,3,2,13,12,8,3)$ in Figure \ref{ex1a}, using the target histogram $H''=(10,8,6,2,13,4,4,3)$ in Figure \ref{ex1c}, $JS$-divergence as the quality distance function $d_q$ and the privacy distance function $d_p$, and $\epsilon=0.05$. In step \ref{RO1}, the algorithm constructs the search space graph in Figure \ref{RunExampleROfig}. The graph has $n+1=9$ layers of nodes, where $n=8$ is the length of $H$. Layer $0$ contains the node $(0,0)$ and each other layer contains $N+1=51$ nodes, where $N=50$ is the size of $H$. Each node in layers $1, \ldots, 8$ is labeled $(i,j)$; $i\in[1,8]$ denotes the layer of the node and corresponds to bin $i$, while $j\in[0,50]$ denotes the counts allocated to bins $1, \ldots, i$. For example, the node $(8,50)$ denotes that all $50$ counts of $H$ are allocated to bins $1,\ldots, 8$.
	In addition, there is an edge from each node $(i,j)$ to every node $(i+1, j+k)$, for each $k\in [0, 50-j]$. The edge weight is a pair $(p_{err}, q_{err})$, where the privacy error $p_{err}$ (respectively, quality error $q_{err}$) quantifies the error with respect to $JS$-divergence that is incurred by allocating  $k$ counts to bin $i+1$ of the sanitized histogram (see Figure \ref{RunExampleROfig}). For example, the node $(0,0)$ is connected to the nodes $(1,0), \ldots, (1,50)$, and the edge $((0,0), (1,10))$ has  $(p_{err},q_{err})^{1,10}=(0, 3.8\cdot 10^{-3})$, incurred by allocating $10$ counts to the first bin. In steps \ref{RO2} to \ref{RO9}, $RO$ computes the vector $\mathcal{V}_u$ for each node $u$. In step \ref{RO10}, the algorithm finds the shortest path from node $(0,0)$ to node $(8,50)$ with $q_{err} \leq \epsilon$ (see Figure \ref{RunExampleROfig}), and in step \ref{RO11} it constructs the sanitized histogram $H'=(10,6,5,2,14,5,5,3)$ that corresponds to the shortest path (see Figure \ref{ex1d}). Note that $j$ in the label $(i,j)$ of each node in the shortest path corresponds to the counts of sensitive locations that are allocated to bins $1, \ldots, i$ in $H'$. Last, in step \ref{RO12}, $H'$ is returned. \label{RunExampleRO} \qedsymbol
\end{example}

\begin{figure}[ht!]\hspace{-4mm}
	\includegraphics[trim={0 0 0 0},scale=0.26]{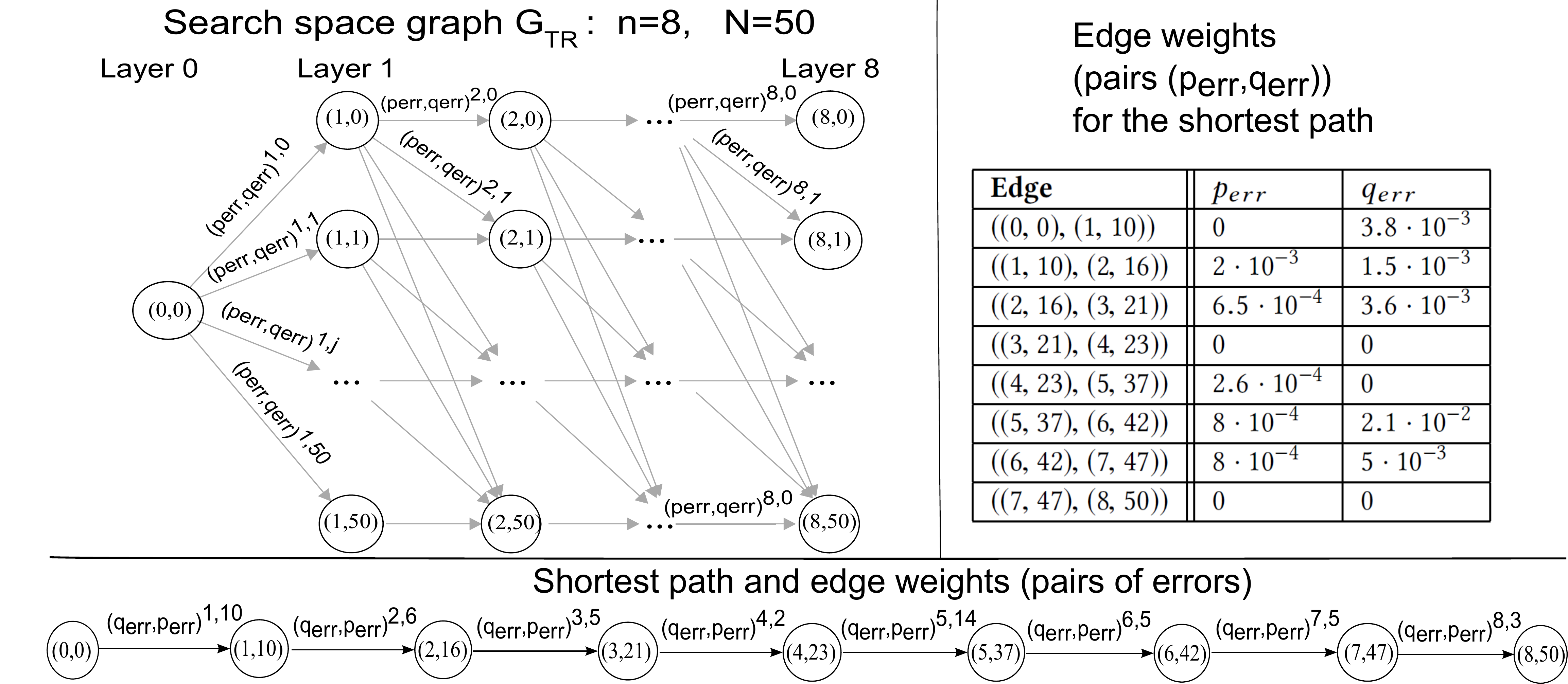}
	\caption{Search space graph $G_{TR}$ for Example~\ref{RunExampleRO} (the missing nodes and edges are denoted with ``$\ldots$''), and shortest path along with its corresponding edge weights}\label{RunExampleROfig}
\end{figure}

The time complexity of the $RO$ algorithm is $O((n\cdot N)^2\cdot \binom{N + n - 1}{n - 1})$. The total cost is the sum of the cost of
constructing $G_{TR}$ and of finding the constrained shortest path from $(0, 0)$ to $(n, N)$.

The construction of $G_{TR}$ takes $O(n\cdot N^2)$ time. This is because
the algorithm constructs $O(n\cdot (N + 1) + 1)=O(n\cdot N)$ nodes, each of which has $O(N)$ outgoing edges, for a total of $O(n\cdot N^2)$ edges. Note also that the computation of each edge weight takes $O(1)$ time. The cost of computing the shortest path is $O((n\cdot N)^2\cdot \binom{N + n - 1}{n - 1})$. This is because it requires (I) constructing a vector $\mathcal{V}_v$ with $O(n\cdot \binom{N + n - 1}{n - 1})$ entries, for each node $v$ of the $O(n\cdot N)$ nodes of $G_{TR}$, which takes $O(n^2\cdot N\cdot \binom{N + n - 1}{n - 1})$ time, and (II) updating each entry of $\mathcal{V}_v$ once, which takes $O(N)$ time per node since there are $O(N)$ incoming edges to each node (see Eq. \ref{eq:vk}), for a total of $O((n \cdot N)^2 \cdot \binom{N + n - 1}{n - 1})$ across all nodes.

Last, we note how the variation {\TR}$_{|H''|_1}$ of the {\TR} problem (see Section~\ref{subsec:TR_definition}) can be optimally solved with $RO$.
This is possible by simply using $RO$ to allocate $|H''|_{1}$ counts instead (i.e., construct a search space graph whose layers $1, \ldots n$ have $|H''|_1+1$ counts each, and then find the shortest path from $(0,0)$ to $(n,|H''|_1)$ and the sanitized histogram $H'$ corresponding to the path).

\subsection{Heuristics for Target Resemblance and Target Avoidance}\label{sec:solution:heuristics}
Our heuristics, $RH$ for the Target Resemblance and $AH$ for the Target Avoidance problem, work in a greedy fashion to avoid the cost of constructing and searching the search space graph.
\hspace{-2mm}
\scalebox{0.6}{}{
	\begin{algorithm}[ht]\footnotesize
		\Algorithm{RH (Resemblance Heuristic)}
		\KwIn{Histogram $H$, target histogram $H''$, privacy distance function $d_p$, quality distance funct. $d_q$, quality thr. $\epsilon$}
		\KwOut{Sanitized histogram $H'$}
		$SrcBins\gets \{i~such~that~H[i]> H''[i]\}$ \label{RH1}\\
		$DstBins\gets \{i~such~that~H[i]< H''[i]\}$ \label{RH2}\\
		$H'\gets H$ \label{RH3}\\
		$\epsilon_{rem}\gets \epsilon$\label{RH4} ~~~~\tcp{Remaining quality budget}
		\While{$SrcBins\neq \varnothing$}{ \label{RH5}
			\tcp{Perform the best move on $H'$}
			$(H', Opt\Delta{d_q})\gets {\bf BestMove}(H', H'', SrcBins, DstBins, \epsilon_{rem})$\label{RH6}\\
			\tcp{Exit if the remaining budget is exhausted}
			\uIf{$Opt\Delta{d_q} = -1$}{\label{RH7}
				{\bf break}\label{RH8}\\
			}
			\tcp{Update the remaining budget}
			$\epsilon_{rem}\gets \epsilon_{rem} -Opt\Delta{d_q}$\label{RH9}\\
			\tcp{Update the set of source and dest. bins}
			$SrcBins\gets \{i~such~that~H'[i]> H''[i]\}$ \label{RH10}\\
			$DstBins\gets \{i~such~that~H'[i]< H''[i]\}$ \label{RH11}\\
		}
		\Return $H'$ \label{RH12}\\
	\end{algorithm}
}

\scalebox{0.6}{}{
	\begin{algorithm}[htbp]\footnotesize
		\Function{BestMove}
		\KwIn{Sanitized histogram $H'$, target histogram $H''$, Set of source bins $SrcBins$, Set of destination bins $DstBins$, Remaining budget $\epsilon_{rem}$, Privacy distance function $d_p$, Quality distance function $d_q$}
		\KwOut{Sanitized histogram $H'$ after performing the best move, difference in $d_q$ incurred by the  best move}
		$MaxRatio \gets 0$\\\label{FB1}
		$Opt\Delta{d_q}\gets -1$ \label{FB2}\\
		$H'_{BestMove}\gets H'$\label{FB3}\\
		\ForEach{bin $i$ in $SrcBins$}{ \label{FB4}
			\ForEach{bin $j \neq i$ in $DstBins$}{ \label{FB5}
				\ForEach{$k\in [1,H'[i]]$}{ \label{FB6}
					\tcp{Try moving $k$ counts from a source bin $H'[i]$ to a destination bin $H'[j]$}
					$H'_{tmp}\gets H'$ \label{FB7}\\
					$H'_{tmp}[i]\gets H'[i]-k$ \label{FB8}\\
					$H'_{tmp}[j]\gets H'[j]+k$ \label{FB9}\\
					${\bf\Delta}{d_p}\gets \left|d_p(H'_{tmp},H'')-d_p(H',H'')\right|$\label{FB10}\\
					$\Delta{d_q}\gets d_q(H'_{tmp},H)-d_q(H',H)$\label{FB11}\\
					\tcp{Store the best sanitized histogram so far, its ratio and remaining budget}
					\If{$\frac{\Delta{d_p}}{\Delta{d_q}}>MaxRatio$~~and~~$\Delta{d_q}<\epsilon_{rem}$}{\label{FB12}
						$H'_{BestMove}\gets H'_{tmp}$\label{FB13}\\
						$MaxRatio\gets \frac{\Delta{d_p}}{\Delta{d_q}}$\label{FB14}\\
						$Opt\Delta{d_q}\gets  \Delta{d_q}$\label{FB15}\\
					}
				}
			}
		}
		$H'\gets H'_{BestMove}$\label{FB16}\\
		\Return ($H'$, $Opt\Delta{d_q}$) \label{FB17}\\
	\end{algorithm}
}

We first discuss the $RH$ heuristic. The main idea in $RH$ is to try to greedily reduce the differences in the counts of corresponding bins between $H$ and $H''$.
As can be seen in the pseudocode (steps \ref{RH1} and \ref{RH2}), $RH$ identifies \emph{source bins}, i.e. bins in $H$ with more counts in $H$ than in $H''$, and \emph{destination bins}, bins with fewer counts in $H$ than in $H''$. Bins with equal counts in $H$ and $H''$ are ignored. Then, in steps \ref{RH3} and \ref{RH4}, the sanitized histogram $H'$ is initialized to the original histogram $H$ and the remaining quality budget $\epsilon_{rem}$ to the quality threshold $\epsilon$. In steps \ref{RH5} and \ref{RH6}, $RH$ moves some counts from a source bin to a destination bin using a function $BestMove$.

As can be seen in the pseudocode of $BestMove$ (steps \ref{FB4} to \ref{FB6}), the function performs an exhaustive search of all possible ways (``moves'') to move $k$ counts from a source bin $i$ to a destination bin $j$. For each move, $BestMove$ computes the privacy effect $\Delta d_p$ and the quality effect $\Delta d_q$ (steps \ref{FB10} and \ref{FB11}), and it selects the move that maximizes the ratio $\frac{\Delta d_p}{\Delta d_q}$, subject to the constraint that $\Delta d_q$ cannot exceed the remaining quality budget $\epsilon_{rem}$ (steps \ref{FB12} to \ref{FB15})\footnote{Note that $\Delta d_q > 0$, because $\Delta d_q$ is the sum of two positive terms (in square brackets):
	$\Delta{d_q} = d_q(H'_{tmp}, H) - d_q(H', H) =
	q(H'_{tmp}[i], H[i]) + q(H'_{tmp}[j], H[j]) - q(H'[i], H[i]) - q(H'[j], H[j]) = \left[q(H'[i]-k, H[i]) - q(H'[i], H[i])\right] + \left[q(H'[j]+k, H[j])  - q(H'[j], H[j])\right].$ The first term is positive: Bin $i$ is a source bin, which means $H'[i] \leq H[i]$. But $H'[i]-k < H'[i]$, so the distance $q(H'[i]-k, H[i])$ is larger than $q(H'[i], H[i])$. Similarly, the second term is positive, because bin $j$ is a destination bin.}. The rationale is to prioritize moves with a large improvement in privacy $\Delta d_p$ and a small reduction in quality $\Delta d_q$.

Next, in step \ref{RH7}, $RH$ checks whether the remaining budget is exhausted. If it is, no more moves are performed (step \ref{RH8}). Otherwise, in steps \ref{RH9} to \ref{RH11}, $RH$ reduces the quality budget by $Opt\Delta_{d_q}$ (i.e., by the quality effect of the best move), and updates the sets of source and destination bins by no longer considering as source or destination bins any bins whose count has become equal to the corresponding bin in $H''$. Moves continue until the budget is exhausted or there are no more source/destination bins. Since moves cannot increase the count of a source bin nor increase the remaining quality budget, $RH$ will always terminate.

The main idea in $AH$ is very similar. The difference is in the definition of source and destination bins: $AH$ aims to make $H'$ dissimilar to $H''$, and so it tries to
\emph{increase} differences in counts between bins in $H'$ and corresponding bins in $H''$. Hence, source bins (from which counts will be taken) are bins that have a shortage of counts (fewer counts in $H'$ than in $H''$), whereas destination bins are the ones with a surplus of counts. Unlike in $RH$, bins with equal counts in $H$ and $H''$ are source bins \emph{as well as} destination bins (see steps \ref{AH1} and \ref{AH2} in the pseudocode of $AH$, in which source and destination bins are initialized, and steps \ref{AH10} and \ref{AH11} in which source and destination bins are updated). $AH$ then proceeds as $RH$, only making sure never to use the same bin as both source and destination.

\scalebox{0.6}{}{
	\begin{algorithm}[ht]\footnotesize
		\Algorithm{AH (Avoidance Heuristic)}
		\KwIn{Histogram $H$, target histogram $H''$, privacy distance function $d_p$, quality distance funct. $d_q$, quality thr. $\epsilon$}
		\KwOut{Sanitized histogram $H'$}
		$SrcBins\gets \{i~such~that~H[i]\leq H''[i]\}$ \label{AH1}\\
		$DstBins\gets \{i~such~that~H[i]\geq H''[i]\}$ \label{AH2}\\
		$H'\gets H$ \label{AH3}\\
		$\epsilon_{rem}\gets \epsilon$\label{AH4} ~~~~\tcp{Remaining quality budget}
		\While{$SrcBins\neq \varnothing$}{ \label{AH5}
			\tcp{Perform the best move on $H'$}
			$(H', Opt\Delta_{d_q}) \gets {\bf BestMove}(H', H'', SrcBins, DstBins, \epsilon_{rem})$\label{AH6}\\
			\tcp{Exit if the remaining budget is exhausted}
			\uIf{$Opt\Delta{d_q} = -1$}{\label{AH7}
				{\bf break}\label{AH8}\\
			}
			\tcp{Update the remaining budget}
			$\epsilon_{rem}\gets \epsilon_{rem} -Opt\Delta{d_q}$\label{AH9}\\
			\tcp{Update the set of source and dest. bins}
			$SrcBins\gets \{i~such~that~0 < H'[i] \leq H''[i]\}$ \label{AH10}\\
			$DstBins\gets \{i~such~that~H'[i] \geq H''[i]\}$ \label{AH11}\\
		}
		\Return $H'$ \label{AH12}\\
	\end{algorithm}
}

The time complexity of $RH$ and $AH$ is $O(n^3\cdot N)$. This is because the loop in step \ref{RH5} runs $O(n)$ times (once per source bin), and each time there is a cost of $O(n^2\cdot N)$ incurred by $BestMove$. The cost of $BestMove$ is $O(n^2\cdot N)$, because there are $O(n^2)$  source/destination bin pairs, and for each pair $O(N)$ temporary moves are performed. The time complexity analysis refers to the worst case. In practice, a histogram can be sanitized with a smaller number of moves (i.e., executions of $BestMoves$), and the heuristics scale well with respect to $n$. For example, in our experiments, the heuristics scale close to linearly with respect to $n$.

Last, we note that the variation {\TR}$_{|H''|_1}$ of the {\TR} problem (see Section \ref{subsec:TR_definition}) can be directly dealt with by the $RH$ heuristic. This is because $RH$ does not pose any restriction on the size of $H''$, so $H''$ can have a different size than that of $H$. Similarly, the variation {\TA}$_{|H''|_1}$ of the {\TA} problem (see Section \ref{subsec:TA_definition}) can be directly dealt with by the $AH$ heuristic.

\section{Evaluation}\label{sec:evaluation}

In this section, we evaluate our approach in terms of effectiveness and efficiency. We do not compare against existing
histogram sanitization methods, because they cannot be used to solve the Sensitive Location Hiding or the Target Avoidance/Resemblance problem (see Related Work in Section~\ref{related_work_histogram_privacy}).  

\subsection{Setup and datasets}

To calculate the loss in quality (utility) incurred by replacing the original histogram $H$ with the sanitized histogram $H'$, we compute the distance $d_q(H,H')$, 
where $d_q$ is either the Jensen-Shannon (JS) divergence (Section~\ref{sec:problem:ql}), or the $L_2$ distance. Our algorithms can optimize either of these measures. However, we present results for optimizing $JS$ divergence, because, as we show, the two measures lead to qualitatively similar results.

In addition, we measure how well sanitization preserves the quality of two applications: (I) (histogram) clustering, and (II) location recommendation.
Clustering partitions the elements (frequencies) of a histogram into clusters, so that each cluster contains similar frequencies. Clustering is typically applied to visualize (long) histograms, or to segment locations into pre-defined categories based on frequency of visits. For example, clustering a histogram $H=(1,2,3,98,99,100)$ may result in two clusters,  $(1,2,3)$  and $(98,99,100)$, corresponding to ``rarely-visited'' and ''frequently-visited'' locations, respectively. Clustering quality is measured using Normalized Conditional Entropy ($NCE$), a normalized version of the well-known Conditional Entropy cluster quality measure \cite{dmbook}. We first provide the definition of Conditional Entropy and then that of $NCE$. Let $\mathcal{C}$ (respectively, $\mathcal{C}'$) be a partition of the elements of $H$ (respectively, $H'$) into $k$ mutually disjoint and nonempty clusters. The partitions $\mathcal{C}$ and $\mathcal{C}'$ will also be referred to as clusterings. The Conditional Entropy $\mathcal{H}(C'|C)=\sum_{c'\in \mathcal{C}',c\in \mathcal{C}}\Pr(c,c')\cdot ln(\frac{\Pr(c)}{\Pr(c,c')})$ is the entropy of $\mathcal{C}'$ conditioned on $\mathcal{C}$. $\Pr(c)=\frac{|c|}{N}$ is the probability that a randomly selected element of $H$ is contained in the cluster $c\in \mathcal{C}$,  $\Pr(c')=\frac{|c'|}{N}$ is the probability that a randomly selected element of $H'$ is contained in the cluster $c'\in\mathcal{C}'$, and $\Pr(c',c)=\frac{|c'\cap c|}{N}$ is the probability that a randomly selected element is contained in $c' \cap c$. Intuitively, $\mathcal{H}(\mathcal{C}'|\mathcal{C})$ quantifies the amount of information needed to describe $\mathcal{C}'$, when $\mathcal{C}$ is known. The Normalized Conditional Entropy $NCE(\mathcal{C}'|\mathcal{C})$ is defined as $NCE(\mathcal{C}'|C)=\frac{\mathcal{H}(\mathcal{C}'|\mathcal{C})}{\mathcal{H}(\mathcal{C}')}$, where $\mathcal{H}(\mathcal{C}')=-\sum_{c'\in \mathcal{C}'}\Pr(c')\cdot ln(\Pr(c'))$ is the entropy of $\mathcal{C}'$. Since $\mathcal{H}(C'|C)\leq \mathcal{H}(C')$ \cite{dmbook}, $NCE(\mathcal{C}'|\mathcal{C})$ is bounded in $[0,1]$, with $0$ implying that the two clusterings are the same and $1$ that they are independent. We apply the optimal dynamic-programming algorithm of Wang and Song \cite{ckmeans} to produce $\mathcal{C}$ and $\mathcal{C}'$ with $k=3$.
Results with different $k$ are similar (omitted). Given a histogram $H$ and an integer $k$, the algorithm produces a clustering $\mathcal{C}$ that minimizes $\sum_{c\in \mathcal{C}} \sum_{H[i]\in c} ({\bf L_2}(H[i], \bar{c}))^2$, where $c$ is a cluster, $\bar{c}$ is the mean of the elements in the cluster, and ${\bf L_2}$ is the $L_2$ distance.

Location recommendation suggests to a user, referred to as \emph{active user} and denoted with $\alpha$, a location that might interest them. A popular location recommendation approach is user-based Collaborative Filtering (CF), which works as follows (see \cite{Melville} for details): (I) It finds a set $U_{\alpha,\mathtt{k}}=\{u_1, \ldots, u_{\mathtt{k}}\}$ of $\mathtt{k}$ users who are the most similar to the active user $\alpha$, with respect to their histograms.
The similarity between a user $u$ and the active user $\alpha$ is measured by the Pearson correlation coefficient, which is defined as
$PCC(u,{\alpha}) =
\frac{\sum_{L_i\in L_{u,\alpha}} (f_{\alpha}(L_i) - \mu_{\alpha}) \cdot (f_{u}(L_i) - \mu_{u})} {\sqrt{\sum_{L_i\in L_{u,\alpha}} (f_{\alpha}(L_i) - \mu_{\alpha})^2 \cdot (f_{u}(L_i) - \mu_{u})^2}}$, where $L_{u,\alpha}$ is the set of locations in the histograms of both $u$ and $\alpha$, $f_{u}(L_i)$ (respectively, $f_{{\alpha}}(L_i)$ ) is the number of times the user $u$ (respectively, $\alpha$) visited location $L_i$ (i.e., the count of the location $L_i$ in the histogram of a user), and $\mu_{u}$ (respectively, $\mu_{\alpha}$) is the average number of times $u$ (respectively, $\alpha$) visited a location.
(II) For the active user $\alpha$ and each location $L_i$, it computes a recommendation score (predicted frequency of visiting $L_i$), defined as $r_\alpha(L_i) = \mu_{\alpha} + \frac{\sum_{u\in U_{\alpha,{\mathtt{k}}}} (f_u(L_i)-\mu_u) \cdot PCC(u,\alpha)} {\sum_{u\in U_{\alpha,{\mathtt{k}}}} PCC(u,\alpha)}$.
(III) It recommends to the active user the location with the maximum recommendation score. In our experiments, we use the aforementioned user-based CF method with $\mathtt{k}=25$. 
The \emph{recommendation error} \cite{Melville} for an active user $\alpha$ and a location $L_{test}^\alpha$ is defined as the difference between $f_\alpha(L_{test}^\alpha)$, the user's \emph{true} frequency of visits to $L_{test}^\alpha$, and $r_\alpha(L_{test}^\alpha)$, the frequency of visits as \emph{predicted} by the recommendation algorithm based on the given dataset. Below we use both the absolute error $|f_\alpha(L_{test}^\alpha) - r_\alpha(L_{test}^\alpha)|$ and the square error $(f_\alpha(L_{test}^\alpha) - r_\alpha(L_{test}^\alpha))^2$.

To capture the impact of the sanitization algorithm on recommendation quality, we compute the above defined recommendation error in two ways: First, we compute recommendations based on the dataset of original user histograms, and then based on the dataset of sanitized histograms.
Clearly, the impact of sanitization is small when the average recommendation error for the dataset of original user histograms is similar to that for the dataset of sanitized histograms.
In particular, we perform the following steps on each of the two datasets (we use the absolute difference as the recommendation error, but the same steps apply for the squared difference):\\
(I) We randomly partition the dataset into 2 subsets, a training set $D_{train}$ with 90\% of the histograms and a test set $D_{test}$ with 10\% of the histograms.\\
(II) For an active user $\alpha$ in the test set $D_{test}$, we randomly select a location $L_{test}^\alpha$ in the histogram of $\alpha$. We compute the similarities $PCC(u,\alpha)$ between $\alpha$ and all users $u$ in the training set $D_{train}$ (see step (I) of the Collaborative Filtering algorithm). In the computation of $PCC$, we exclude $L_{test}^\alpha$ from the set of locations $L_{u,\alpha}$ that are in the histograms of both $u$ and $\alpha$. Then, we compute the recommendation score $r_\alpha(L_{test}^\alpha)$, and finally we compute the error $|f_\alpha(L_{test}^\alpha) - r_\alpha(L_{test}^\alpha)|$.\\
(III) We compute the absolute recommendation error considering each user in the test dataset $D_{test}$ as the active user $\alpha$, and then we average the errors to obtain the \emph{Mean Absolute Error} $MAE(D_{test}) = \frac{1}{|D_{test}|}\sum_{(\alpha, L_{test}^\alpha)} |f_\alpha(L_{test}^\alpha) - r_\alpha(L_{test}^\alpha)|$.

For the square error, the average we compute is the \emph{Root Mean Square Error} $RMSE(D_{test}) = \sqrt{\frac{1}{|D_{test}|}\sum_{(\alpha, L_{test}^\alpha)} (f_\alpha(L_{test}^\alpha) - r_\alpha(L_{test}^\alpha))^2}$.

All algorithms are implemented in Python and applied to the New York City (\emph{NYC}) and Tokyo (\emph{TKY}) datasets.
The datasets were downloaded from \cite{dataseturl} and include long-term check-in data in New York city and Tokyo, collected from Foursquare from 12 April 2012 to 16 February 2013.
The datasets have been used in several prior works \cite{datasetpaper,lpd,datasetpaper3}.
Each record in the datasets contains a location that was visited by a user at a certain time and
corresponds to a leaf in the Foursquare taxonomy (available at \url{https://developer.foursquare.com/docs/resources/categories}). There are in total $713$ locations in the taxonomy, and on average each user visits fewer than $41$ locations.
For each dataset, we produce the input histograms for our algorithms by constructing one histogram $H$ per user. The histogram $H$ contains a count $f(L_i)>0$ for every location $L_i$ visited by the user. That is, $H$ is constructed based on the user's values (location visits), which is line with histogram sanitization methods \cite{acs2012differentially,xu2013differentially, kellaris, sdmhistogram,sortaki,qardaji,li2016improving}. Table~\ref{tabdata} shows the characteristics of \emph{NYC} and \emph{TKY}, and Table \ref{tabdefault} shows the default values used in our experiments.
\begin{table}[!ht]
	\footnotesize\centering
	\scalebox{0.9}{
		\begin{tabular}{|c||c|c|c|c|}
			\hline 
			\textbf{Dataset} & \textbf{\# histograms} & \textbf{mean of length} $n$  & \textbf{max. length} $n$ & \textbf{mean of size} $N$   \tabularnewline
			\hline
			\emph{NYC} & 1,083 & 40.28 & 139 & 209.99   \tabularnewline \hline
			\emph{TKY} & 2,293 & 32.36 & 158 & 221.57 \tabularnewline \hline
	\end{tabular}}
	\caption{Characteristics of datasets}\label{tabdata}
	\vspace{-1mm}
\end{table}

\begin{table}[!ht]\centering
	\footnotesize
	\scalebox{0.9}{
		\begin{tabular}{|c||c|c|c|c|c|}
			\hline {\bf Dataset} &  $n$ & $K$ & $|L'|$ & $\epsilon$ & $N$  \tabularnewline\hline
			\emph{NYC} & 25 & 20 & 5 & $5\cdot10^{-3}$ & 100 \tabularnewline \hline
			\emph{TKY} & 35 & 20 & 5 & $5\cdot10^{-3}$ & 100 \tabularnewline \hline
	\end{tabular}}
	\caption{Default values for each dataset with respect to: length $n$, total frequency of sensitive locations $K$, number of sensitive locations $|L'|$, threshold $\epsilon$, and size of histogram $N$}\label{tabdefault}
	\vspace{-1mm}
\end{table}

We also construct synthetic histograms containing some bins whose frequency is equal to zero. These bins correspond to locations that are not visited by a user but should be considered in sanitization, e.g. to allow the sanitization algorithm to redistribute frequency counts to locations that are not visited by the user. The synthetic histograms are constructed by appending zeros to a histogram of length $n=78$ and size $N=192$ in \emph{NYC} and to a histogram of $n=99$ and $N=642$ in \emph{TKY}, and their length is up to $400$, including the zero-frequency bins. We use the synthetic histograms to test the impact of length on the runtime performance of our methods. In total, we test the algorithms on approximately 3400 different histograms. For all experiments we use an Intel Xeon at 2.60GHz with 256GB of RAM.

\vspace{-3mm}
\subsection{Evaluation of the $LHO$ algorithm}

We evaluate the quality and runtime performance of the $LHO$ algorithm as a function of (I) $n$, the length of the original histogram, (II) $K$, the total frequency of sensitive locations, and
(III) $|L'|$, the number of sensitive locations. We consider JS-divergence, $L_2$ distance, and $NCE$ as measures of quality loss $d_q()$. The results for the $L_2$ distance are in Appendix~\ref{L2appendix}, because they
are similar to those for JS-divergence. Unless otherwise stated, the set of sensitive locations $L'$ is constructed by selecting $5$ sensitive locations uniformly at random.

\begin{figure*}[htbp]\hspace{-2mm}
	\begin{subfigure}[b]{.24\textwidth}\centering
		\includegraphics[trim={0 0 0 0},scale=0.3,clip]{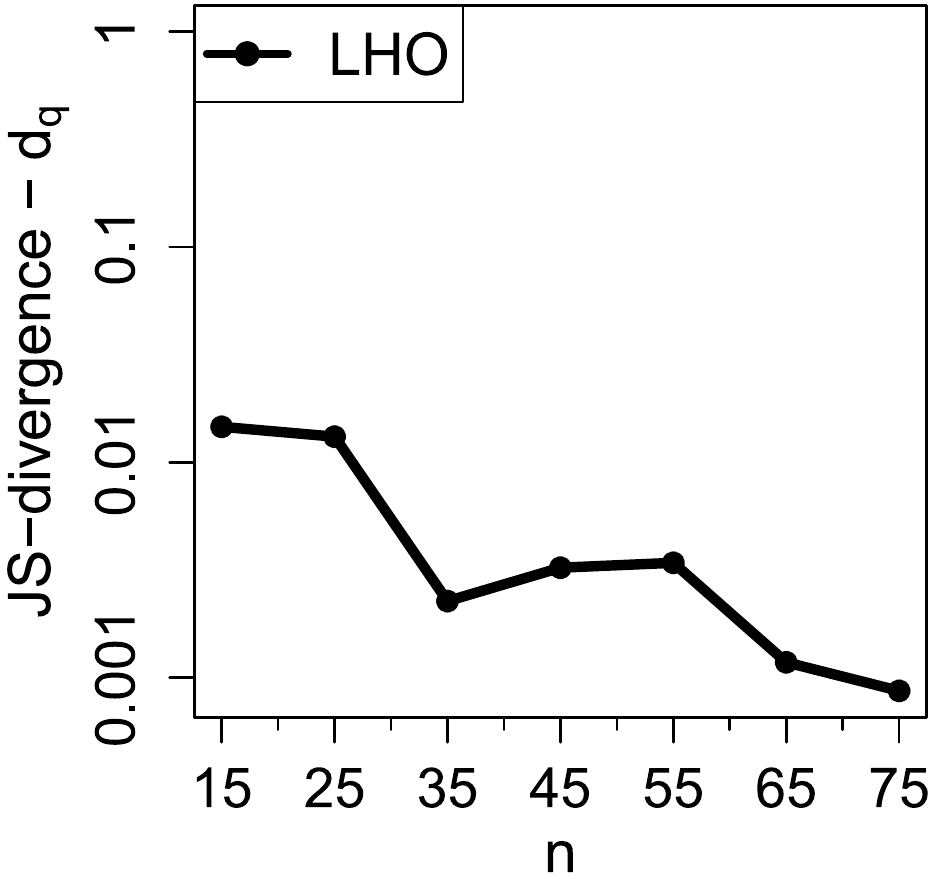}
		\caption{NYC}\label{fig1a}
	\end{subfigure}\hspace{+1mm}
	\begin{subfigure}[b]{.24\textwidth}\centering
		\includegraphics[trim={0 0 0 0},scale=0.3,clip]{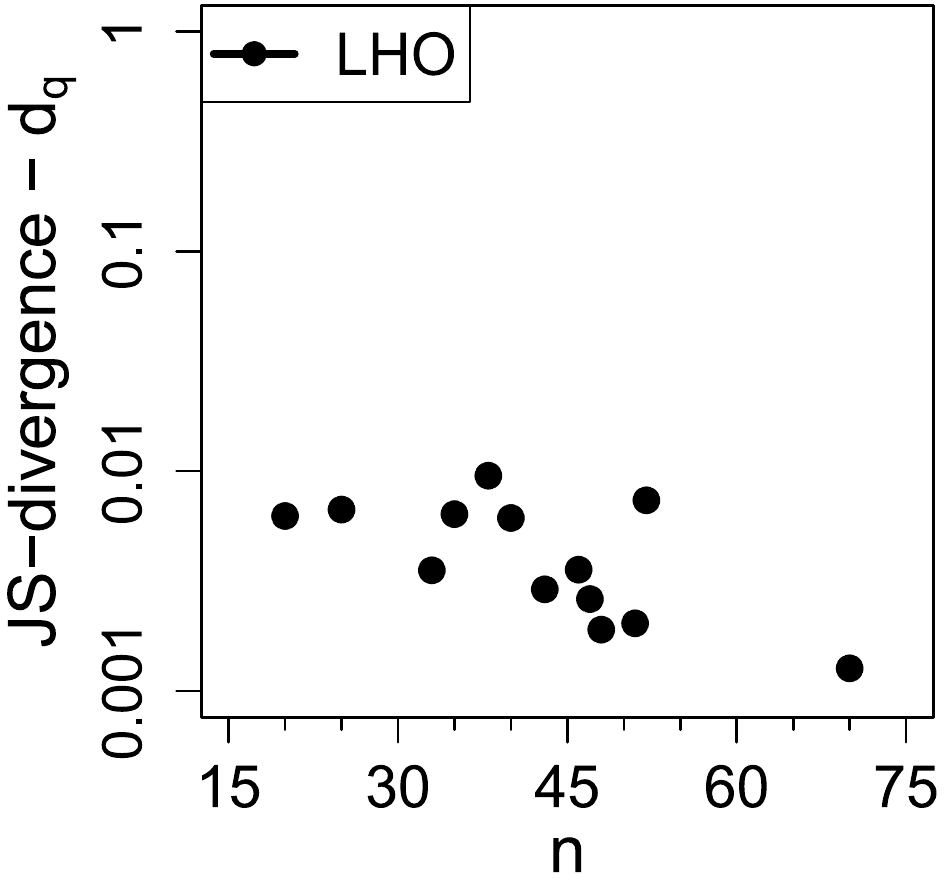}
		\caption{NYC}\label{fig1b}
	\end{subfigure}\hspace{+1mm}
	\begin{subfigure}[b]{.24\textwidth}\centering
		\includegraphics[trim={0 0 0 0},scale=0.3,clip]{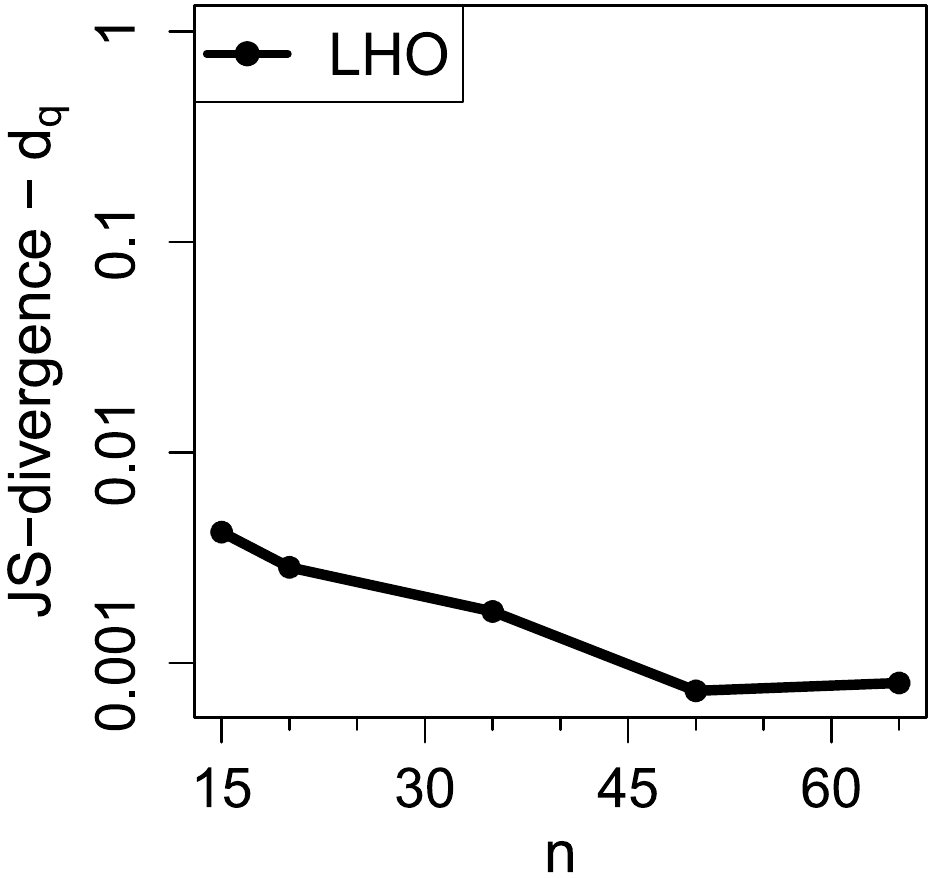}
		\caption{TKY}\label{fig1c}
	\end{subfigure}\hspace{+1mm}
	\begin{subfigure}[b]{.24\textwidth}\centering
		\includegraphics[trim={0 0 0 0},scale=0.3,clip]{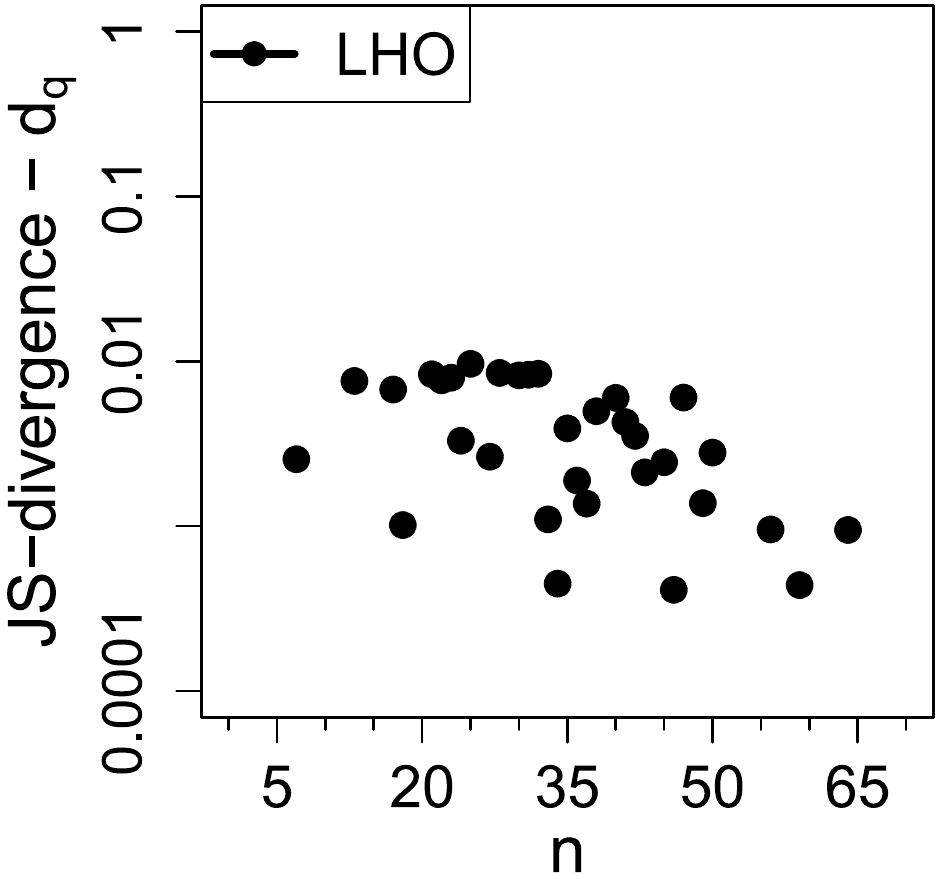}
		\caption{TKY}\label{fig1d}
	\end{subfigure}
	\caption{JS-divergence vs length $n$: (a) Median JS-divergence for histograms of length $n$ in \emph{NYC}. (b) JS-divergence for each histogram with $K=20$ in \emph{NYC}. (c)
		Median JS-divergence for histograms of length $n$ in \emph{TKY}. (d) JS-divergence for each histogram with $K=20$ in \emph{TKY}}\label{fig1}
\end{figure*}

\subsubsection{Quality preservation for the $LHO$ algorithm}

\noindent\paragraph{Impact of histogram length $n$} We show that JS-divergence decreases with $n$, in \figref{fig1} (the $y$ axis is in logarithmic scale).
This is because there are more bins whose counts may increase: The space considered by $LHO$ is larger and the change can be ``smoothed'' over more bins.
In addition, the JS-divergence scores are relatively low (recall that JS-divergence takes values in [0,1]).
This suggests that sanitization preserves the distribution of nonsensitive locations fairly well.

We also show that $NCE$ decreases with $n$, in \figref{fignew1}. This is because there are more bins whose frequency does not change, or it changes slightly so that they are not moved into a different cluster.
In addition, the $NCE$ scores are relatively low (recall that $NCE$ takes values in [0,1]). This suggests that the clustering quality is preserved fairly well after sanitization.

\begin{wrapfigure}[7]{l}{.48\textwidth}
	\begin{subfigure}[b]{.23\textwidth}\centering
		\includegraphics[trim={0 0 0 0},scale=0.25,clip]{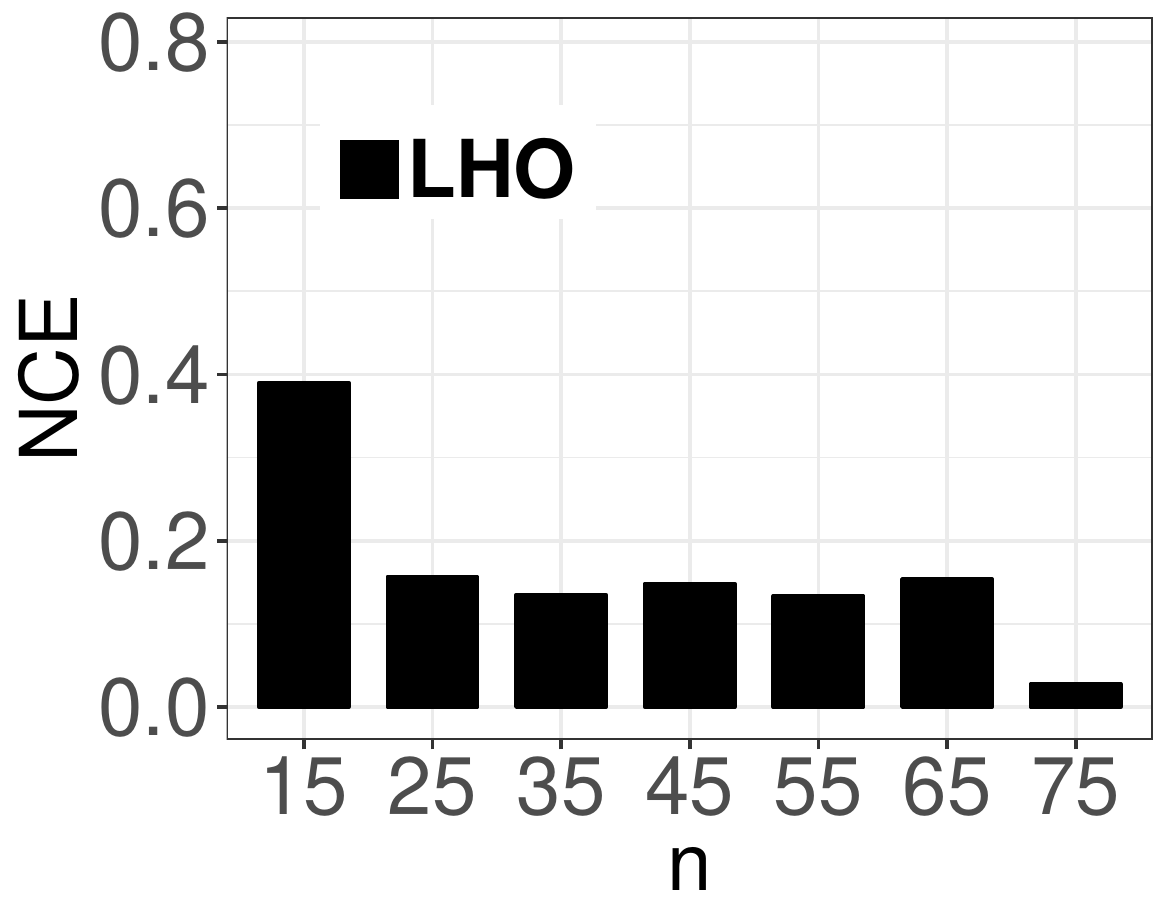}
		\caption{NYC}\label{fig2a}
	\end{subfigure}
	\begin{subfigure}[b]{.23\textwidth}\centering
		\includegraphics[trim={0 0 0 0},scale=0.25,clip]{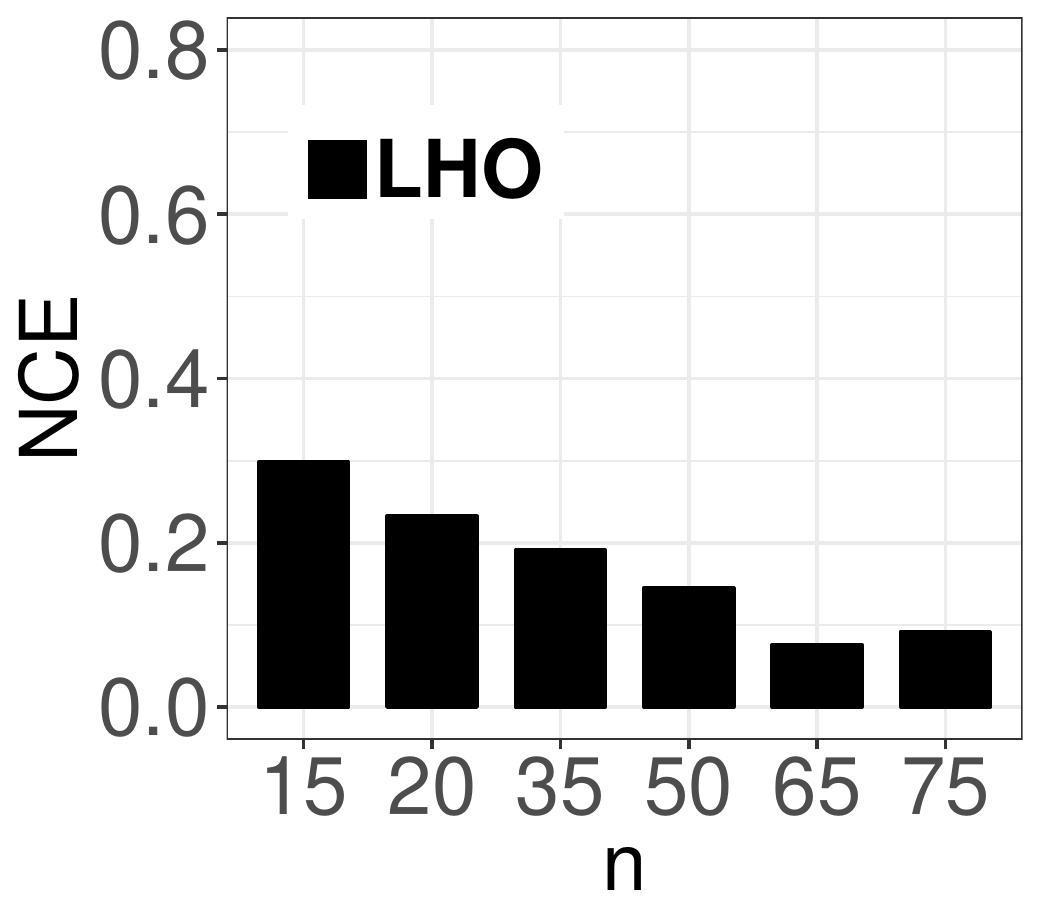}
		\caption{TKY}\label{fig2b}
	\end{subfigure}
	\caption{$NCE$ vs length $n$ for: (a) \emph{NYC}, and (b) \emph{TKY}}\label{fignew1}
\end{wrapfigure}

\noindent\paragraph{Impact of total frequency of sensitive locations $K$}
We show that JS-divergence increases with $K$ in \figref{fig2}
(the $y$-axes are in logarithmic scale).
This is because there are more counts that need to be redistributed into the bins of nonsensitive locations, and this incurs a larger amount of distortion.
In addition, the JS-divergence scores are relatively low. This suggests that the distribution of nonsensitive locations is preserved fairly well after sanitization.
\begin{figure*}[htbp]\hspace{-2mm}
	\begin{subfigure}[b]{.24\textwidth}\centering
		\includegraphics[trim={0 0 0 0},scale=0.3,clip]{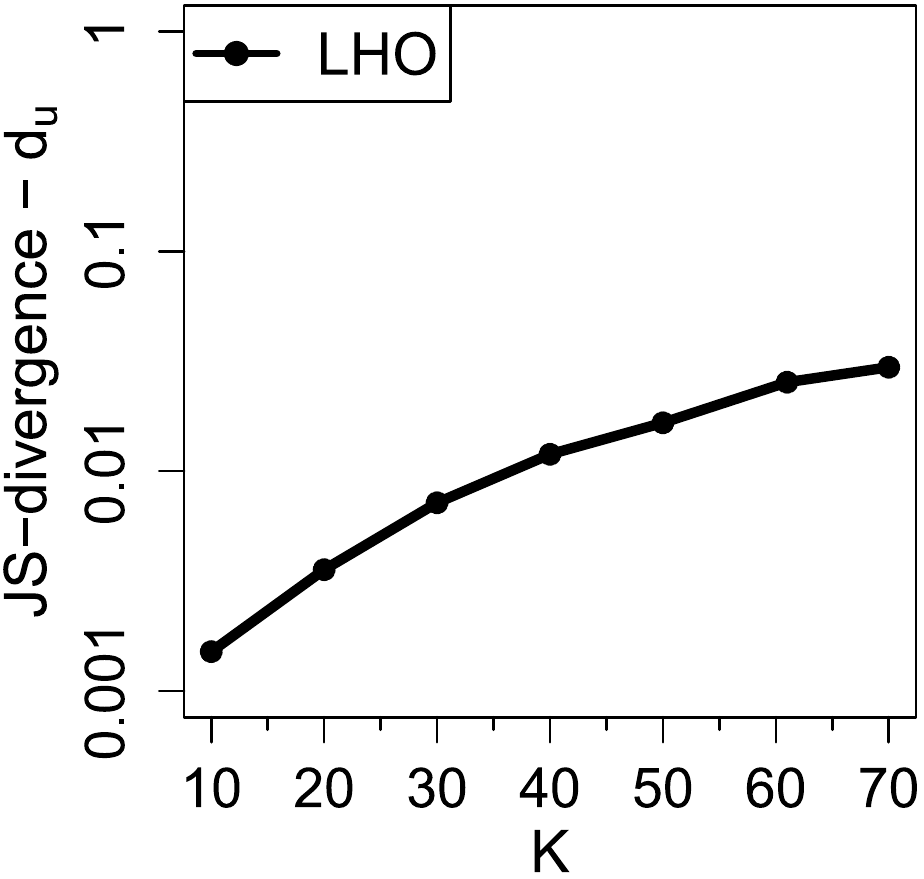}
		\caption{NYC}\label{fig2c}
	\end{subfigure}\hspace{+1mm}
	\begin{subfigure}[b]{.24\textwidth}\centering
		\includegraphics[trim={0 0 0 0},scale=0.3,clip]{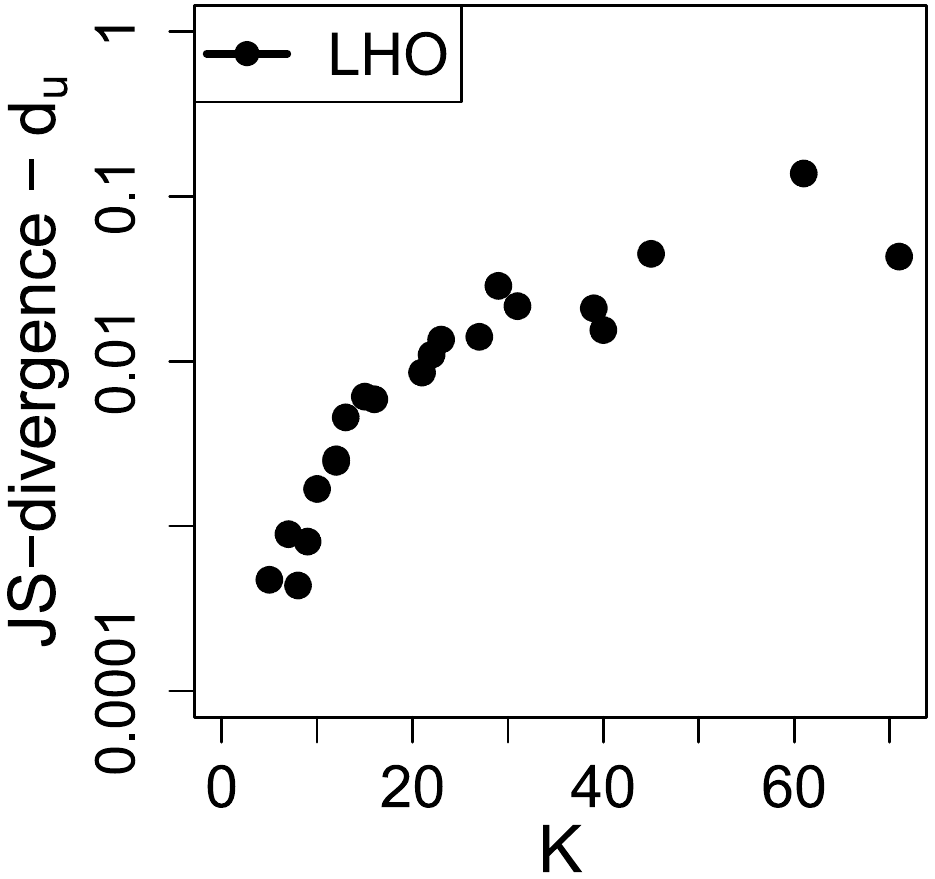}
		\caption{NYC}\label{fig2d}
	\end{subfigure}\hspace{+1mm}
	\begin{subfigure}[b]{.24\textwidth}\centering
		\includegraphics[trim={0 0 0 0},scale=0.3,clip]{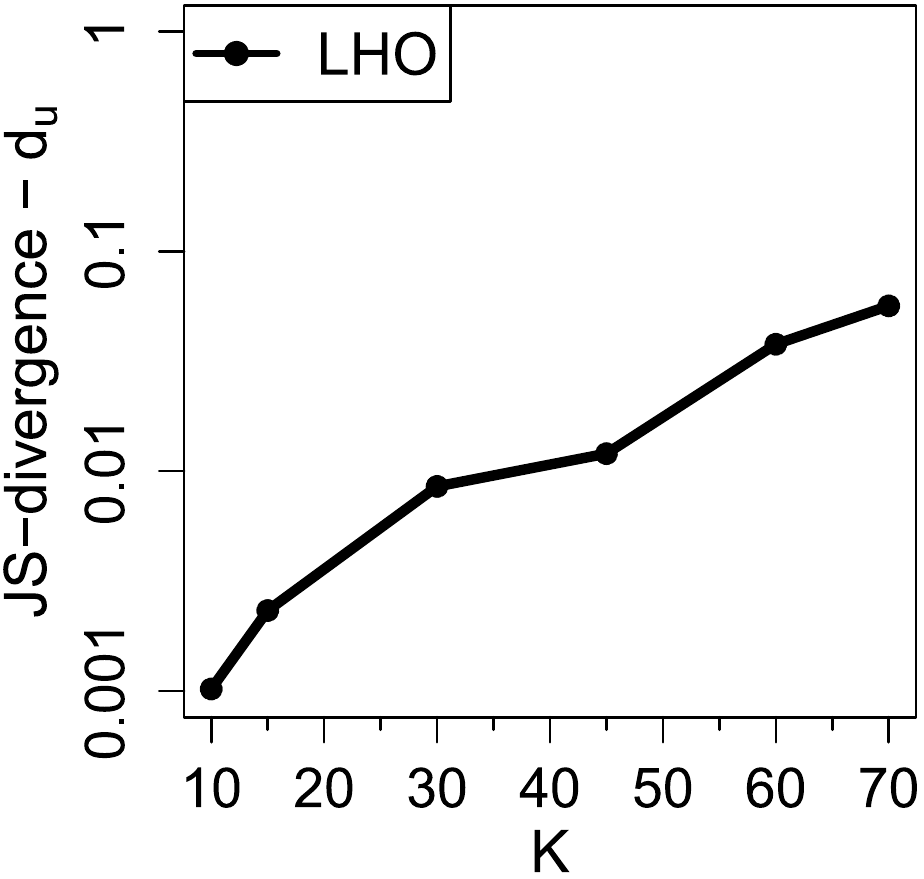}
		\caption{TKY}\label{fig3a}
	\end{subfigure}\hspace{+1mm}
	\begin{subfigure}[b]{.24\textwidth}\centering
		\includegraphics[trim={0 0 0 0},scale=0.3,clip]{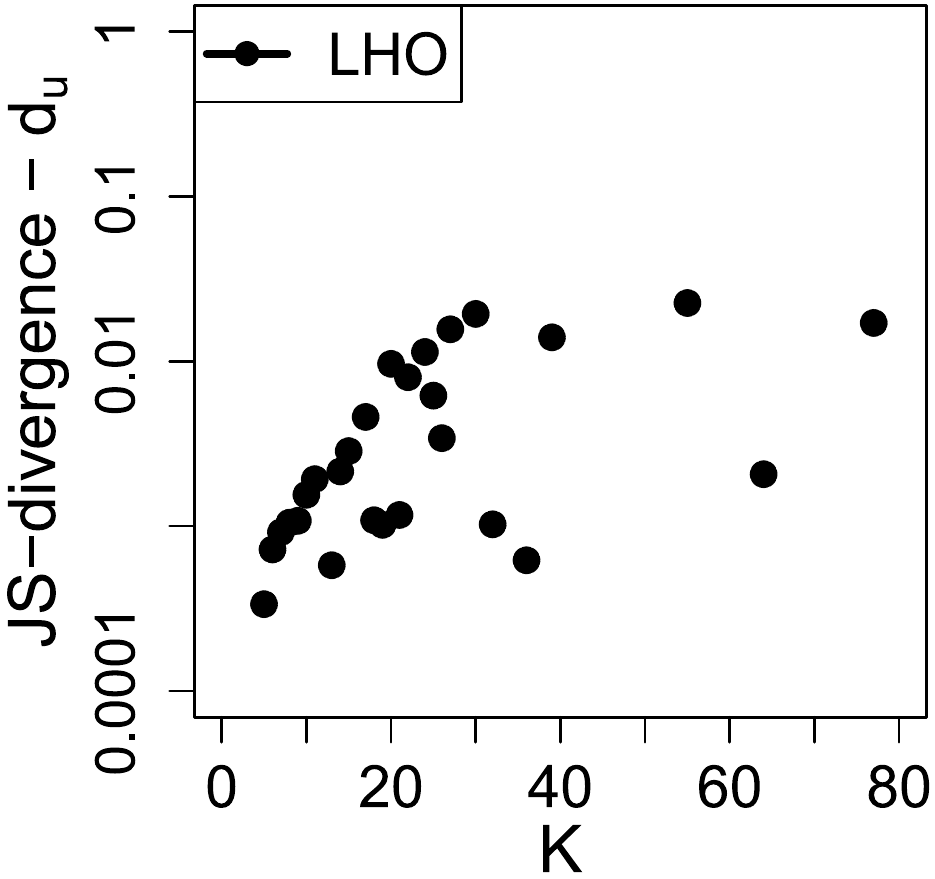}
		\caption{TKY}\label{fig3b}
	\end{subfigure}
	\caption{JS-divergence vs total frequency of sensitive locations $K$:
		(a) Median JS-divergence for varying $K$ in \emph{NYC}.
		(b) JS-divergence for each histogram with $n=30$ in \emph{NYC}.
		(c) Median JS-divergence for varying $K$ in \emph{TKY}.
		(d) JS-divergence for each histogram with $n=40$ in \emph{TKY}}\label{fig2}
\end{figure*}
\begin{wrapfigure}[10]{L}{.5\textwidth}
	\begin{subfigure}[b]{.24\textwidth}\centering
		\includegraphics[trim={0 0 0 0},scale=0.25,clip]{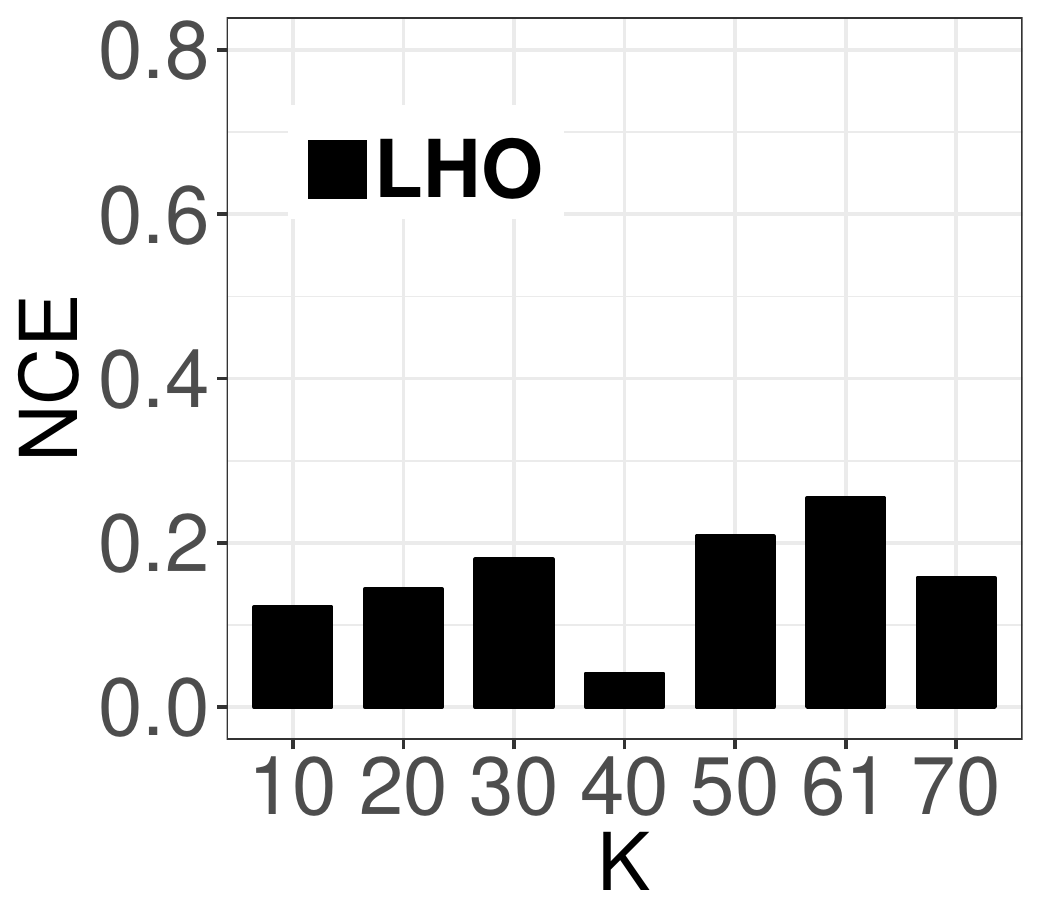}
		\caption{NYC}\label{fig3c}
	\end{subfigure}\hspace{+1mm}
	\begin{subfigure}[b]{.24\textwidth}\centering
		\includegraphics[trim={0 0 0 0},scale=0.25,clip]{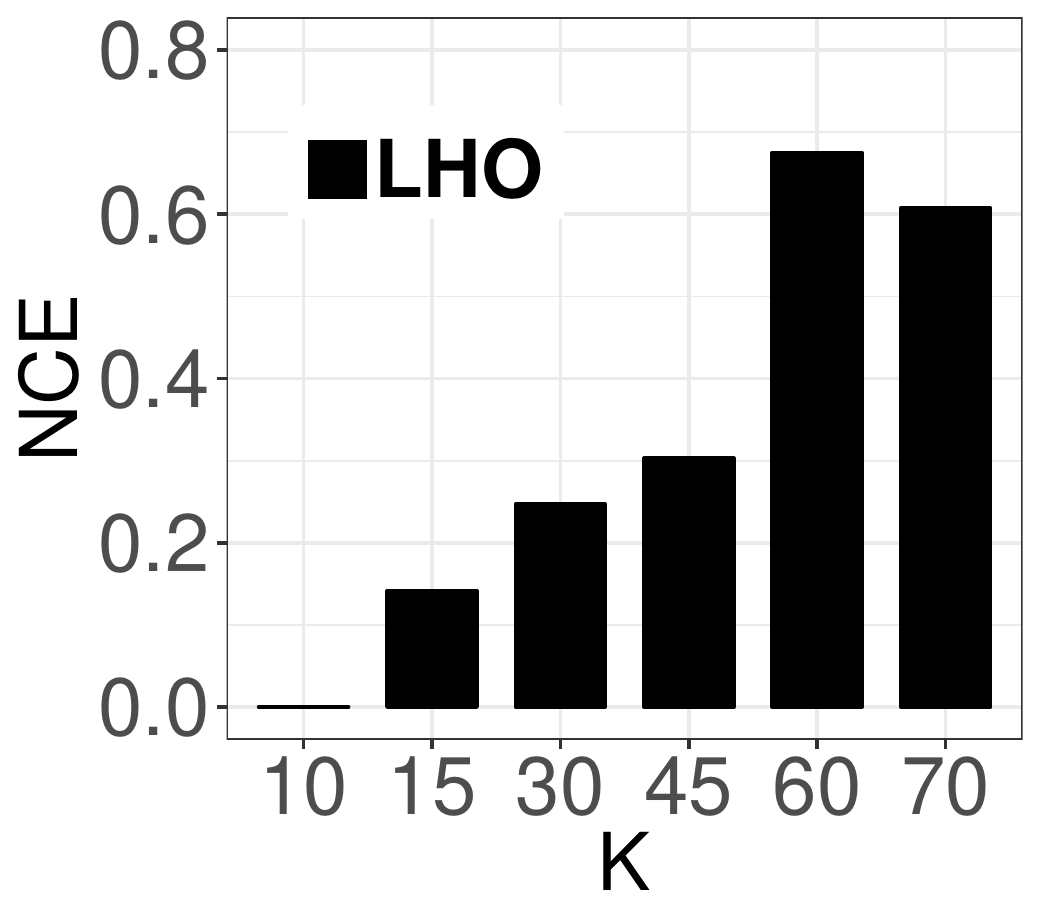}
		\caption{TKY}\label{fig3d}
	\end{subfigure}
	\caption{Median $NCE$ for varying $K$ vs total frequency of sensitive locations $K$ for: (c) \emph{NYC}, and (d) \emph{TKY}}\label{fignew2}
\end{wrapfigure}

We also show that $NCE$ increases with $K$, in \figref{fignew2}.  This is because a larger $K$ incurs larger changes to the frequency of the nonsensitive locations, which negatively impact the quality of
clustering. The $NCE$ scores are relatively low (on average $0.15$ and $0.25$ for \emph{NYC} and \emph{TKY}, respectively). This suggests that the clustering quality is preserved fairly well.
The high scores for $K\geq 60$ in the case of \emph{TKY} are obtained because the histograms have small total frequency (i.e., $K$ corresponds to $46\%$ of the total frequency of all locations on average). In this case,
the impact of sanitization on the histogram is inevitably large.

\begin{figure*}[ht!]\hspace{-2mm}
	\begin{subfigure}[b]{.24\textwidth}\centering
		\includegraphics[trim={0 0 0 0},scale=0.3,clip]{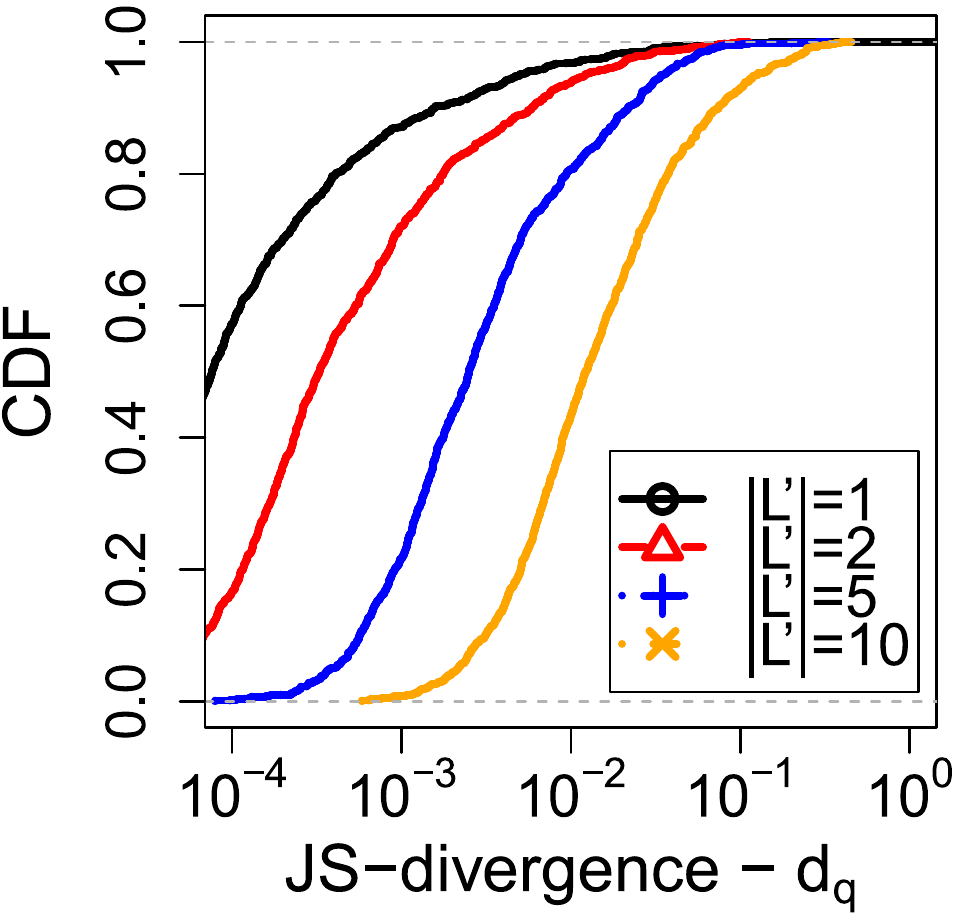}
		\caption{NYC}\label{fig4a}
	\end{subfigure}\hspace{+1mm}
	\begin{subfigure}[b]{.24\textwidth}\centering
		\includegraphics[trim={0 0 0 0},scale=0.3,clip]{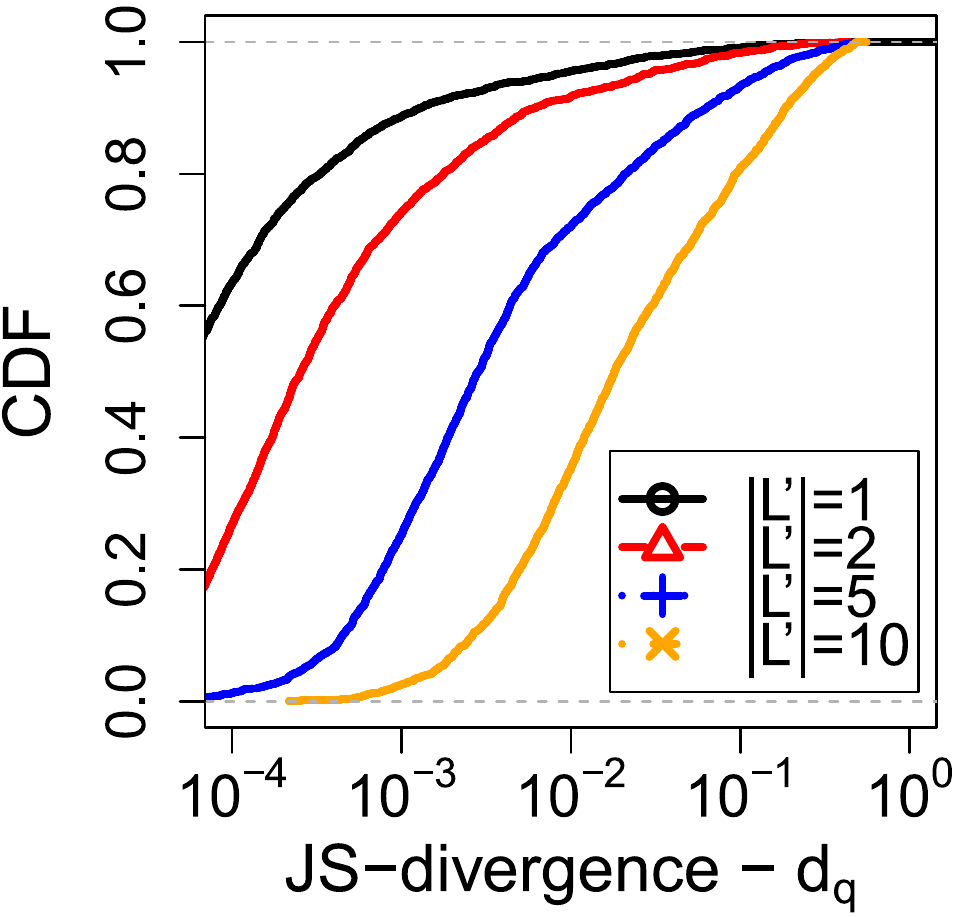}
		\caption{TKY}\label{fig4b}
	\end{subfigure}\hspace{+1mm}
	\begin{subfigure}[b]{.24\textwidth}\centering
		\includegraphics[trim={0 0 0 0},scale=0.3,clip]{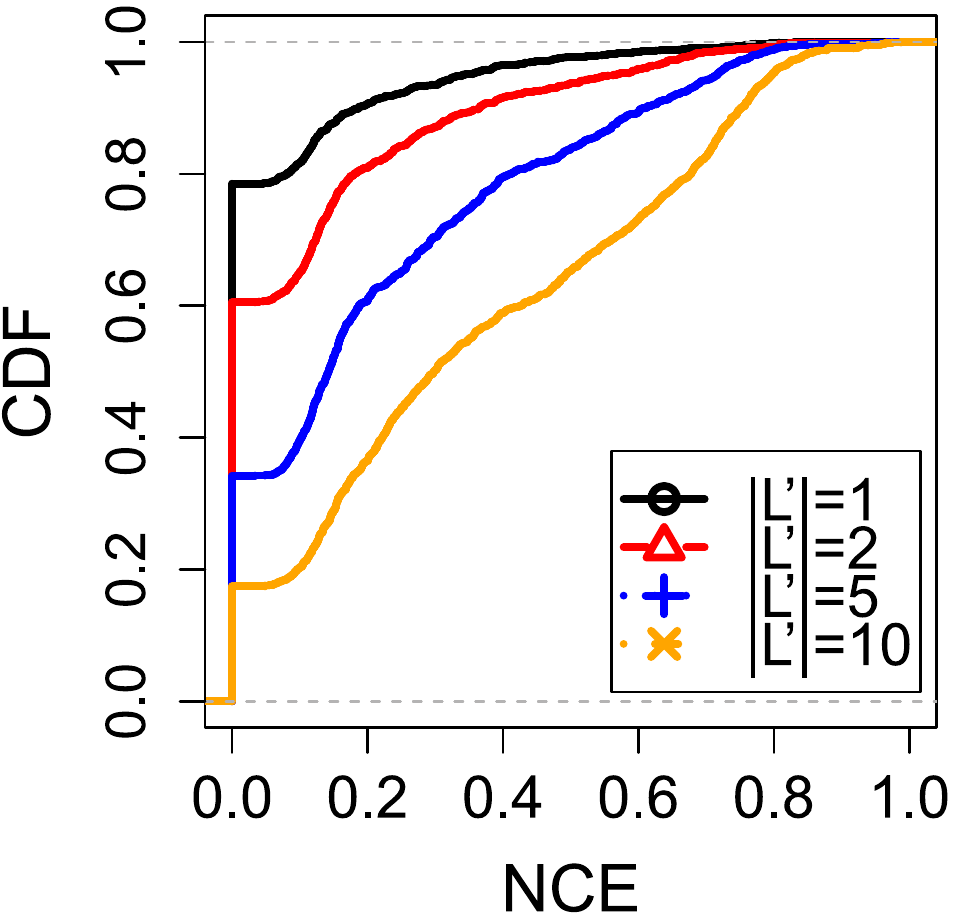}
		\caption{NYC}\label{fig4c}
	\end{subfigure}\hspace{+1mm}
	\begin{subfigure}[b]{.24\textwidth}\centering
		\includegraphics[trim={0 0 0 0},scale=0.3,clip]{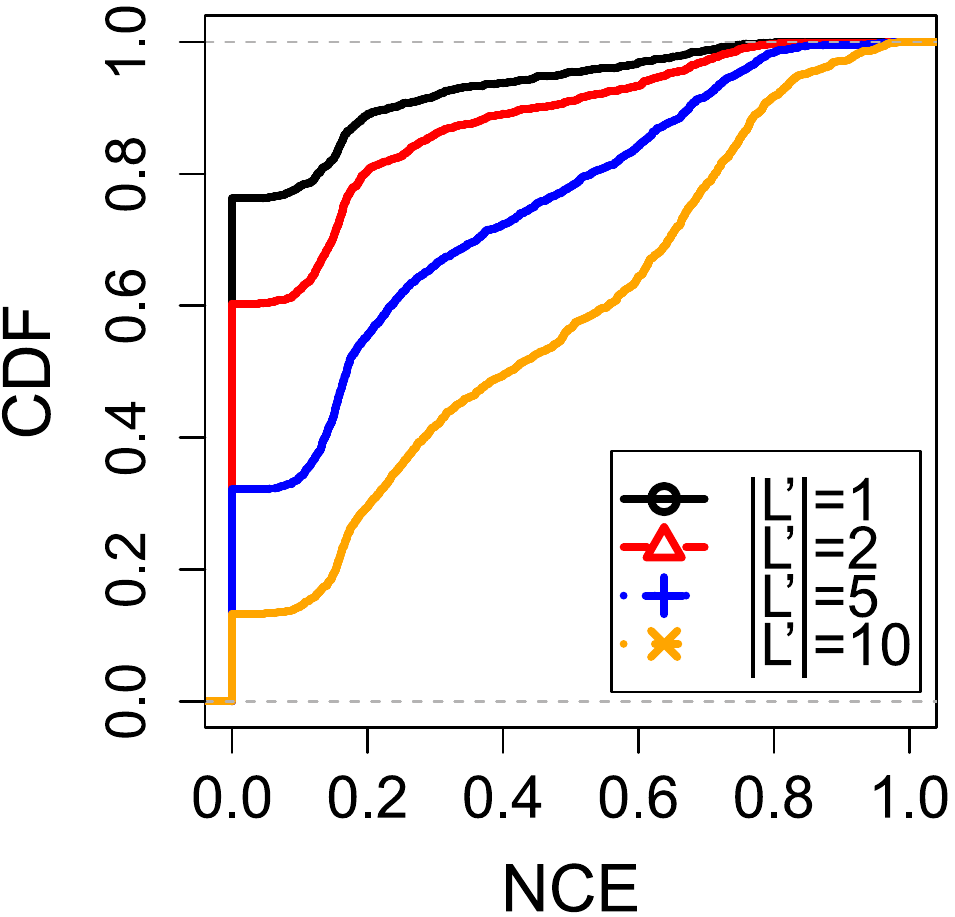}
		\caption{TKY}\label{fig4d}
	\end{subfigure}
	\caption{Cumulative Distribution Function of JS-divergence (i.e., ratio of histograms with JS-divergence at most equal to a score in $x$ axis) for varying $|L'|$ for: (a) \emph{NYC}, and (b) \emph{TKY}.
		Cumulative Distribution Function of $NCE$ for varying $|L'|$ for: (c) \emph{NYC}, and (b) \emph{TKY}}\label{fig4}
\end{figure*}

\vspace{-1mm}
\noindent\paragraph{Impact of number of sensitive locations $|L'|$} We show that JS-divergence increases with $|L'|$, in \figsref{fig4a} and \ref{fig4b}. This is because there are (I) more counts that need to be redistributed
into the bins of the nonsensitive locations, and (II) fewer bins to which the counts may be redistributed into, and, as demonstrated above, both (I) and (II) increase the amount of distortion incurred by sanitization to the original histogram. For instance, the median value for JS-divergence increases from 7.6E-06 to 1.2E-02 when $|L|'$
increases from 1 to 10, in the case of the \emph{NYC} dataset. However, the distributions of the original histograms are preserved well even when $|L'|=10$, with
$90\%$ of them having a JS-divergence score of at most $0.07$. The remaining histograms have higher scores because on average $70\%$ of their locations are treated as sensitive.

\begin{wrapfigure}[10]{L}{.48\textwidth}
	\begin{subfigure}[b]{.23\textwidth}
		\includegraphics[trim={0 0 0 0},scale=0.25,clip]{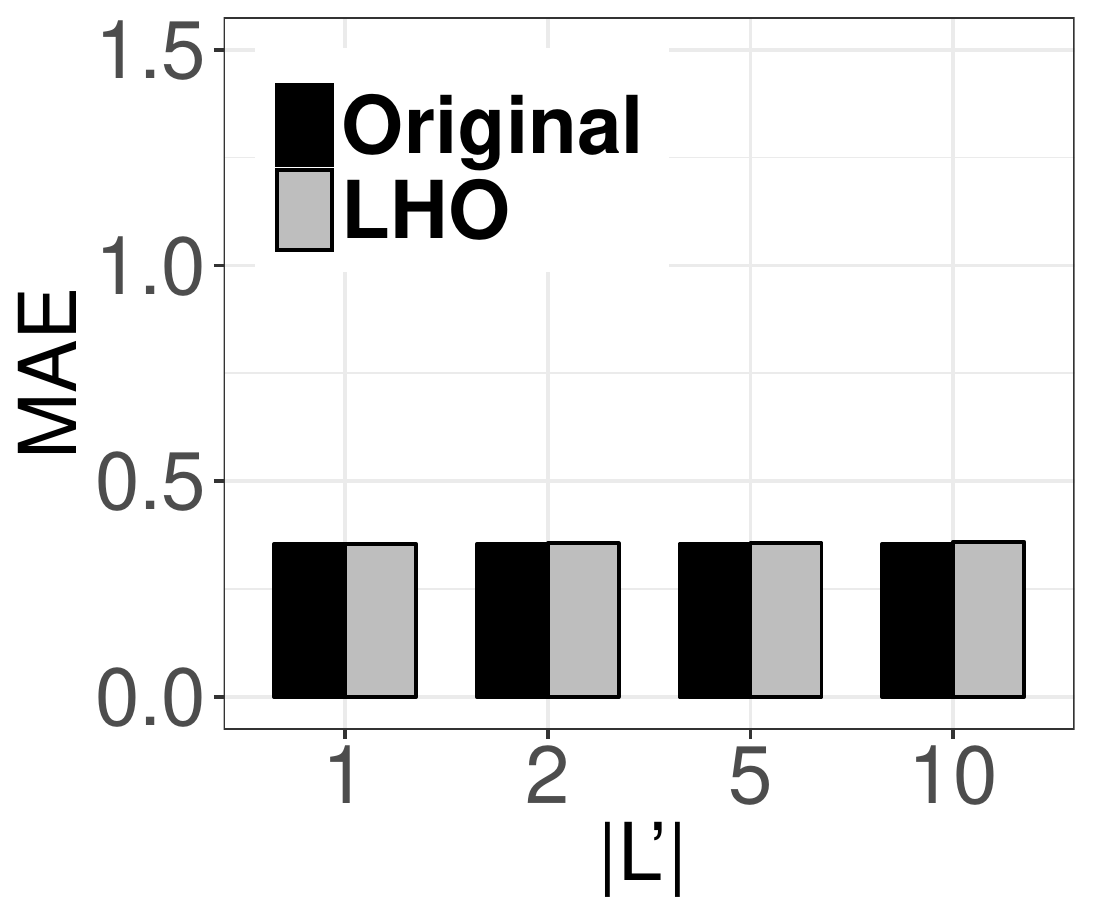}
		\caption{NYC}\label{recom1a}
	\end{subfigure}\hspace{+0.1mm}
	\begin{subfigure}[b]{.23\textwidth}
		\includegraphics[trim={0 0 0 0},scale=0.25,clip]{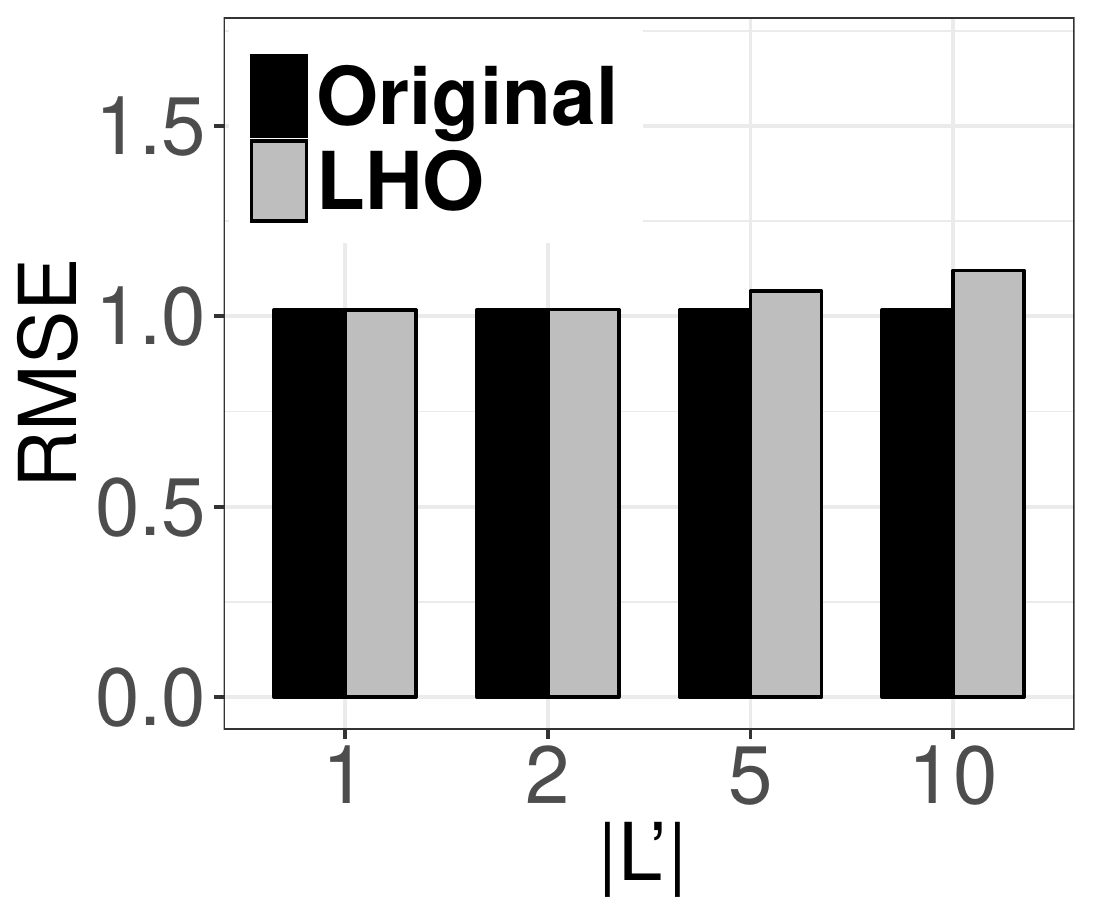}
		\caption{NYC}\label{recom1b}
	\end{subfigure}
	\caption{Recommendation quality for varying $|L'|$ with respect to: (a) $MAE$, and (b) $RMSE$}\label{recom1}
\end{wrapfigure}

We also show that $NCE$ increases with $|L'|$, in \figsref{fig4c} and \ref{fig4d}.
This is because a larger $|L'|$ causes larger changes to the frequency of the nonsensitive locations, which negatively impact the quality of clustering.
However, the $NCE$ scores are relatively low, which suggests that the clustering quality is preserved fairly well.
Specifically, sanitization does not affect at all the clustering quality (i.e., $NCE=0$) for approximately
$70\%$ of histograms when $|L'|=1$ and for approximately $20\%$ when $|L'|=10$. The median $NCE$ was $0$, $0$, $0.14$, and $0.3$ for $|L'|$ equal to $1$, $2$, $5$, and $10$, respectively.

\vspace{-1mm}
\noindent\paragraph{Recommendation quality} We investigate the impact of sanitization on recommendation quality by using test datasets of original histograms vs test datasets of histograms that are sanitized with different values of $|L'|$. Figures \ref{recom1a} and \ref{recom1b} show that $MAE$ and $RMSE$ are not substantially affected by sanitization, for all tested $|L'|$ values. The change in $MAE$ and $RMSE$ is on average $0.1\%$ and $2.4\%$, respectively. This suggests that recommendation quality is preserved fairly well.

\subsubsection{Runtime performance for the $LHO$ algorithm}

We evaluate the runtime performance of $LHO$ as a function of (I) $n$, histogram length, (II) $K$, total frequency of sensitive locations, and (III) $|L'|$, number of sensitive locations. To isolate the effect of each parameter, we vary just one and keep the other two fixed. We then examine the joint impact of all three parameters, which is given by the time complexity formula
$O\left((n-|L'|)\cdot K^2\cdot \log((n-|L'|)\cdot K^2)\right)$, because we used Dijkstra's algorithm with binary heap to find shortest paths
(see Section \ref{LHOalgorithmsectin}). For brevity, we use $\lambda$ to denote $(n-|L'|)\cdot K^2\cdot \log((n-|L'|)\cdot K^2)$. Thus, we expect the runtime 
to be linear in $\lambda$.
\begin{figure*}[ht!]\hspace{-2mm}
	\begin{subfigure}[b]{.23\textwidth}\centering
		\includegraphics[scale=0.32,clip]{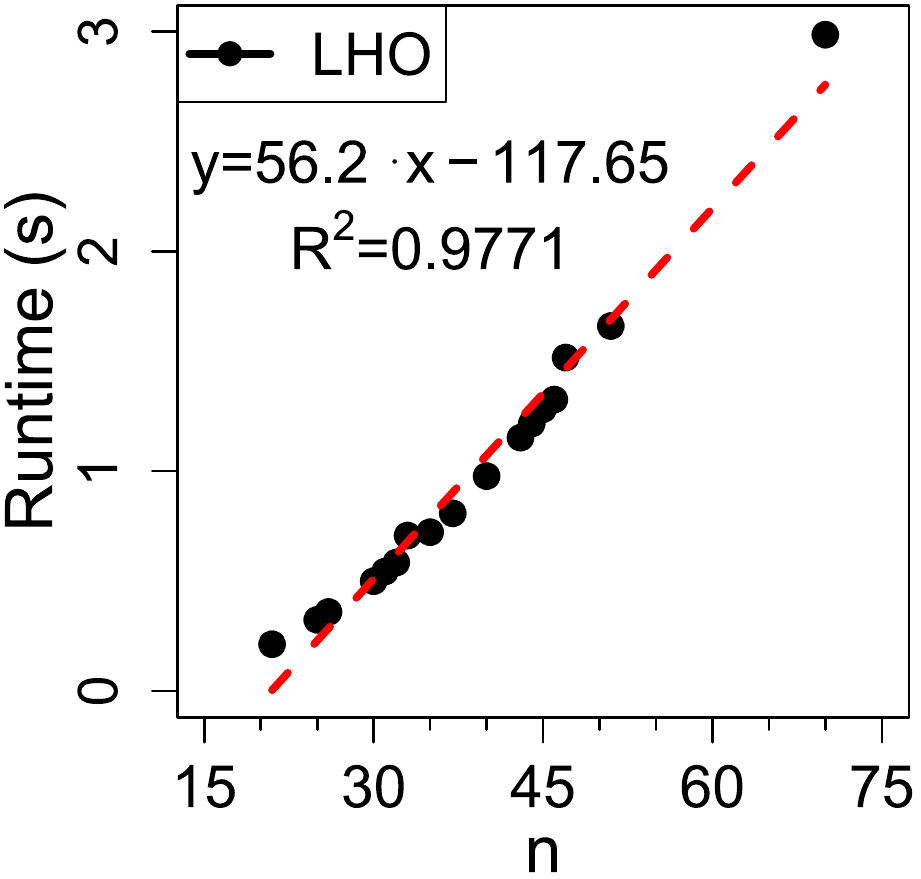}
		\caption{NYC}\label{fig5a}
	\end{subfigure}\hspace{+1mm}
	\begin{subfigure}[b]{.23\textwidth}\centering
		\includegraphics[scale=0.32,clip]{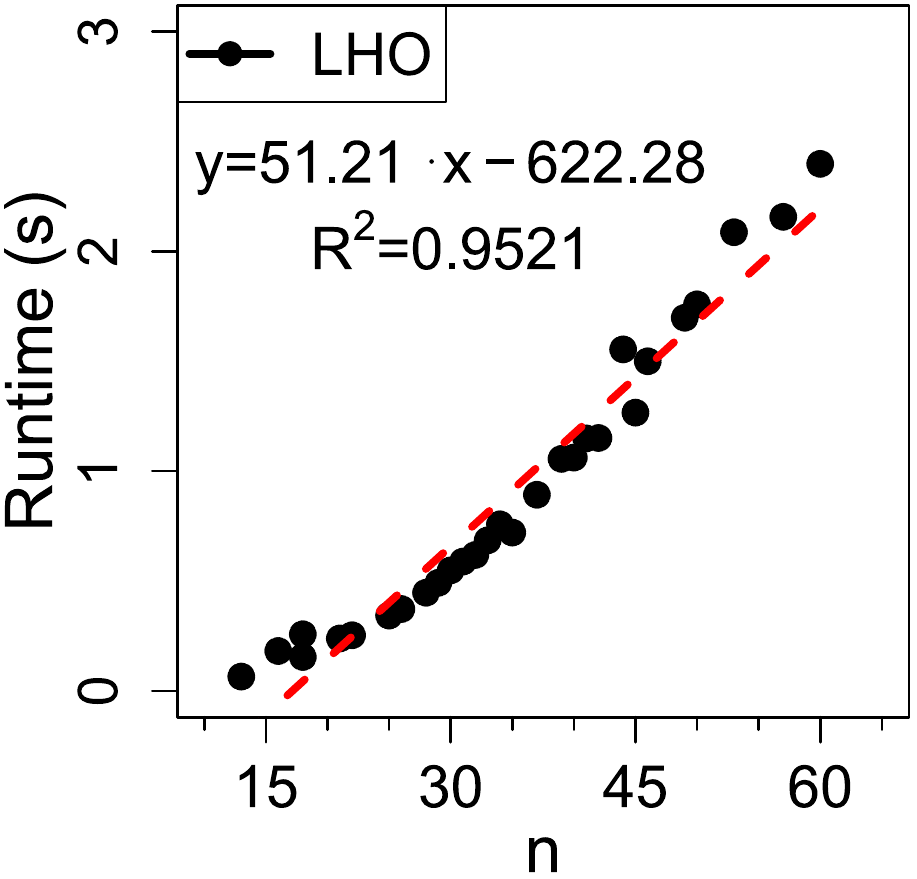}
		\caption{TKY}\label{fig5b}
	\end{subfigure}\hspace{+1mm}
	\begin{subfigure}[b]{.24\textwidth}\centering
		\includegraphics[scale=0.32,clip]{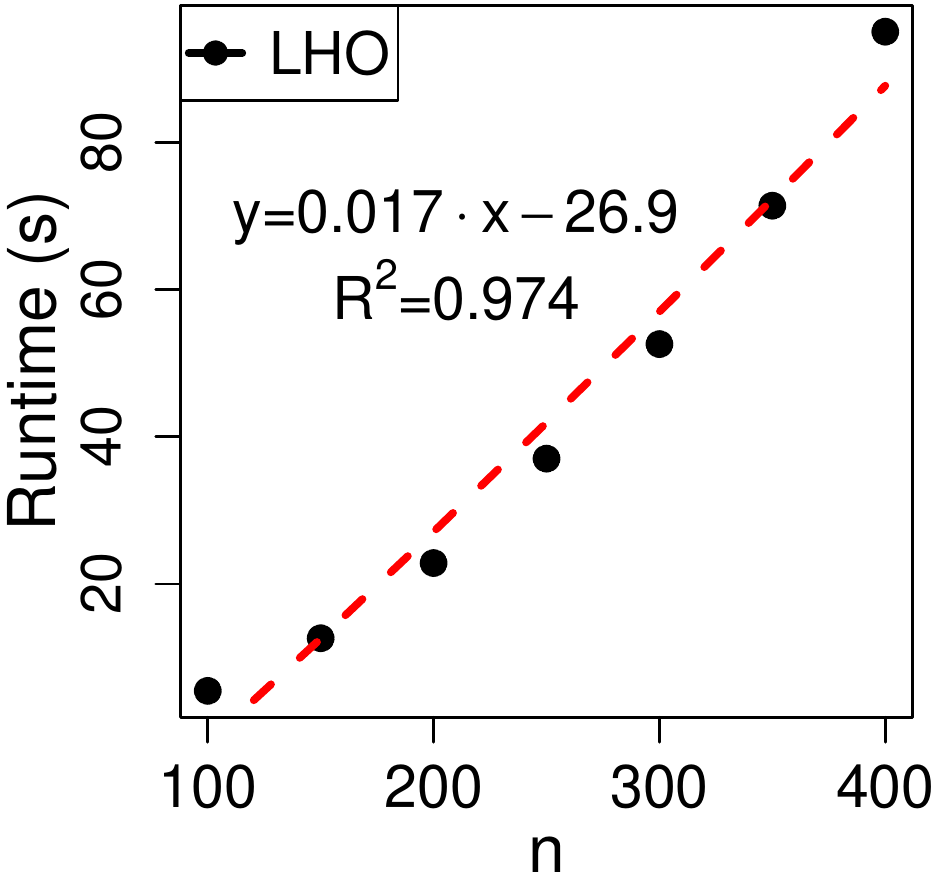}
		\caption{NYC}\label{figsyn1}
	\end{subfigure}\hspace{+1mm}
	\begin{subfigure}[b]{.24\textwidth}\centering
		\includegraphics[scale=0.32,clip]{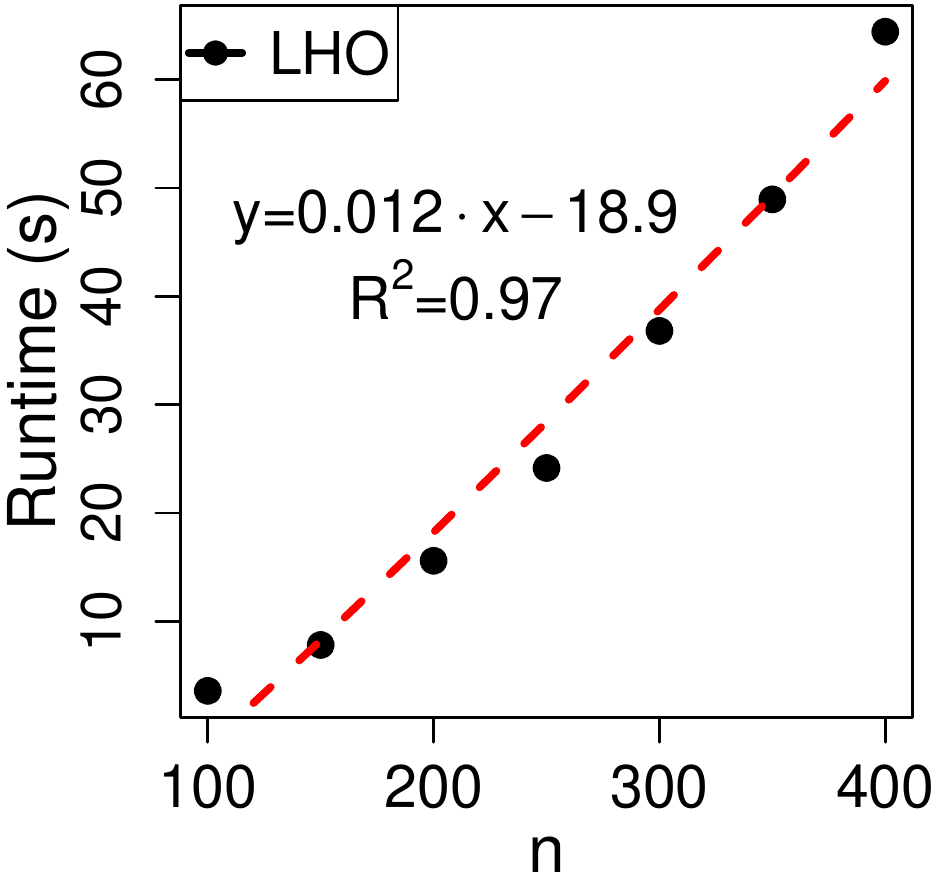}
		\caption{TKY}\label{figsyn2}
	\end{subfigure}
	\caption{Runtime vs length $n$, for each histogram with $K=20$ in: (a) \emph{NYC}, and (b) \emph{TKY}.
		Runtime vs length $n$, for synthetic histograms with varying $n$, $K=20$, and: (a) N = 192, (b) N = 642}\label{fig5}
\end{figure*}

\begin{wrapfigure}[13]{L}{.5\textwidth}
	\begin{subfigure}[b]{.24\textwidth}\centering
		\includegraphics[scale=0.32,clip]{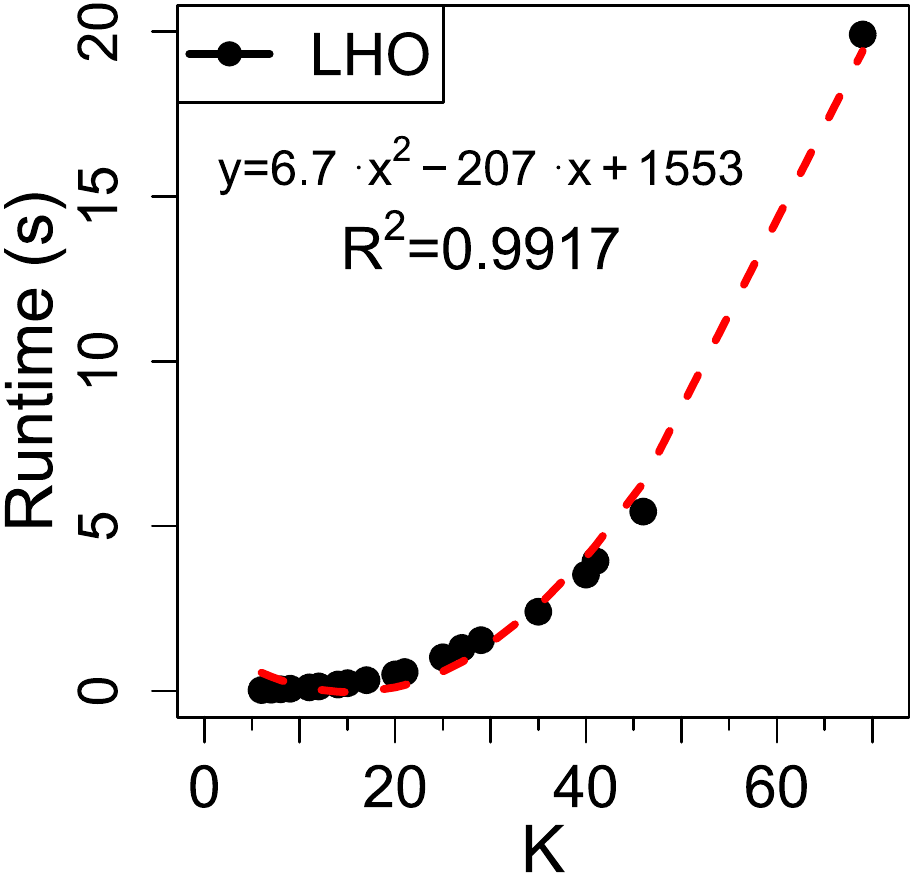}
		\caption{NYC}\label{fig5c}
	\end{subfigure}\hspace{+1mm}
	\begin{subfigure}[b]{.24\textwidth}\centering
		\includegraphics[scale=0.32,clip]{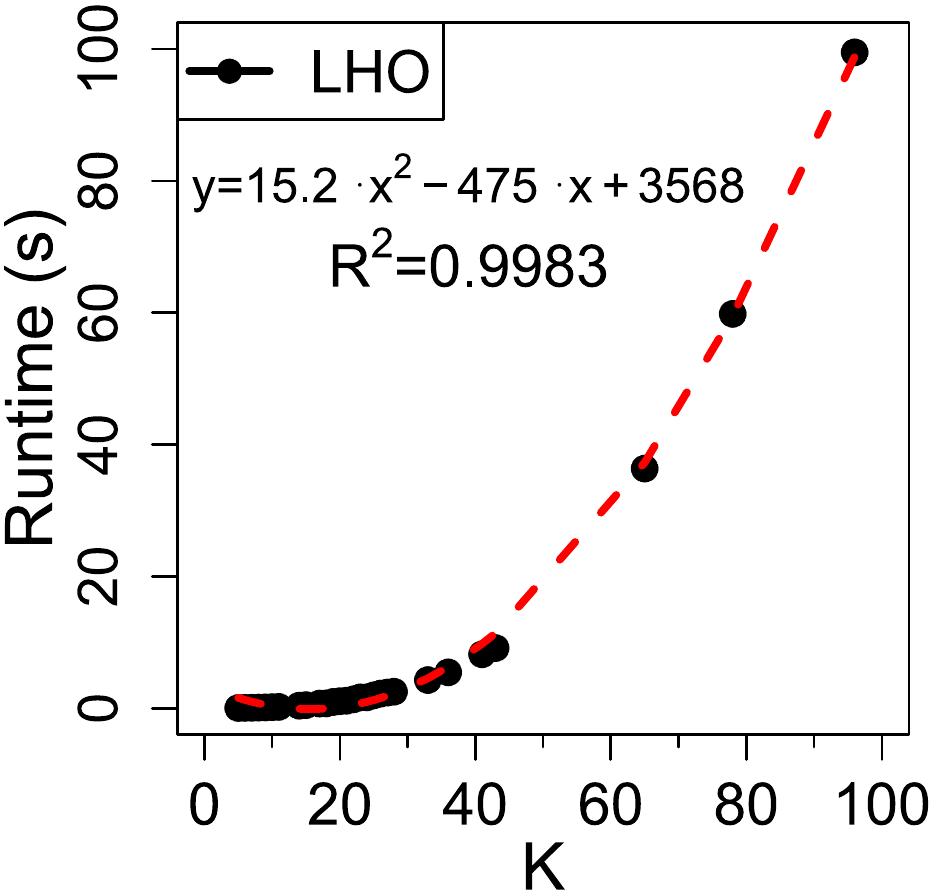}
		\caption{TKY}\label{fig5d}
	\end{subfigure}
	\vspace{-1mm}
	\caption{Runtime vs total frequency of sensitive locations $K$, for each histogram with: (a) $n=30$ in \emph{NYC}, and (b) $n=40$ in \emph{TKY}}\label{fig5vsK}
\end{wrapfigure}

\vspace{-1mm}
\noindent\paragraph{Impact of length $n$} We show that runtime increases with $n$, in \figsref{fig5a} and \ref{fig5b}. This is because, when $n$ is larger, there are more bins into which the counts may be redistributed.
More bins means that the multipartite graph $G_{TR}$, created by the $LHO$ algorithm, has more layers (and consequently more nodes and edges).
Note also that runtime increases linearly with $n$ (i.e., the linear regression models in \figsref{fig5a} and \ref{fig5b} are good fit), as expected by the time complexity analysis
(see Section \ref{LHOalgorithmsectin}), and that the algorithm took less than $3$ seconds. We also show that runtime increases linearly with $n$ when the algorithm is applied to the synthetic histograms, which are more demanding to sanitize (see \figsref{figsyn1} and \ref{figsyn2}).

\vspace{-1mm}
\noindent\paragraph{Impact of total frequency of sensitive locations $K$} We show that runtime increases with $K$, in \figsref{fig5c} and \ref{fig5d}. This is because there are more counts that are redistributed into the bins of
nonsensitive locations when $K$ is larger. That is, the graph $G_{TR}$ contains more edges and nodes. 
Note also that the runtime increases approximately quadratically with $K$ (i.e., the quadratic regression models in \figsref{fig5c} and \ref{fig5d} are good fit), as expected by the time complexity analysis (see Section \ref{LHOalgorithmsectin}), and that the algorithm took less than $100$ seconds.

\vspace{-1mm}
\noindent\paragraph{Impact of number of sensitive locations $|L'|$} We show that runtime increases with $|L'|$, in \figsref{fig6a} and \ref{fig6b}, which report results for each histogram in NYC and TKY, respectively.
This is because there are (I) more counts that need to be redistributed into the bins of the nonsensitive locations, and (II) fewer bins to which the counts may be redistributed to, and, as demonstrated above,
the impact of more counts on runtime is larger than that of fewer bins (quadratic increase vs linear decrease).
For example, $95\%$ of the histograms in the \emph{NYC} dataset take less than $1$ second to be sanitized
when $|L'| = 1$, but the corresponding percentage was $25\%$ when $|L'| = 10$.
However, the algorithm remains relatively efficient even for $|L'| = 10$, with $99\%$ of the histograms in \emph{NYC} requiring less than $5$ minutes to be sanitized.

\vspace{-1mm}
\noindent\paragraph{Joint impact of $n$, $K$, $|L'|$} In \figsref{fig6c} and \ref{fig6d}, we report results for all histograms in NYC and TKY, respectively.
Note that runtime increases linearly with $\lambda=(n-|L'|)\cdot K^2\cdot log((n-|L'|)\cdot K^2)$ (i.e., the linear regression models are good fit). This is in line with the time complexity analysis (see Section \ref{LHOalgorithmsectin}).
\begin{figure*}[ht!]\hspace{-2mm}
	\begin{subfigure}[b]{.24\textwidth}\centering
		\includegraphics[scale=0.32,clip]{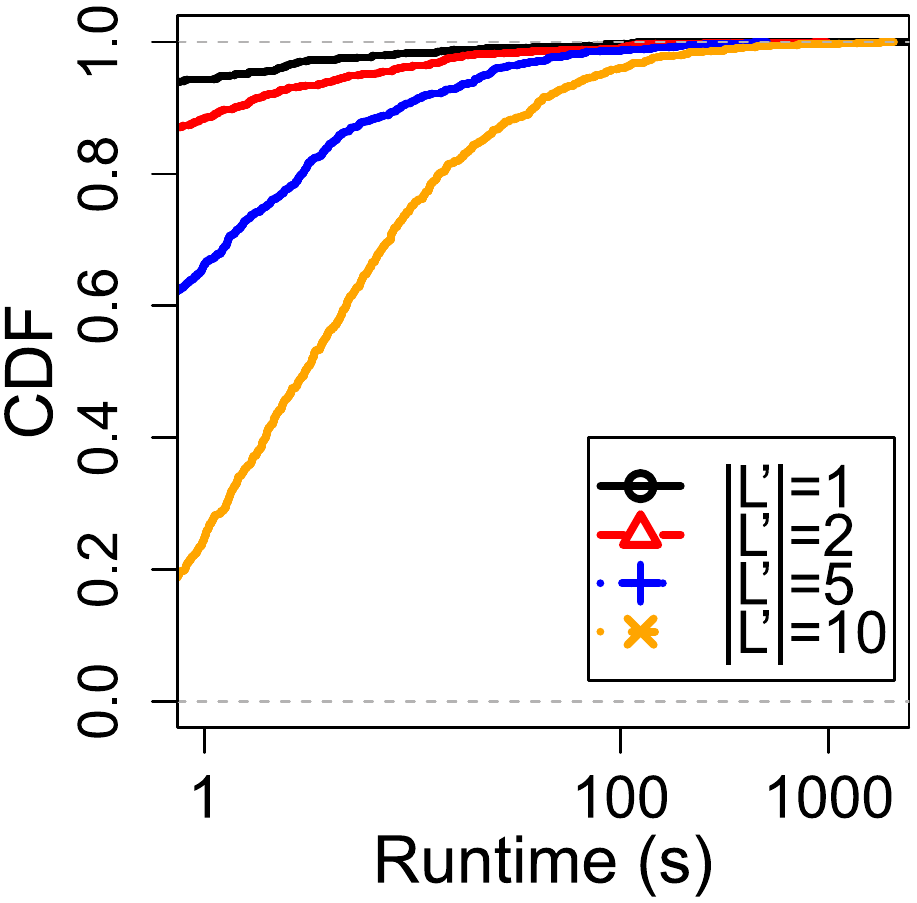}
		\caption{NYC}\label{fig6a}
	\end{subfigure}\hspace{+0.5mm}
	\begin{subfigure}[b]{.24\textwidth}\centering
		\includegraphics[scale=0.32,clip]{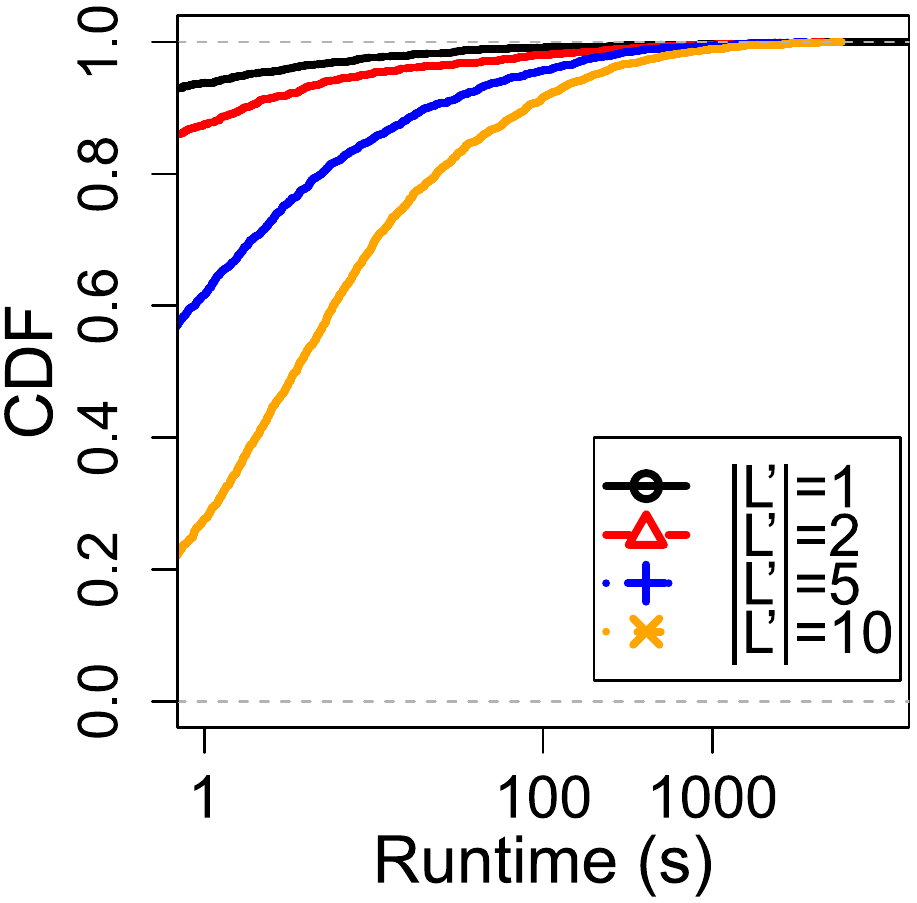}
		\caption{TKY}\label{fig6b}
	\end{subfigure}\hspace{+0.5mm}
	\begin{subfigure}[b]{.25\textwidth}\centering
		\includegraphics[scale=0.32,clip]{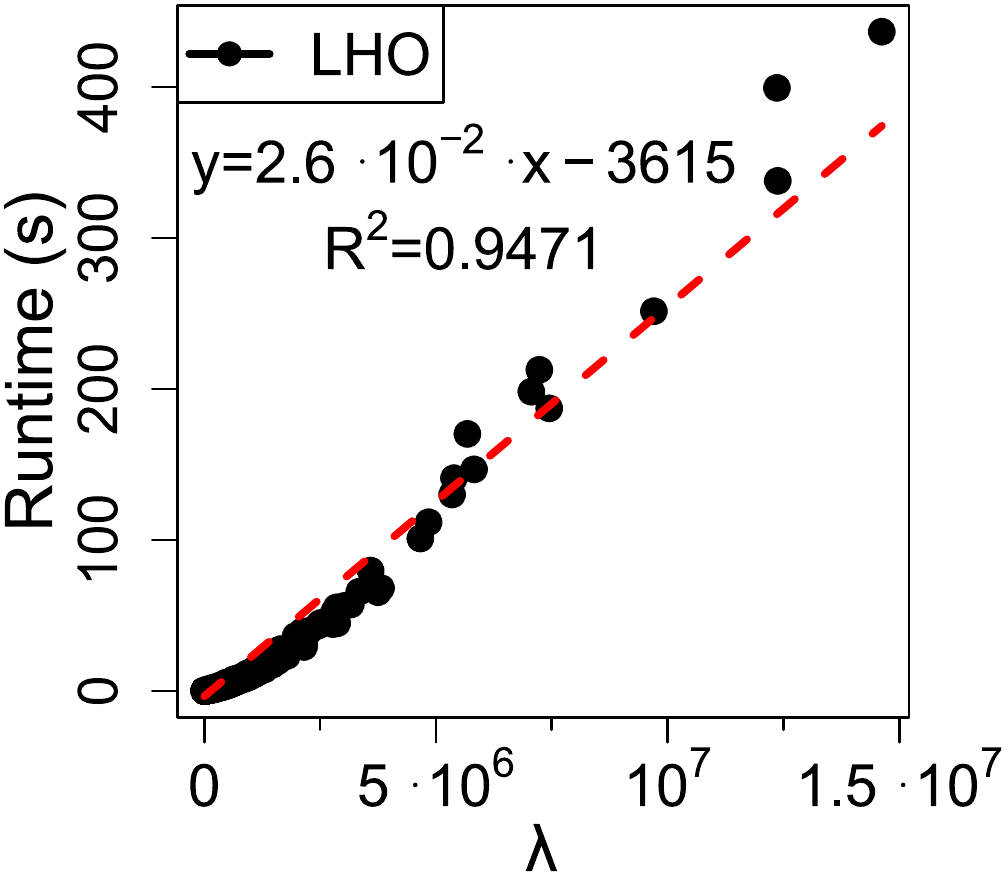}
		\caption{NYC}\label{fig6c}
	\end{subfigure}\hspace{+0.5mm}
	\begin{subfigure}[b]{.25\textwidth}\centering
		\includegraphics[scale=0.32,clip]{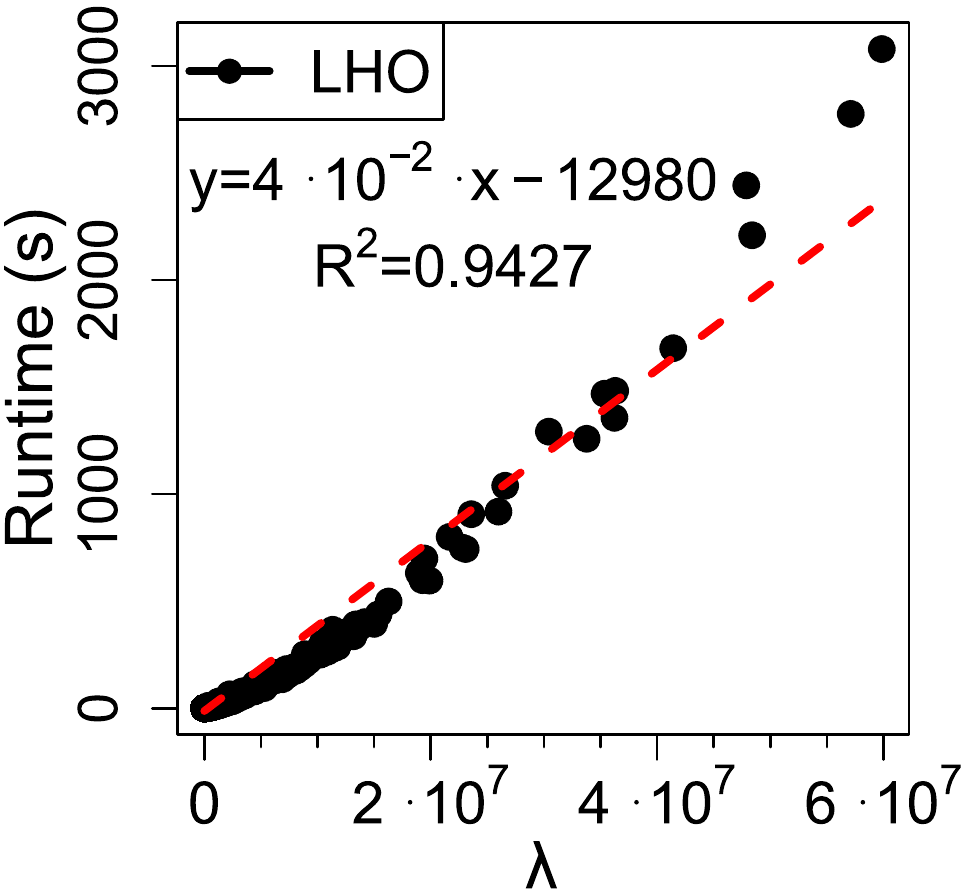}
		\caption{TKY}\label{fig6d}
	\end{subfigure}
	\caption{Cumulative Distribution Function of runtime (i.e., ratio of histograms with runtime at most equal to a score in $x$ axis) for varying $|L'|$ for: (a) each histogram in \emph{NYC}, and (b) each histogram in \emph{TKY}.
		Runtime vs $\lambda=(n-|L'|)\cdot K^2\cdot log((n-|L'|)\cdot K^2)$ (i.e., joint impact of $n$, $K$, $|L'|$ according to the time complexity analysis) for: (c) each histogram in \emph{NYC}, and (d) each histogram in \emph{TKY}}
\end{figure*}

\subsection{Target Resemblance}

We evaluate the quality and runtime of $RO$ and $RH$, as a function of (I) $n$, the length of original histogram, (II) $N$, the size (i.e., total frequency of locations) of the histogram, and
(III) $\epsilon$, the quality threshold. We additionally examine the impact of the target histogram $H''$ on the runtime, as well as the runtime of $RO$ and $RH$ when applied to histograms with large (up to the maximum possible) length.
To measure quality, we use JS-divergence, $L_2$ distance, and $NCE$. The results for the $L_2$ distance are in Appendix \ref{L2appendixRORH}, because they
are similar to those for JS-divergence. 
Unless stated otherwise, the target histogram for a histogram $H$ is a ``uniform'' histogram $H''$,
such that $H''$ has the same size, $N$, and length, $n$, as $H$, and each count of $H''$ is approximately equal to  $\frac{N}{n}$.
Aiming to resemble a uniform histogram indicates a user with strong privacy requirements, since the uniform distribution has the maximum entropy (i.e., provides the least information about the frequencies in $H$ to an attacker with no knowledge except $N$ and $n$). Moreover, uniform target histograms are difficult to resemble, because the original histograms typically follow skewed distributions.

\subsubsection{Quality and Privacy for the $RO$ algorithm and the $RH$ heuristic}

\noindent \paragraph{Impact of length $n$} To illustrate the impact of $n$ on quality and privacy, we present results obtained for randomly selected histograms of varying $n$ and $N=100$. We do not report the median of all histograms of certain $n$, because the results followed skewed distributions
(e.g., the runtime for histograms with $n=26$ and $N=100$ varied from $2.5$ to $45$ seconds).

We show that the privacy measure $d_p$ (JS-divergence) decreases with $n$, in \figsref{fig7a} and \ref{fig7b}.
This is because the larger number of bins gives more choices to the algorithm to reduce $d_p$ without substantially increasing $d_q$.
In addition, \figsref{fig7a} and \ref{fig7b} show that $RO$ and $RH$ achieve very similar results: the $d_p$ values for $RH$ were no more than $1\%$ and $2.1\%$ higher for \emph{NYC} and \emph{TKY}, respectively. This suggests that $RH$ is an effective heuristic.

\begin{figure*}[ht!]\hspace{-2mm}
	\begin{subfigure}[b]{.24\textwidth}\centering
		\includegraphics[scale=0.26,clip]{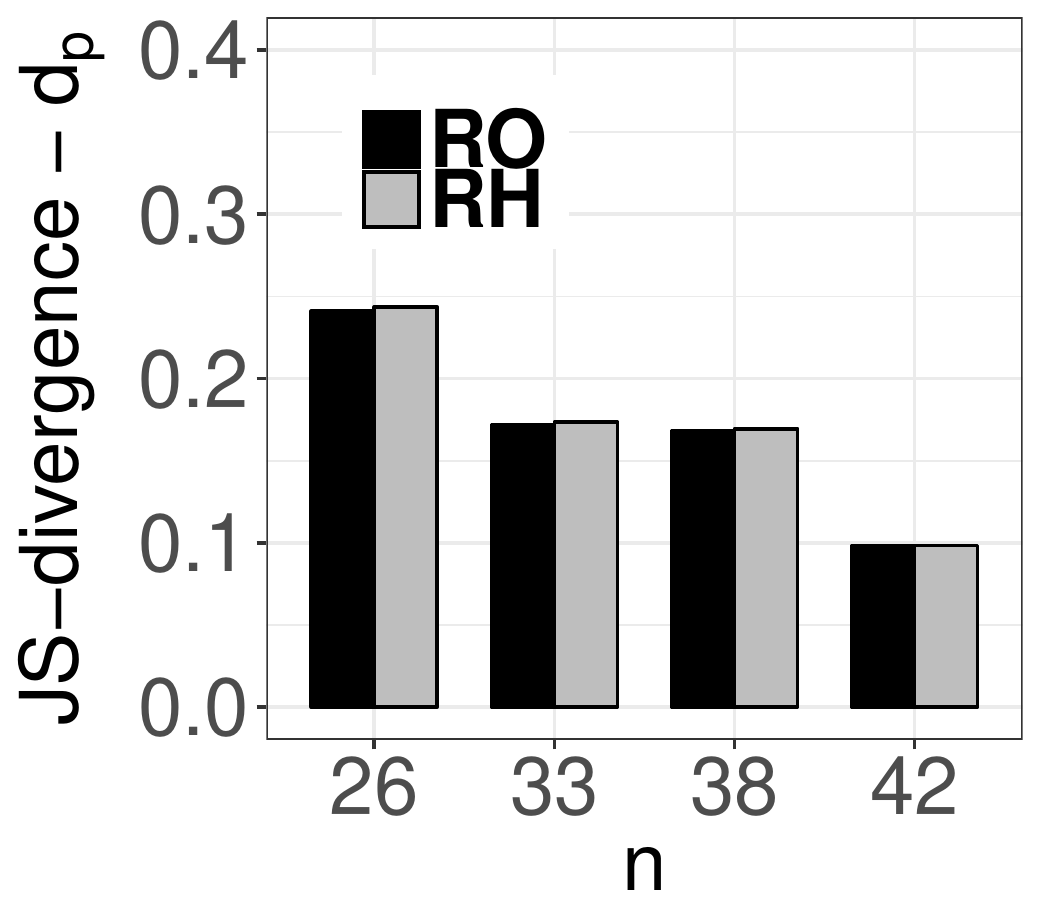}
		\caption{NYC}\label{fig7a}
	\end{subfigure}\hspace{+1mm}
	\begin{subfigure}[b]{.24\textwidth}\centering
		\includegraphics[scale=0.26,clip]{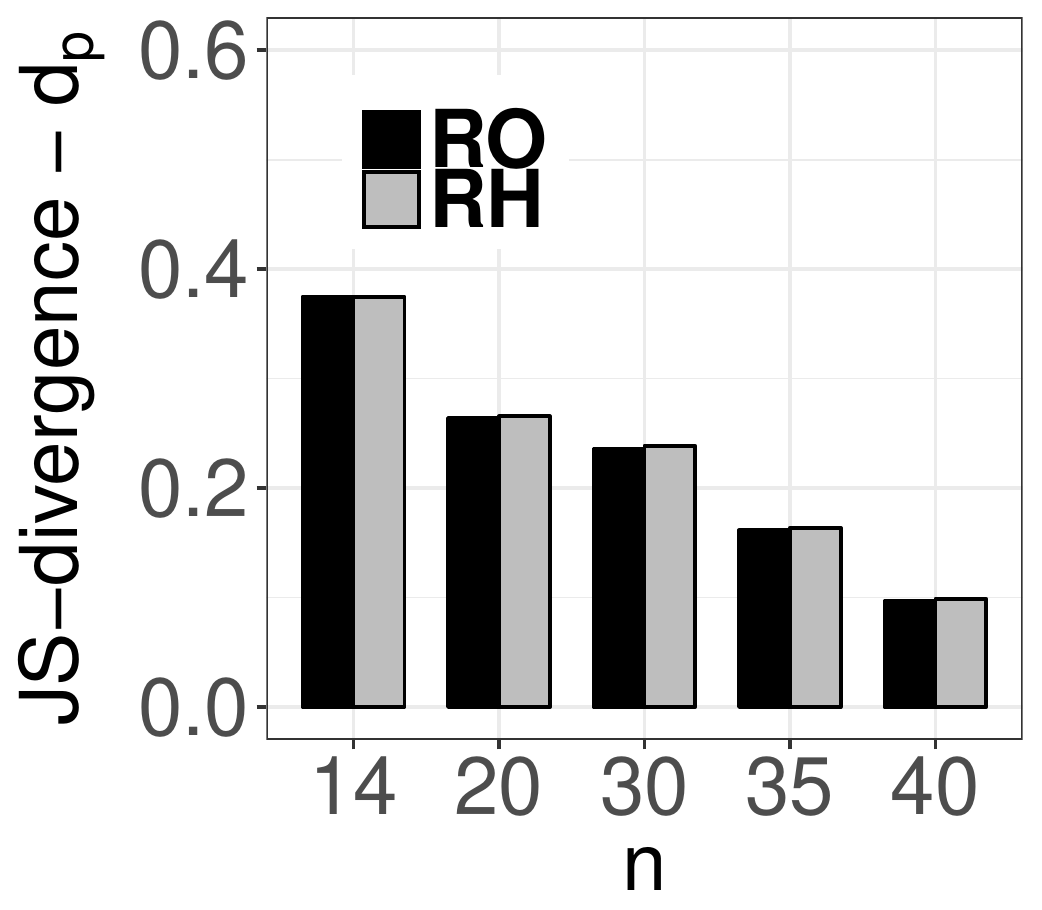}
		\caption{TKY}\label{fig7b}
	\end{subfigure}\hspace{+1mm}
	\begin{subfigure}[b]{.24\textwidth}\centering
		\includegraphics[scale=0.26,clip]{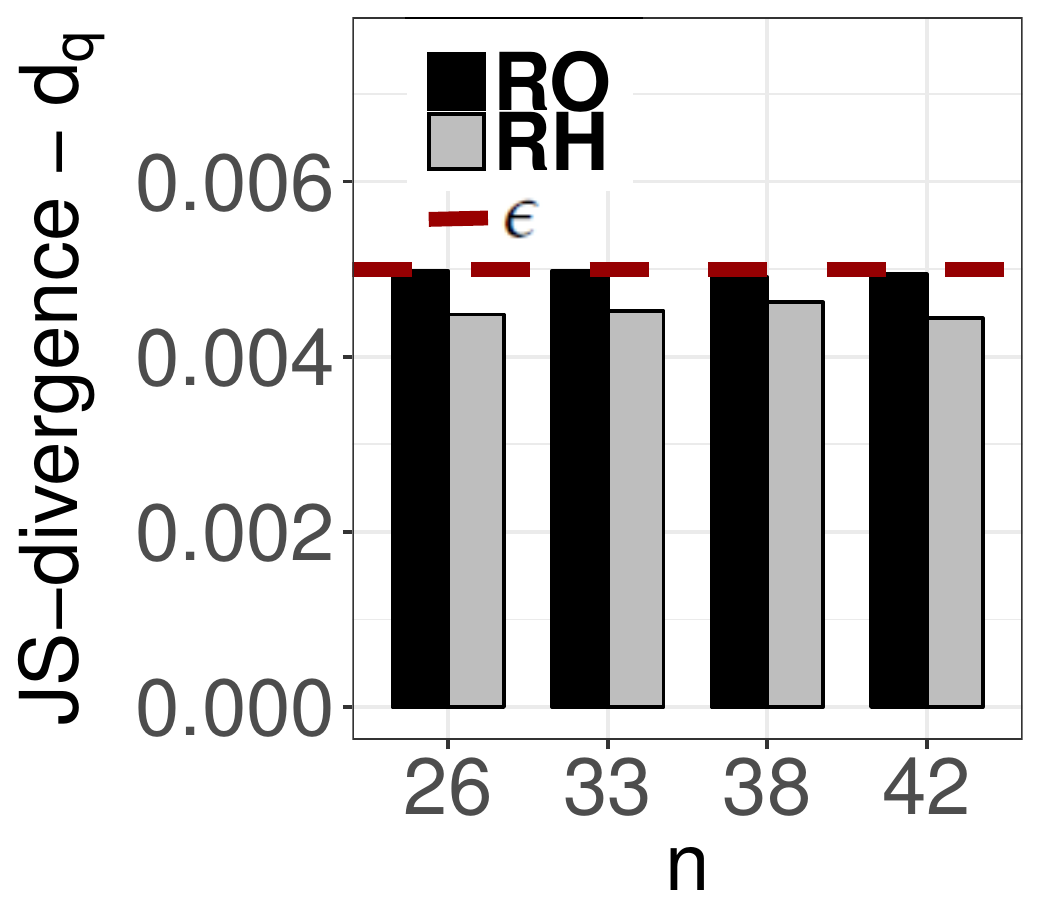}
		\caption{NYC}\label{fig7c}
	\end{subfigure}\hspace{+1mm}
	\begin{subfigure}[b]{.24\textwidth}\centering
		\includegraphics[scale=0.26,clip]{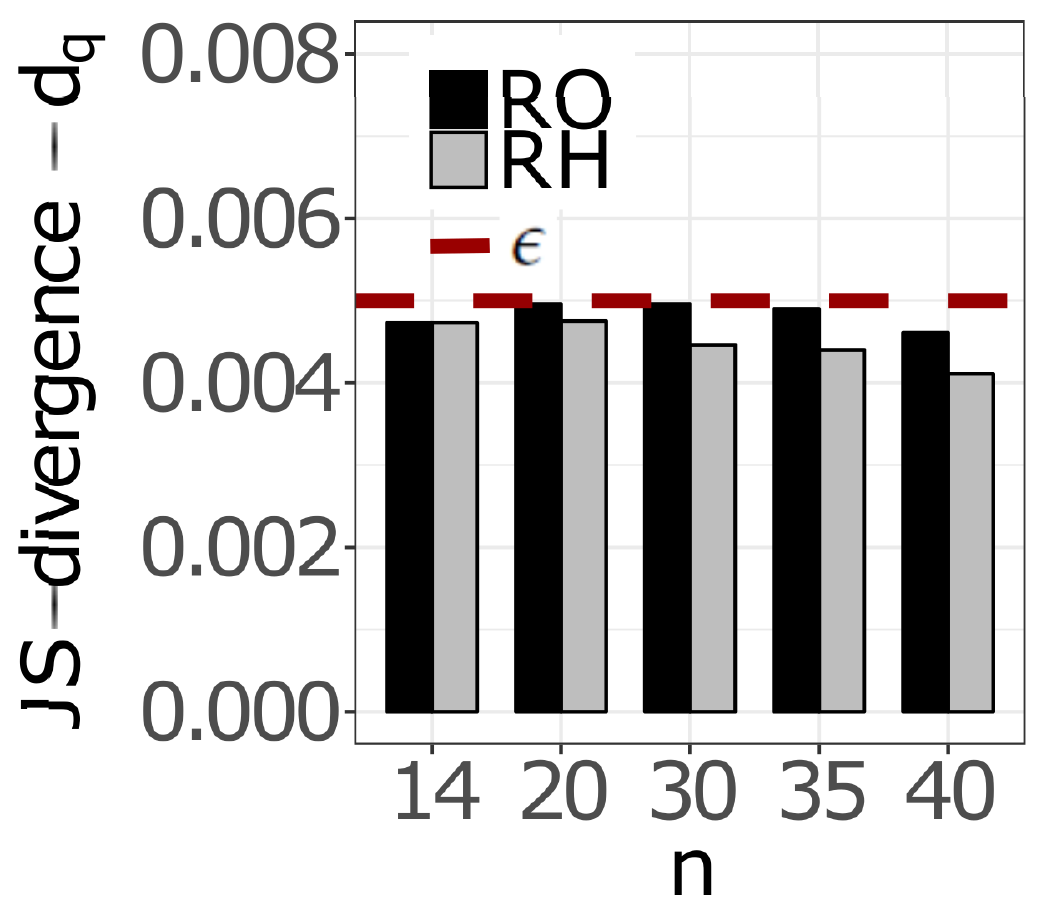}
		\caption{TKY}\label{fig7d}
	\end{subfigure}
	\caption{$d_p$ vs length $n$ for histograms with $N=100$ in: (a) \emph{NYC}. (b) \emph{TKY}. $d_q$ vs length $n$ for histograms with $N=100$ in: (c) \emph{NYC}. (d) \emph{TKY}}
\end{figure*}

We also show that the quality measure $d_q$ (JS-divergence) is not affected by $n$ and, as expected, it does not exceed the threshold $\epsilon$, in \figsref{fig7c} and \ref{fig7d}.
$RO$ finds solutions with larger $d_q$ than $RH$. This is because $RH$ works in a greedy fashion. That is, the initial bins are sanitized heavily, which increases $d_q$ and does not leave
much room for sanitizing the subsequent bins without exceeding $\epsilon$.

In addition, we show the impact of $n$ on $NCE$, in \figsref{fig8a} and \ref{fig8b}. Both algorithms achieve similar scores, because they
aim to optimize $d_p$ with constraint $d_q\leq \epsilon$ and achieve similar results with respect to $d_q$ (see \figsref{fig7c} and \ref{fig7d}).
The scores were zero (no quality loss), or low, which suggests that both $RO$ and $RH$ are able to preserve clustering quality. 

\begin{figure*}[ht!]\hspace{-2mm}
	\begin{subfigure}[b]{.24\textwidth}\centering
		\includegraphics[scale=0.25,clip]{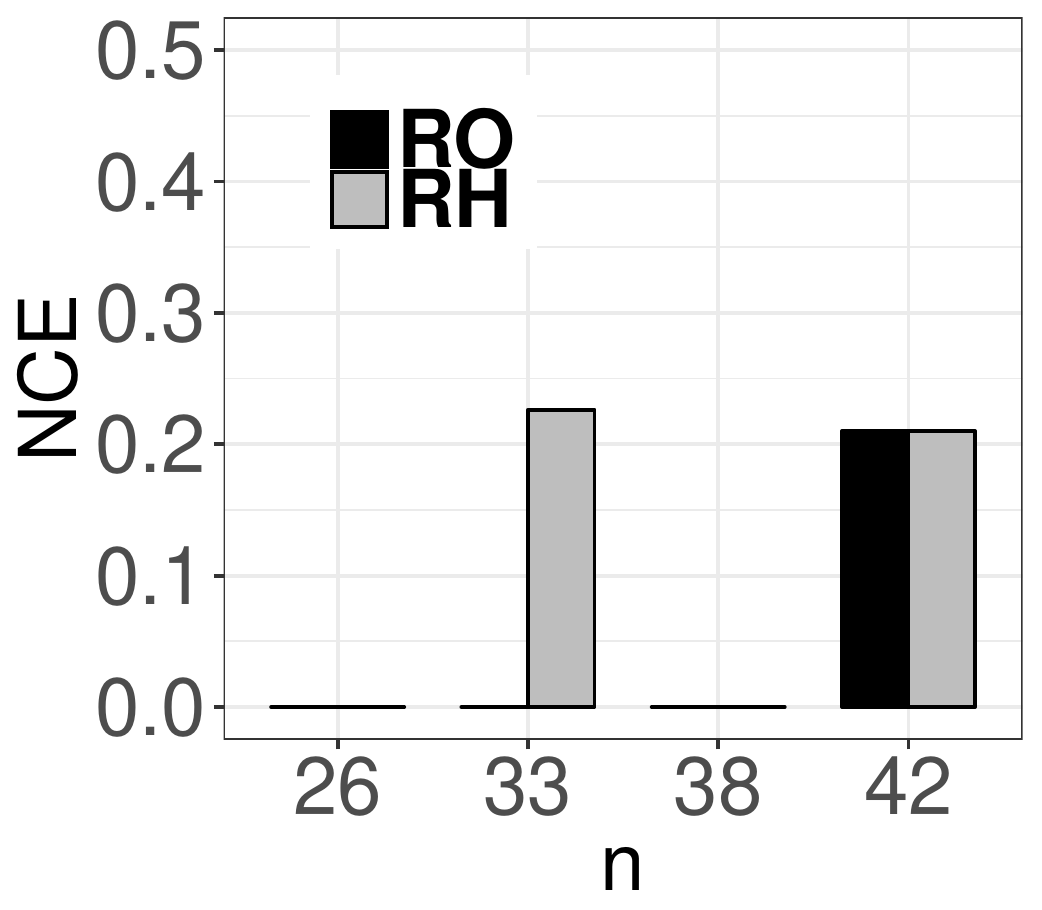}
		\caption{NYC}\label{fig8a}
	\end{subfigure}\hspace{+1mm}
	\begin{subfigure}[b]{.24\textwidth}\centering
		\includegraphics[scale=0.25,clip]{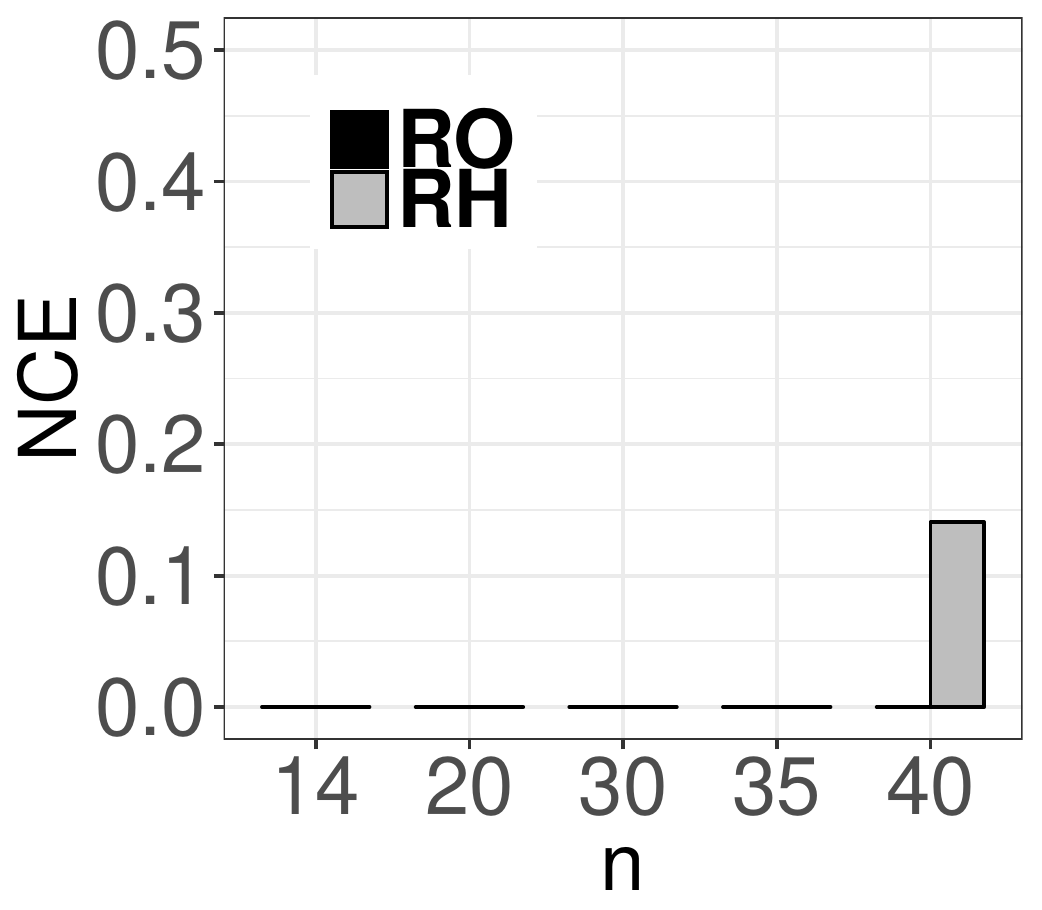}
		\caption{TKY}\label{fig8b}
	\end{subfigure}\hspace{+1mm}
	\begin{subfigure}[b]{.24\textwidth}\centering
		\includegraphics[scale=0.25,clip]{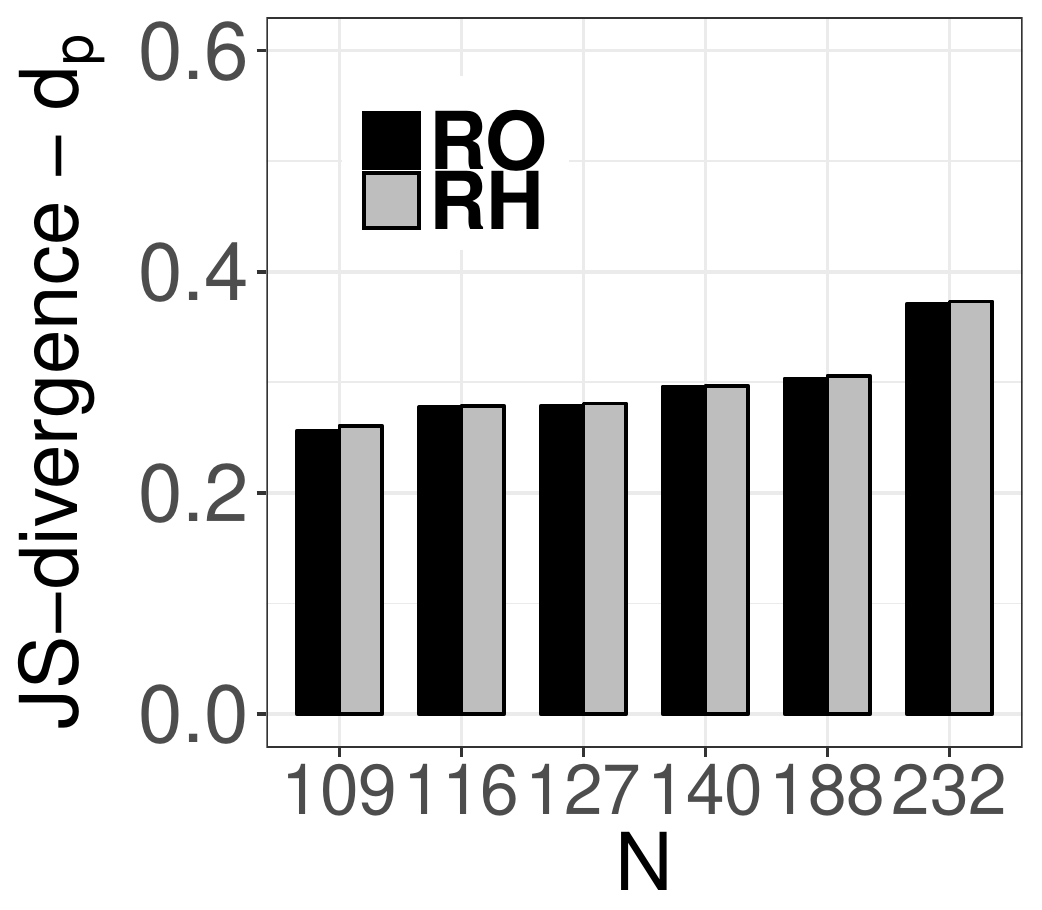}
		\caption{NYC}\label{fig8c}
	\end{subfigure}\hspace{+1mm}
	\begin{subfigure}[b]{.24\textwidth}\centering
		\includegraphics[scale=0.25,clip]{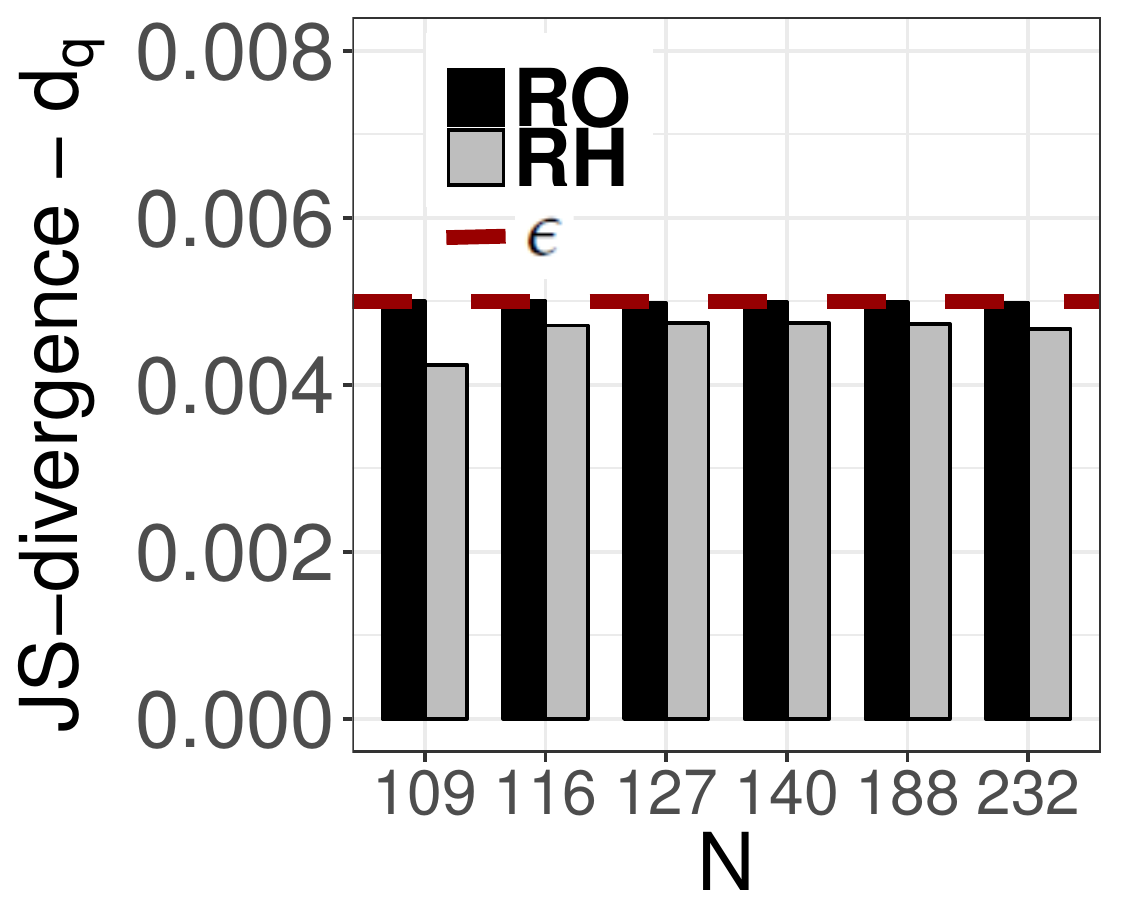}
		\caption{NYC}\label{fig8d}
	\end{subfigure}
	\caption{NCE vs length $n$ for histograms with $N=100$ in: (a) \emph{NYC}. (b) \emph{TKY}.
		(c) $d_p$ vs size $N$, and (d) $d_q$ vs size $N$, for histograms with $n=25$ in \emph{NYC}}
\end{figure*}

\noindent \paragraph{Impact of size $N$}  To illustrate the impact of $N$ on quality and privacy, we present results obtained for randomly selected histograms of varying $N$ with $n=25$ for \emph{NYC}. The results for \emph{TKY} are qualitatively similar (omitted). We do not report the median of all histograms of certain $N$, because the results follow skewed distributions.

We show that the privacy measure $d_p$ (JS-divergence) increases with $N$, in \figref{fig8c}. This is because there are more counts that need to change (increase or decrease) to minimize $d_p$ subject to
$d_q\leq \epsilon$. The results for $RH$ are very close to those for $RO$; the $d_p$ scores for $RH$ are no more than $1.8\%$ larger. This suggests that $RH$ is an effective heuristic.

We also show that the quality measure $d_q$ (JS-divergence) is not affected by $N$ and that it does not exceed the threshold $\epsilon$, in \figref{fig8d}.
Again, $RO$ finds solutions with larger $d_q$ than $RH$. This is  because, due to its greedy nature, $RH$ sanitizes heavily the first bins,
which increases $d_q$ and prevents the sanitization of subsequent bins without exceeding $\epsilon$.

\begin{wrapfigure}[9]{L}{.4\textwidth}
	\centering
	\includegraphics[scale=0.27,clip]{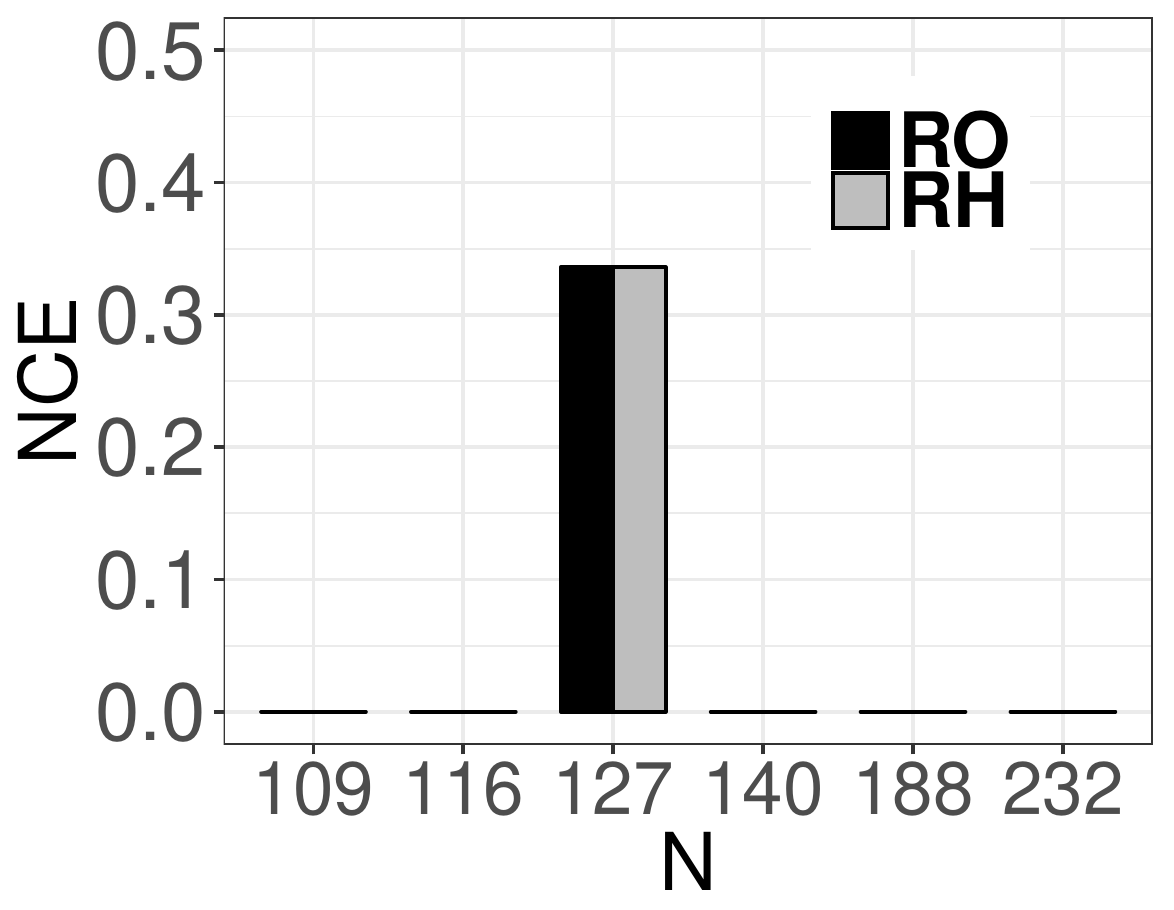}
	\caption{$NCE$ vs size $N$ for histograms with $n=25$ in \emph{NYC}}\label{NCEvsNk3}
\end{wrapfigure}

In addition, we show the impact of $N$ on $NCE$, in \figref{NCEvsNk3}. Both algorithms achieve similar scores, because they aim to optimize $d_p$ with constraint $d_q\leq \epsilon$ and achieve similar results with respect to
$d_q$ (see \figref{fig8d}). Their scores were low or zero. Thus, they are able to preserve clustering quality. 

\noindent\paragraph{Impact of threshold $\epsilon$} To illustrate the impact of $\epsilon$ on quality and privacy, we present results obtained for a histogram with $n=40$ and $N=100$ in \emph{NYC}.
The results for \emph{TKY} are similar (omitted).

We show that the privacy measure $d_p$ (JS-divergence) decreases with $\epsilon$, in \figref{fig9a}.
This is because both $RO$ and $RH$ consider a larger space of possible solutions when $\epsilon$ is larger, and thus they are able to find a better solution with respect to $d_p$.
In addition, the results for $RH$ and $RO$ are very similar; the $d_p$ for $RH$ is at most $2.4\%$ (on average $0.5\%$) higher than that for $RO$.

\begin{figure*}[ht!]\hspace{-2mm}
	\begin{subfigure}[b]{.32\textwidth}\centering
		\includegraphics[scale=0.26,clip]{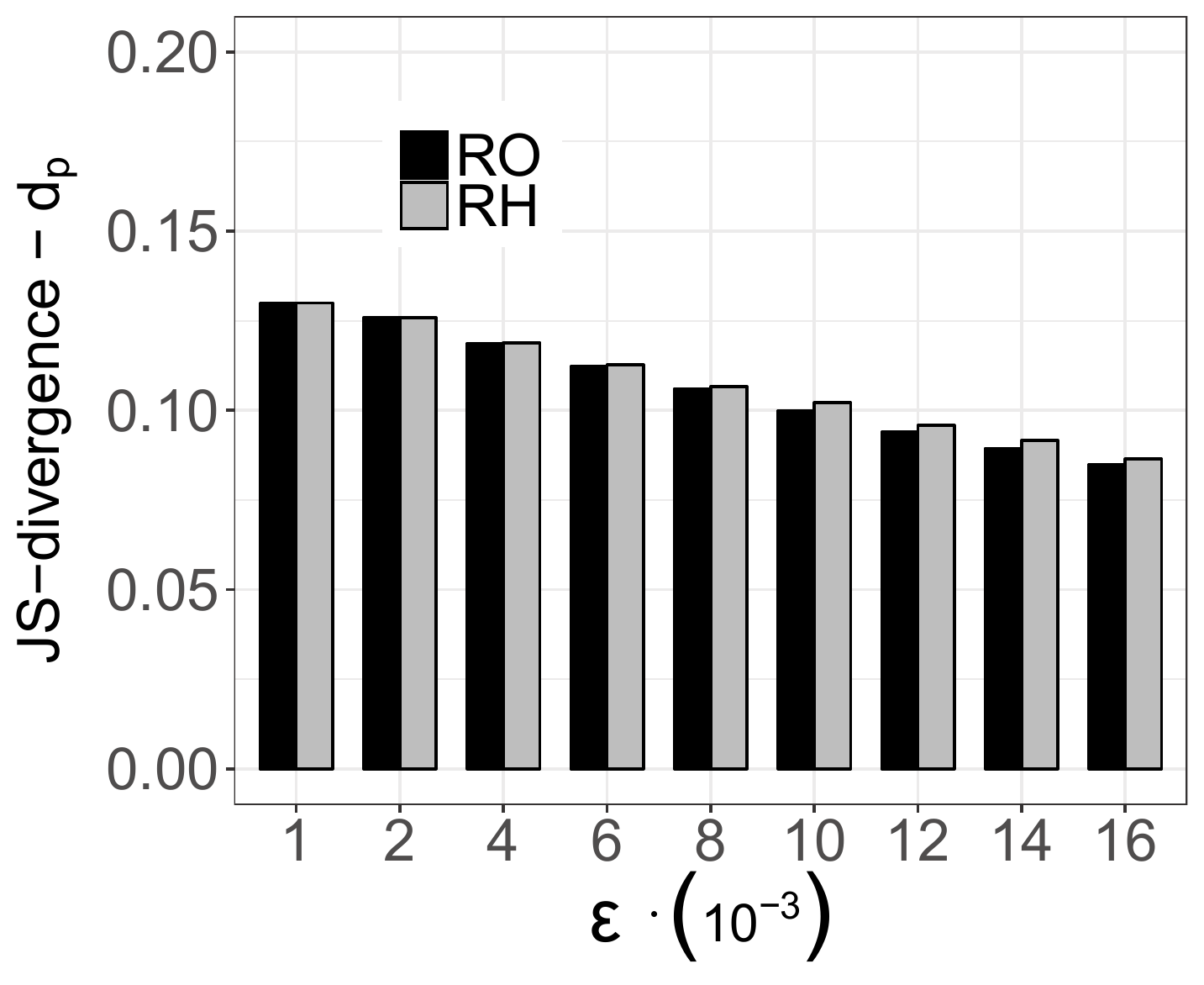}
		\caption{}\label{fig9a}
	\end{subfigure}\hspace{+1mm}
	\begin{subfigure}[b]{.32\textwidth}\centering
		\includegraphics[scale=0.26,clip]{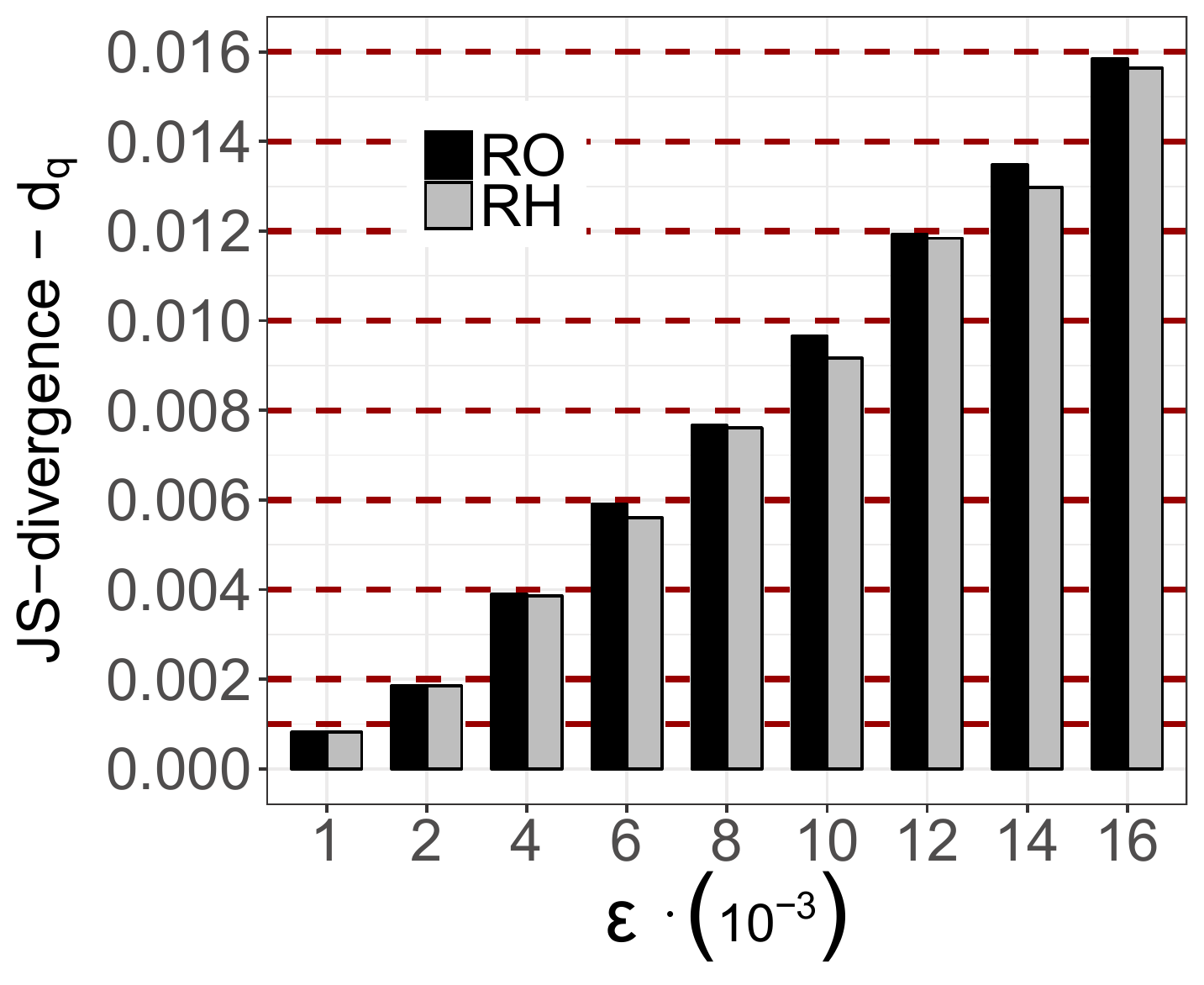}
		\caption{}\label{fig9b}
	\end{subfigure}\hspace{+1mm}
	\begin{subfigure}[b]{.32\textwidth}\centering
		\includegraphics[scale=0.27,clip]{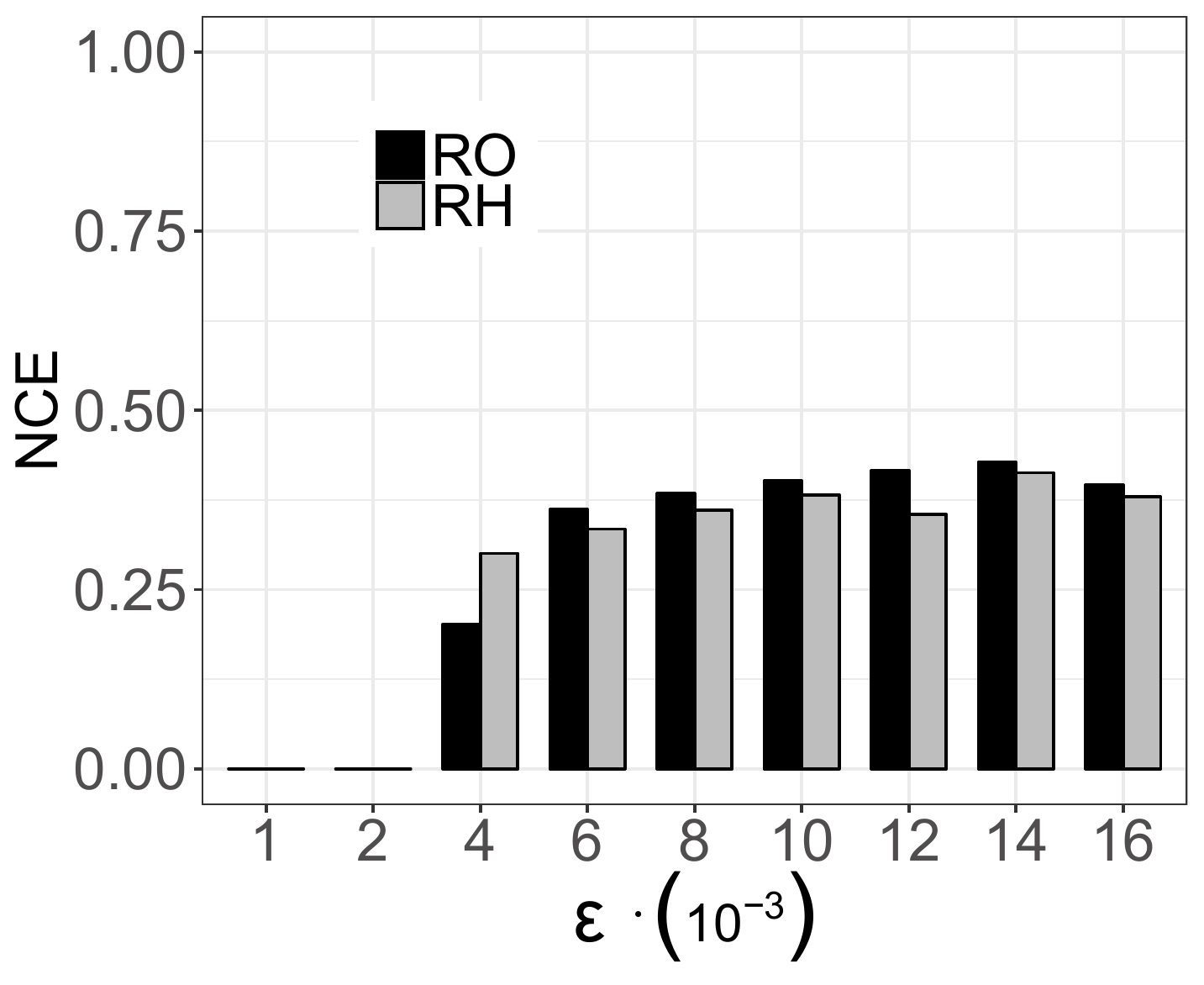} 
		\caption{}\label{fig9c}
	\end{subfigure}
	\caption{(a) $d_p$ vs threshold $\epsilon$, and (b) $d_q$ vs threshold $\epsilon$, for a histogram with length $n=40$ and size $N=100$ in \emph{NYC}.
		(c) $NCE$ vs threshold $\epsilon$, for a histogram with length $n=40$ and size $N=100$ in \emph{NYC}}
\end{figure*}

\begin{wrapfigure}[10]{L}{.52\textwidth}
	\begin{subfigure}[b]{.24\textwidth}\centering
		\includegraphics[scale=0.25,clip]{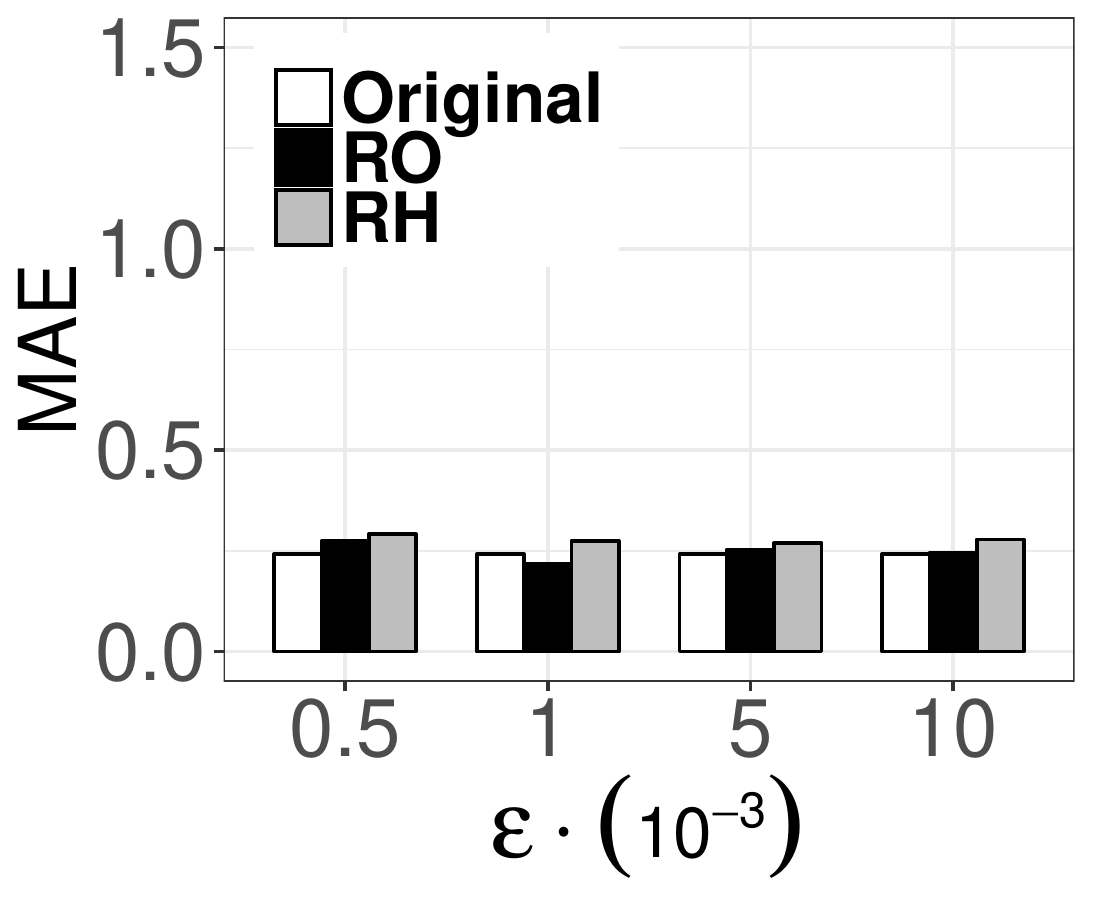}
		\caption{NYC}\label{recom2a}
	\end{subfigure}\hspace{+1mm}
	\begin{subfigure}[b]{.24\textwidth}\centering
		\includegraphics[scale=0.25,clip]{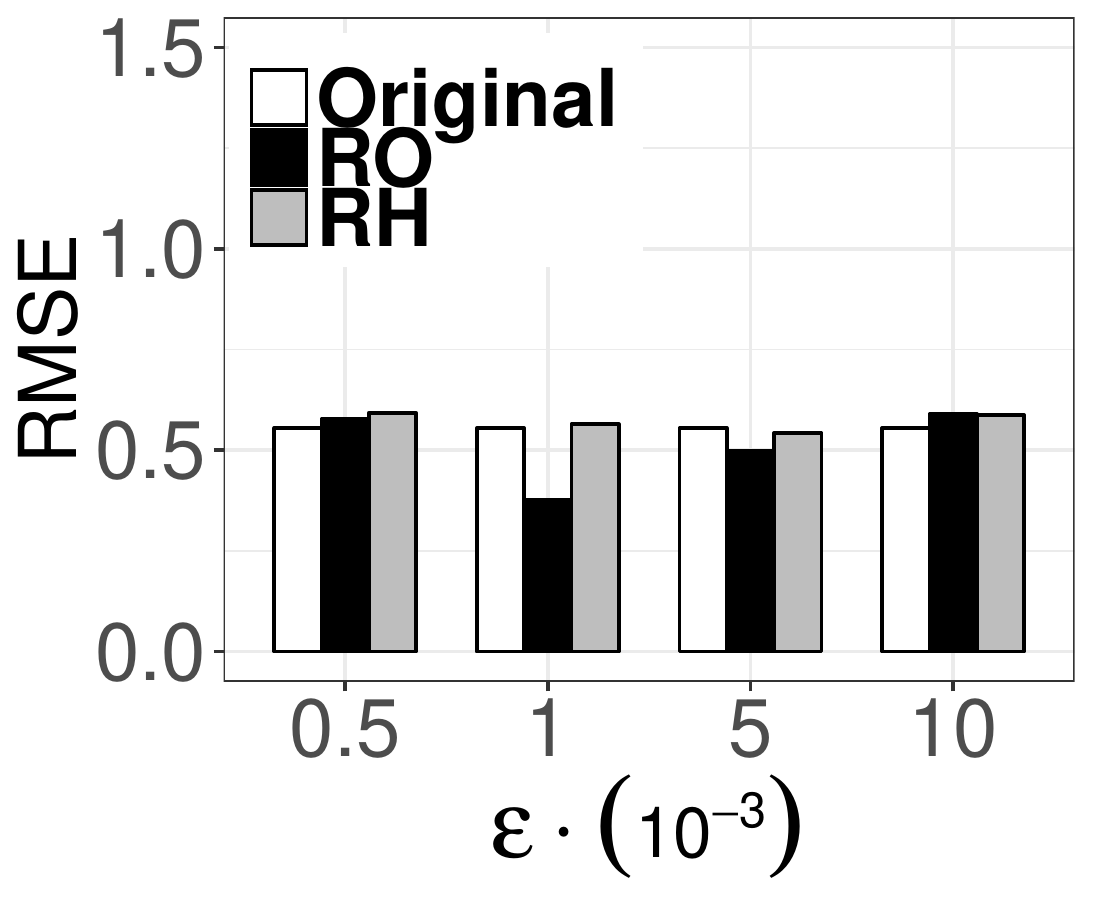}
		\caption{NYC}\label{recom2b}
	\end{subfigure}
	\caption{Recommendation quality for varying $\epsilon$ with respect to: (a) $MAE$, and (b) $RMSE$}\label{recom2}
\end{wrapfigure}

We also show that the quality measure $d_q$ (JS-divergence) for both $RO$ and $RH$ is close to $\epsilon$, in \figref{fig9b}.
Again, the $d_q$ scores of $RO$ are slightly larger than those of $RH$, because $RH$ works in a greedy fashion (i.e., the first bins are sanitized heavily, which increases
$d_q$ and does not leave much room for sanitizing the subsequent bins without exceeding $\epsilon$), as explained above.

In addition, we show that $NCE$ increases with $\epsilon$, in \figref{fig9c}. This is because the algorithms trade-off quality for privacy when $\epsilon$ is larger (i.e., $d_q$ can be as high as $\epsilon$).
The $NCE$ scores for $RO$ and $RH$ are similar, with those for $RH$ being lower (better) by $1.5\%$ on average.

\paragraph{Recommendation quality} We investigate the impact of sanitization on recommendation quality by using test datasets of original vs test datasets of histograms that were sanitized by applying $RO$ or $RH$ with different values of $\epsilon$. Figures \ref{recom2a} and \ref{recom2b} show that $MAE$ and $RMSE$ are not substantially affected by sanitization, for all tested $\epsilon$ values. This suggests that recommendation quality is preserved fairly well. In some cases, the $MAE$ and $RMSE$ scores for the sanitized histograms were lower (better) than those for the original histograms. This is because the recommendation scores for these histograms approach their corresponding true location counts after sanitization.

\subsubsection{Runtime performance for the $RO$ algorithm and the $RH$ heuristic}

\paragraph{Impact of length $n$} We show that the runtime of both $RO$ and $RH$ increases with $n$, in \figref{fig11a}.
This is because $RO$ runs on a multipartite graph, $G_{TR}$, with more layers $n$, and $RH$ needs to consider more bins $n$.
$RH$ is at least two orders of magnitude more efficient than $RO$.
For example, $RH$ requires $5.2$ milliseconds to sanitize a histogram with $n=26$ while $RO$ requires $2.3$ seconds.
Note that $RH$ scales close to linearly with $n$, which shows that the efficiency of $RH$ is better than what is predicted by the worst-case time complexity analysis in Section~\ref{sec:solution:heuristics}.

To further investigate the impact of length on runtime, we apply $RO$ and $RH$ to each histogram with length larger than $75$ in \emph{NYC} and \emph{TKY} (see \figref{alllengthNYC} and \ref{alllengthTKY}).
There are $17$ and $7$ such histograms in \emph{NYC} and \emph{TKY}, respectively. These histograms are generally demanding to sanitize, because they also have large size (up to $2061$). Again, we observe that $RH$ is more efficient than $RO$ by at least two orders of magnitude; $RH$ needs on average $8.2$ milliseconds, while $RO$ needs on average $3.6$ seconds.
In these experiments, we use $\epsilon=10^{-5}$.
For larger $\epsilon$ values, the difference between the two algorithms increases, because $RH$ scales better than $RO$ with respect to $\epsilon$, as explained above.
Repeating the same experiment using the synthetic histograms (see \figref{syntheticNYC} and \ref{syntheticTKY}), we find that both $RO$ and $RH$ scale well with $n$, and $RH$ scales close to linearly with $n$.

\begin{figure*}[!ht]\hspace{-2mm}
	\begin{subfigure}[b]{.24\textwidth}\centering
		\includegraphics[scale=0.31,clip]{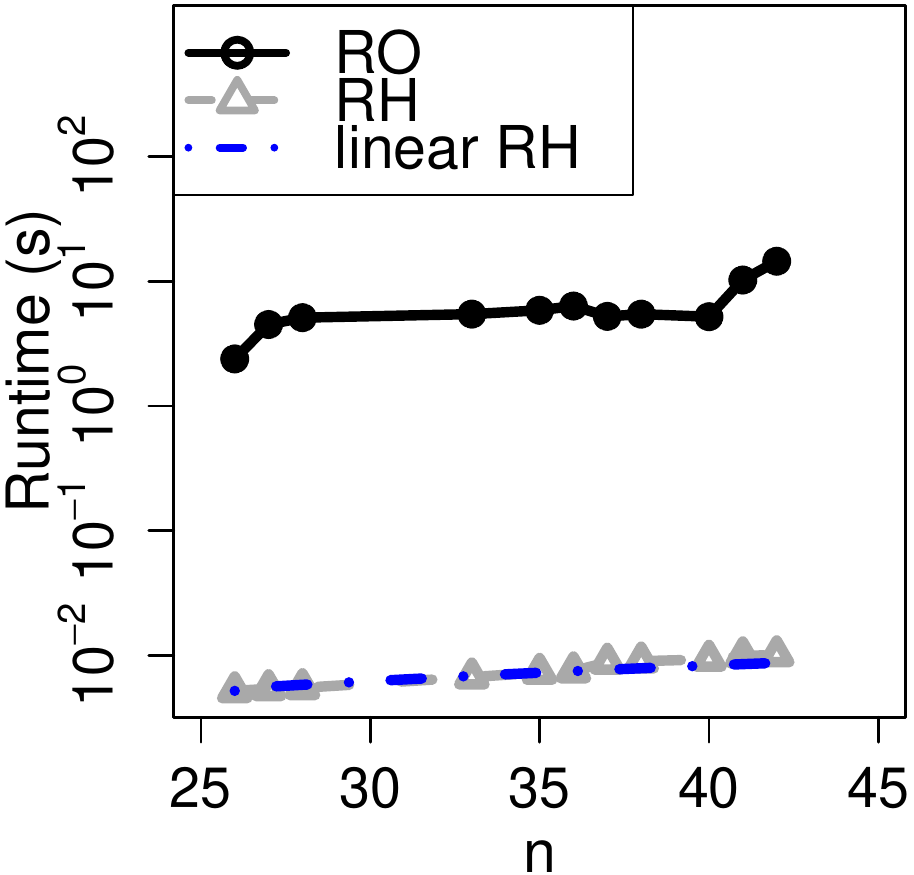}
		\caption{}\label{fig11a}
	\end{subfigure}\hspace{+1mm}
	\begin{subfigure}[b]{.24\textwidth}
		\includegraphics[scale=0.3,clip]{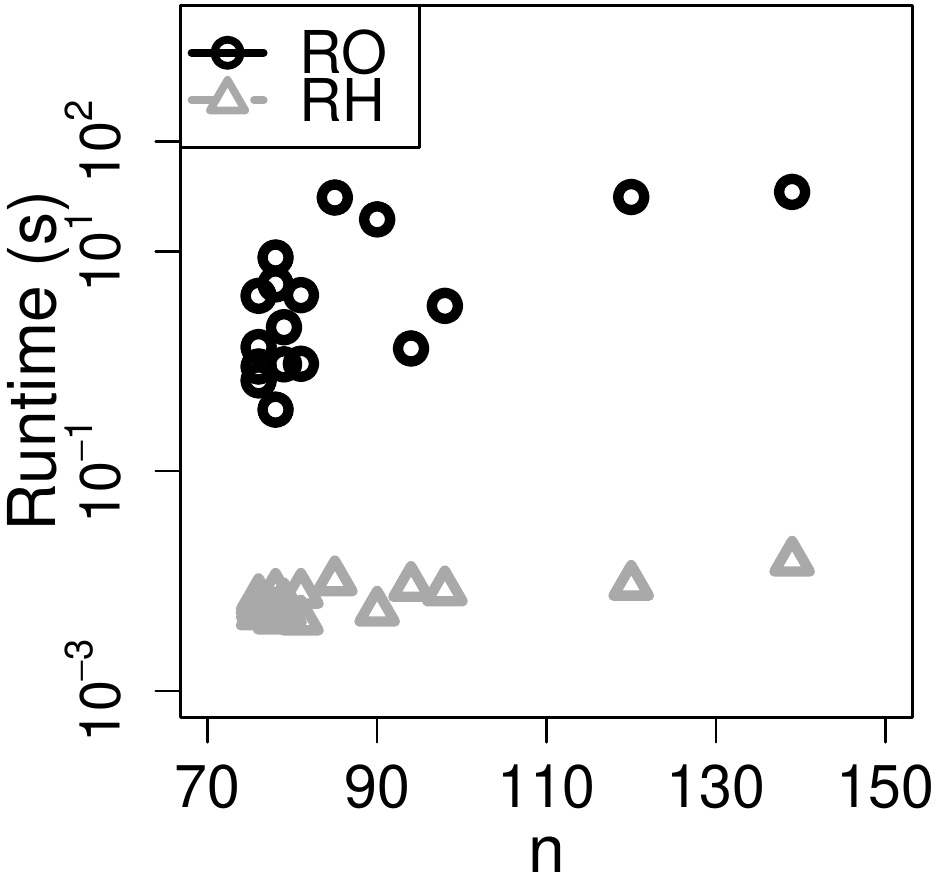}
		\caption{NYC}\label{alllengthNYC}
	\end{subfigure}\hspace{+0.5mm}
	\begin{subfigure}[b]{.24\textwidth}
		\includegraphics[scale=0.3,clip]{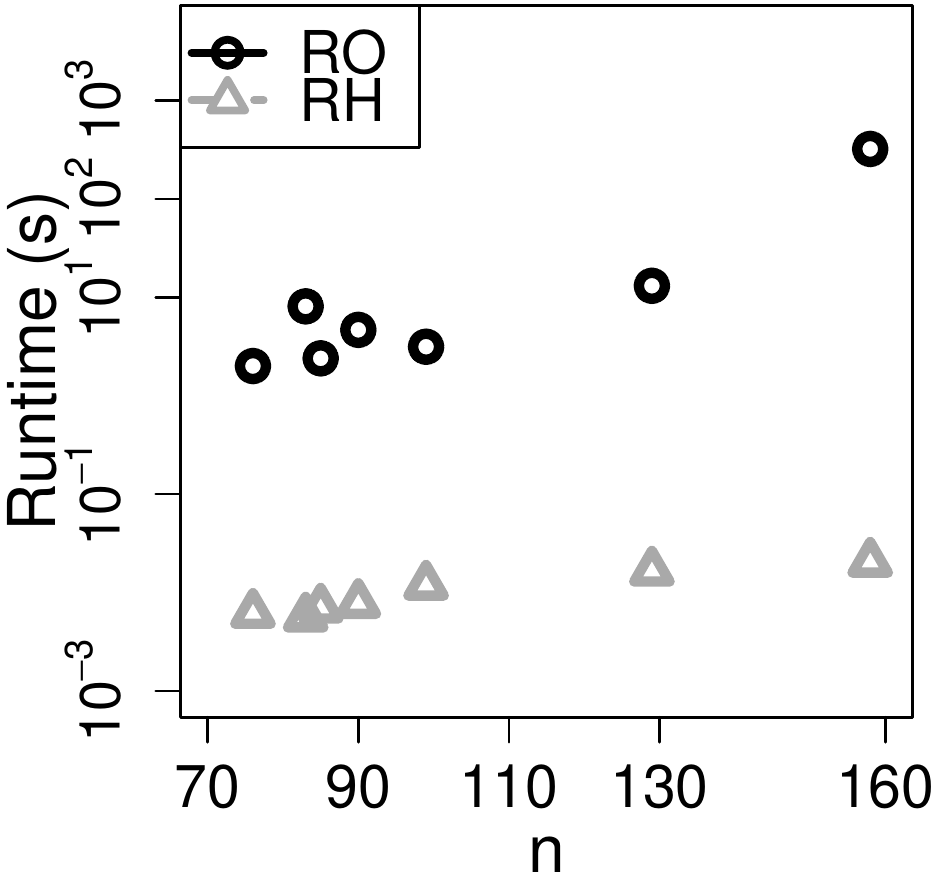}
		\caption{TKY}\label{alllengthTKY}
	\end{subfigure}
	\hspace{+0.5mm}
	\begin{subfigure}[b]{.24\textwidth}
		\includegraphics[scale=0.3,clip]{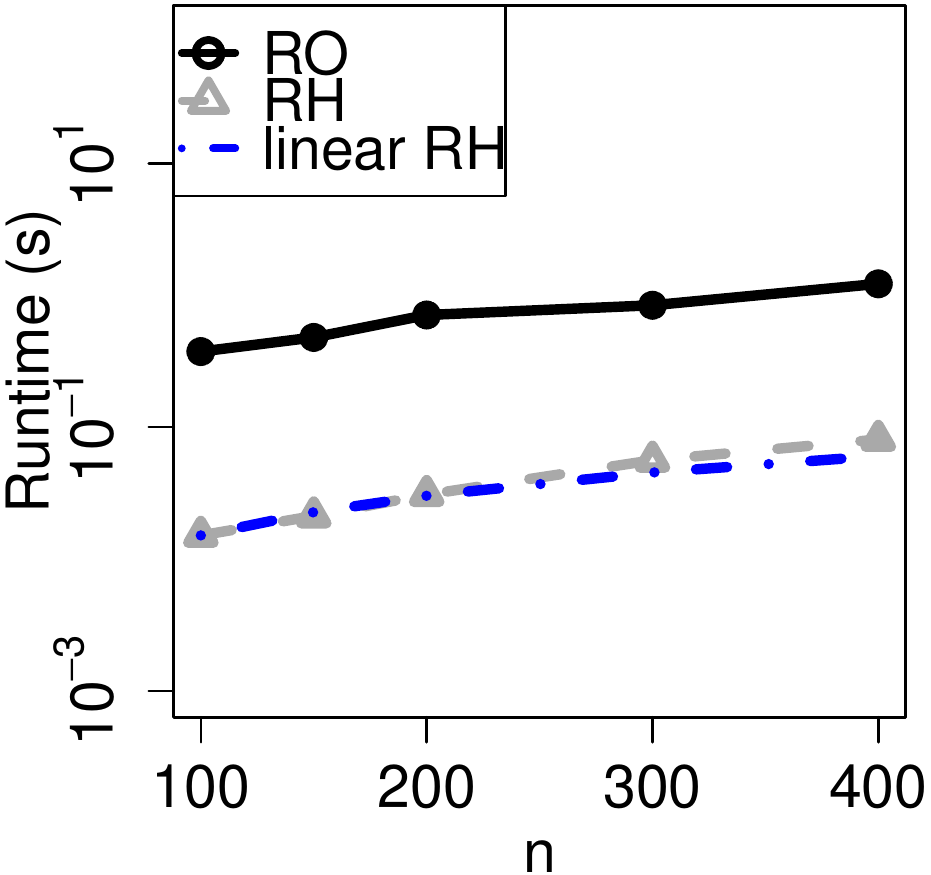}
		\caption{NYC}\label{syntheticNYC}
	\end{subfigure}
	\caption{Runtime vs length $n$, for: (a) each histogram with $N=100$ in \emph{NYC}, (b) each histogram with $n>75$ in \emph{NYC}, (c) each histogram with $n>75$ in \emph{TKY}, and (d) synthetic histograms with varying length and $N=192$}\label{figalllengths}
\end{figure*}

\noindent\paragraph{Impact of size $N$} We show that the runtime of both $RO$ and $RH$ increases with $N$, in \figref{fig11b}. This is because the multipartite graph $G_{TR}$ built by $RO$ has more nodes $O(N\cdot n)$
and thus more paths, and $RH$ needs to consider more ``moves'' from source to destination bins. Again, $RH$ is at least two orders of magnitude more efficient than $RO$. For example, $RH$ required at most $6.7$ milliseconds
to sanitize a histogram with $N=109$ while $RO$ required $5.6$ seconds.

\begin{figure*}[ht!]\hspace{-2mm}
	\begin{subfigure}[b]{.24\textwidth}
		\includegraphics[scale=0.3,clip]{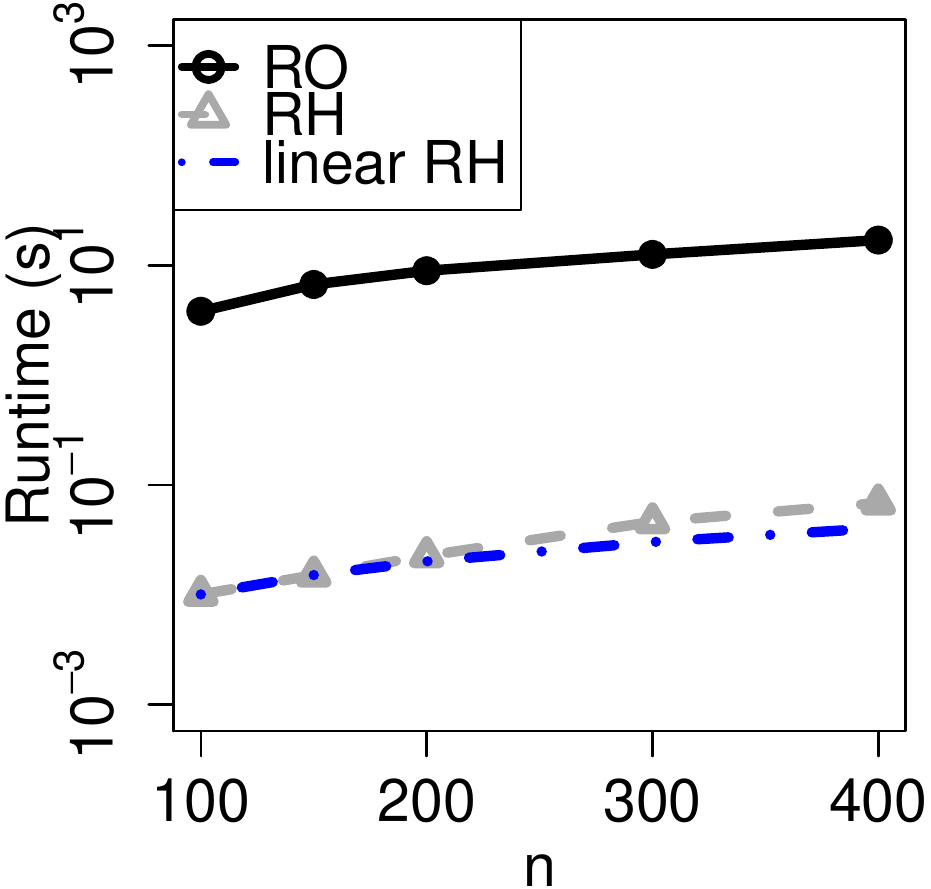}
		\caption{TKY}\label{syntheticTKY}
	\end{subfigure}
	\hspace{+1mm}
	\begin{subfigure}[b]{.24\textwidth}\centering
		\includegraphics[scale=0.31,clip]{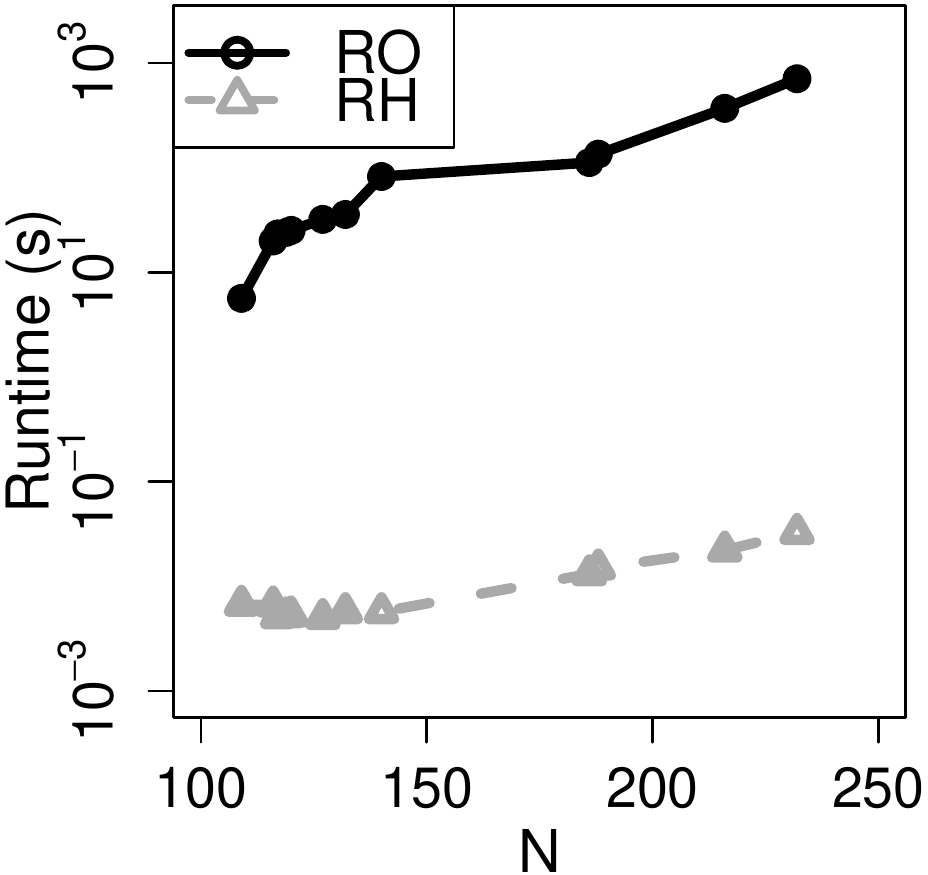}
		\caption{}\label{fig11b}
	\end{subfigure}\hspace{+1mm}
	\begin{subfigure}[b]{.24\textwidth}\centering
		\includegraphics[scale=0.31,clip]{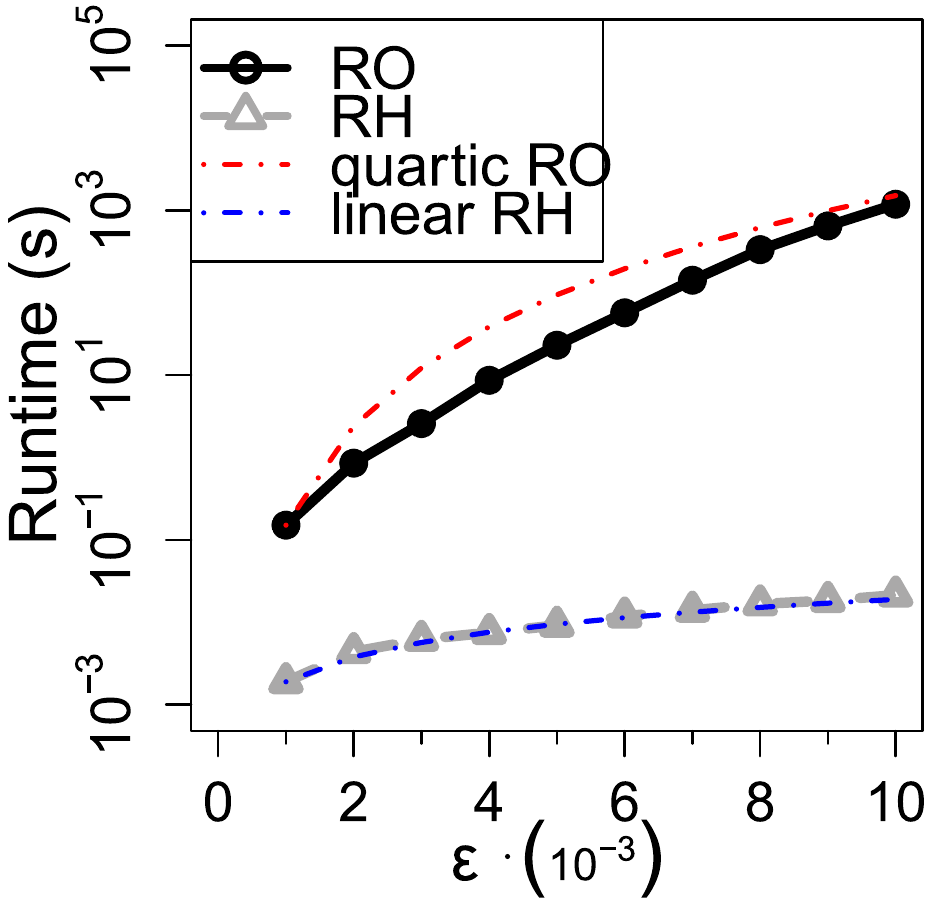}
		\caption{}\label{fig11c}
	\end{subfigure}\hspace{+1mm}
	\begin{subfigure}[b]{.24\textwidth}\centering
		\includegraphics[scale=0.31,clip]{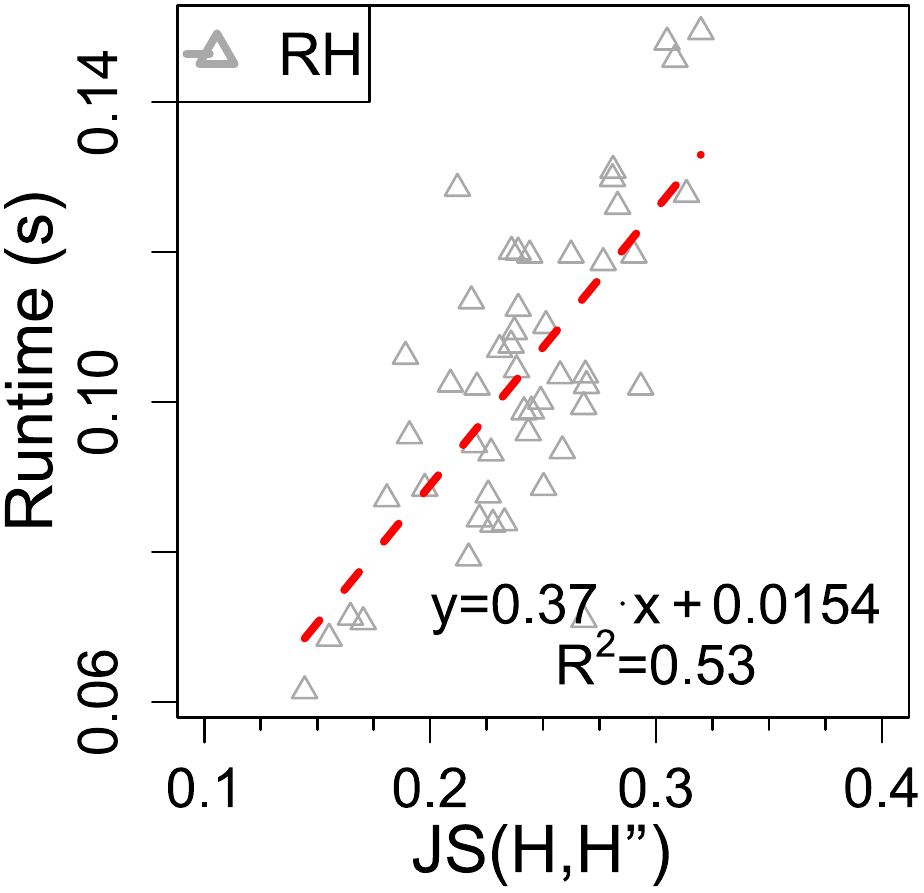}
		\caption{}\label{fig11d}
	\end{subfigure}
	\caption{Runtime vs : (a) Length $n$, for synthetic histograms with varying length and $N=642$. (b) Size $N$, for each histogram with $n=25$ in \emph{NYC}. (c) Threshold $\epsilon$, for a histogram with $n=40$ and $N=100$ in \emph{NYC}.
		(d) JS-divergence between a histogram $H$ with $n=40$ and $N=100$ in \emph{NYC} and different target histograms with increasing JS-divergence from $H$}\label{fig11}
\end{figure*}

\noindent\paragraph{Impact of threshold $\epsilon$} We show that the runtime of both $RO$ and $RH$ increases with $\epsilon$, in \figref{fig11c}.
This is because, when $\epsilon$ is larger, the multipartite graph built by $RO$ has more edges and thus more paths, and $RH$ considers more ``moves'' from source to destination bins.
$RH$ is at least two orders of magnitude more efficient than $RO$, and it scales better with $\epsilon$ i.e., linearly versus quartically (proportionally to $\epsilon^4$). This suggests that $RH$ is a practical heuristic for large $\epsilon$ values, given that it produces solutions similar to those of $RO$.

\noindent\paragraph{Impact of target histogram $H''$} We show that the runtime of $RH$ increases with the distance $JS(H,H'')$, for different target histograms $H''$, in \figref{fig11d}. This is because
$RH$ has more choices (i.e., there are more ways to transfer the counts of a source bin to a destination bin in $H$, when $H''$ is further from $H$ in terms of JS-divergence). In this experiment,
we use $\epsilon=0.5$, because the runtimes with the default $\epsilon$ value are too small (few milliseconds) to obtain a meaningful result.
We do not report the result for $RO$, because its runtime is not affected by the target histograms. The reason is that $RO$ builds the same multipartite graph for each target histogram $H''$, since $H''$ has the same length $n$ and size $N$ with the original histogram.

\section{Related Work}\label{sec:related_work}
This paper is at the intersection of location privacy and  histogram privacy, 
which are discussed in Sections \ref{related_work_location_privacy} and \ref{related_work_histogram_privacy}, respectively. We also discuss privacy-preserving recommendation in Section \ref{related_work_recom}, as a potential application of our methods. 

\vspace{-6mm}
\subsection{Location privacy}\label{related_work_location_privacy}
\vspace{-3mm}

Research on location privacy focuses on (I) location-based services (LBS), or (II) location data publishing.

Research on LBS is mostly inspired from applications running on GPS-enabled mobile devices like smartphones and tablets -- but also cars. 
Consequently, it addresses privacy for users who need to send data on the fly (as they move about), to a server that will provide them with some useful service (e.g. the location of the nearest restaurant). 
Privacy mechanisms in such scenarios need to make protection decisions on the fly, without knowing the future locations that the user will visit \cite{shokri2012protecting,ch2013broadening,fawaz2014location,shokri2017privacy}.
For example, \cite{shokri2017privacy} proposes a method for preventing the inference of locations that have been or will be visited by a user, based on what the user shares at any moment with a location-based service. Other recent research protects sensitive spatiotemporal location sequences \cite{Abul2018}.
As another example, \cite{fawaz2014location} proposes a method that prevents an LBS server from aggregating the locations sent by a user into a histogram and then associating this histogram with the user. The method perturbs the user's locations one by one, before they are sent to the LBS server, by adding noise to them in order to enforce the privacy notion of geo-indistinguishability \cite{ch2015location}.

Research on location data publishing is inspired from the publication of large datasets, possibly as a database. 
Consequently, it addresses more static scenarios, in which the whole dataset to be protected is given to the protection algorithm as input \cite{inference-eurosp16,arapinis,cicek,chenCCS,cheninfsci,ferrer,terrovitisTKDE17,poulistdp}. 
There are works showing the feasibility of attacks on \emph{pseudonymized} data (i.e., data in which a user's identifying information is represented by a random id) \cite{arapinis}, or on \emph{completely anonymized} data (i.e., a sequence $(e_1,\ldots, e_n)$, where the event $e_i=(l,t)$, $i\in[1,n]$, represents a visit to location $l$ at time $t$ and is not associated to a specific user) \cite{inference-eurosp16}. 
For example, reference \cite{inference-eurosp16} shows how an attacker can use completely anonymized data to associate a user with their event subsequence (\emph{path}). 
There are also works
\cite{cicek,cheninfsci,ferrer,terrovitisTKDE17,poulistdp} which propose methods for anonymizing user-specific location data (i.e., a dataset where each record corresponds to a different user and contains a sequence of locations visited by the user and/or the time that these visits occurred). 
For example, reference \cite{terrovitisTKDE17} proposes algorithms for preventing the inference of a user's sensitive locations by an attacker knowing a subsequence of the user's locations. 
The algorithms of \cite{terrovitisTKDE17} use \emph{suppression} (deletion) of locations and \emph{splitting} of user sequences into carefully selected subsequences.

Yet, no research in location privacy has aimed to protect histograms of locations. The object/fact to be protected has been either a single location (in the LBS setting), or a (sub)sequence of locations (in the location data publishing setting). However, protecting single locations separately provides no guarantee about the effect on the histogram as a whole. It could happen that, e.g., each individual location is replaced with another location,
so no single location is disclosed/compromised, but the histogram as a whole is very similar or even identical with the original one. 
Similarly, protecting location data could again lead to the same problem.
It could happen that individual locations in a user's sequence are modified, but the histogram remains unprotected. 
Thus, works on LBS or location data publishing cannot be used as alternatives to our approach.

\vspace{-5mm}
\subsection{Histogram privacy} \label{related_work_histogram_privacy}
\vspace{-2mm}

Research on histogram privacy is inspired from applications where a histogram is published
as a statistical summary (approximation) of the distribution of an attribute in a (relational) dataset. 
For example, consider a dataset, where each record contains the zip-code of a different individual. 
The distribution of the zip-code attribute in the dataset can be represented with a histogram, where each bin is associated with a different zip-code value and the bin frequency (count) is the number of individuals in the dataset who live in the zip-code. 
Publishing such histograms is useful for performing count query answering and data mining tasks (e.g., clustering), but it may lead to the disclosure of sensitive information about individuals \cite{acs2012differentially,xu2013differentially, kellaris, sdmhistogram,sortaki,qardaji,li2016improving}.
For instance, consider an adversary who knows the names of all three individuals, $i_1$, $i_2$, and  $i_3$, in a (non-released) dataset but the zip-codes of only $i_1$ and $i_2$. 
When the published histogram contains the count of each zip-code in the dataset, the adversary can infer the zip-code of $i_3$ from the histogram. 
To prevent this type of disclosure, the frequencies in the histogram are perturbed, typically by noise addition, in order to satisfy  \emph{differential privacy} \cite{dwork2006}. 
Informally, differential privacy ensures that the inferences that can be made by an adversary about an individual  will be approximately independent of whether the individual's record is included in the dataset or not.

Several works have applied differential privacy to sanitize histograms \cite{acs2012differentially,xu2013differentially, kellaris, sdmhistogram,sortaki,qardaji,li2016improving}. A straightforward way to achieve this is by adding noise to the frequency of each bin of the histogram, according to the Laplace mechanism \cite{dwork2006}. However, this procedure results in excessive utility loss \cite{acs2012differentially}. 
Therefore,  existing works \cite{acs2012differentially,xu2013differentially, kellaris, sdmhistogram,sortaki,qardaji,li2016improving} employ clustering to reduce the loss of utility, in three steps: 
(I) They cluster bins with similar frequencies together. 
(II) They apply the Laplace mechanism to the average (mean or median) of the frequencies in each cluster, to obtain a ``noisy center'' of the cluster. 
(III) They publish a histogram where each frequency bin in each cluster is replaced by the noisy center of its corresponding cluster. 
While clustering incurs some utility loss, it reduces the noise that is added by the Laplace mechanism, leading to better overall utility. 
Specifically, the works of \cite{acs2012differentially,xu2013differentially} require each cluster to be formed of adjacent bins, while the work of \cite{kellaris} requires each cluster to have the same number of bins. 
Subsequent works
\cite{sdmhistogram,sortaki,qardaji,li2016improving} lift these restrictions to
further improve utility. 
For example, reference \cite{sdmhistogram} proposes a clustering framework, which can be instantiated by optimal or heuristic algorithms that trade-off the utility loss incurred by clustering with the utility loss incurred by the Laplace mechanism.

At a high level, our work is similar to the works in \cite{acs2012differentially,xu2013differentially, kellaris, sdmhistogram,sortaki,qardaji,li2016improving}, in that it aims to protect a histogram (or it can be applied to each histogram in a dataset of histograms). However, it differs from the works in \cite{acs2012differentially,xu2013differentially, kellaris, sdmhistogram,sortaki,qardaji,li2016improving} along two dimensions: 
(I) It considers a
histogram that represents the locations associated to a \emph{single user}, instead of a histogram representing the values of many different individuals in an attribute of an underlying dataset. 
(II) It sanitizes a histogram by redistributing counts between bins, as specified by Problems \ref{SLHdef}, \ref{TRdef}, and \ref{TAdef}, instead of adding noise into the counts. 
Thus, the methods in \cite{acs2012differentially,xu2013differentially, kellaris, sdmhistogram,sortaki,qardaji,li2016improving} cannot be used to deal with the problems we consider.
In fact, applying any of the methods in \cite{acs2012differentially,xu2013differentially, kellaris, sdmhistogram,sortaki,qardaji,li2016improving} to a
histogram that represents the locations of a single user would simply prevent the inference of the exact frequencies (counts) of locations in the user's histogram. It would not protect against the disclosure of visits to sensitive locations (i.e., it cannot solve the $SLH$ problem), nor against the disclosure of the fact that the histogram is similar/dissimilar to a target histogram (i.e., it cannot solve the \TA{}/\TR{} problem).

A different, less related class of works can be used to protect a histogram by making it indistinguishable within a set of histograms that is published \cite{xu2006utility, ghinita2007fast}. 
These works differ from ours in their setting, in their privacy notion, or both.
They differ in terms of setting because they consider a set of histograms (or more generally, vectors of frequencies \cite{xu2006utility, ghinita2007fast}) rather than a single histogram with the location information of a single user. 
They differ in terms of privacy notion because they aim to prevent the disclosure of the identity of individuals, from the published set of histograms (i.e., the association of a histogram with identity information that is known to an attacker), rather than the inference of location information from a single histogram.

\vspace{-5mm}
\subsection{Privacy-preserving recommendation}\label{related_work_recom}
\vspace{-2mm}

There are several privacy-preserving recommendation methods. Most of them (e.g., \cite{trus1,trus2})  assume there is a trusted server that applies privacy protection (e.g., anonymization) jointly to the data of many users. Unlike these methods, we assume a different setting, in which the user protects their histogram by themselves. Our setting is conceptually similar to the untrusted server setting \cite{shenrecomccs,shenrecomicdm,polat}, in which a user protects their data prior to disseminating them. Specifically, \cite{shenrecomccs,shenrecomicdm} propose methods in which a user applies differential privacy, while \cite{polat} proposes a method in which the user applies \emph{randomized perturbation}. The privacy objective of these methods is to prevent the inference of exact user values. In contrast, we do not directly aim to prevent the inference of exact user values: our privacy notions are formalized by the $SLH$ and \TA{}/\TR{} problems. Also, we do not require that the protected histograms will be used in the task of recommendation, although we experimentally show that the protected histograms that are produced by our approach allow preserving the accuracy of recommendation fairly well.

\vspace{-6mm}
\section{Conclusion}\label{sec:conclusion}
\vspace{-2mm}
In this paper, we propose two new notions of histogram privacy, sensitive location hiding and target avoidance/resemblance, which lead to the following optimization problems: the Sensitive Location Hiding problem ($SLH$), which seeks to enforce the notion of sensitive location hiding with optimal quality, and the Target Avoidance/Resemblance (\TA{}/\TR{}) problem, which seeks to enforce target avoidance/resemblance with bounded quality loss. We also propose optimal algorithms for each problem, as well as an efficient heuristic for the \TA{}/\TR{} problem. 
Our experiments demonstrate that our methods are effective at preserving the distribution of locations in a histogram, as well as the quality of recommendations based on these locations, while being fairly efficient.

\vspace{-5mm}
\appendix
\section{Appendix}

\subsection{Proof of weak NP-hardness for the {\SLH} problem}\label{subsec:SLH_NP_hard}
We first reduce the weakly NP-hard \emph{Multiple Choice Knapsack} ($MCK$) problem \cite{Kellerer2004} to the special case of {\SLH}, where $|L'|=1$. The $MCK$ problem is defined as follows\footnote{The problem also appears with $\geq$ in constraint I \cite{MMCK1}. This variation
	is referred to as $MCK_{\geq}$ and can be transformed to $MCK$ in polynomial time \cite{Kellerer2004}.}:

\begin{equation}
	\min \sum_{i\in[1,m]}\sum_{j\in C_i}c_{ij}\cdot x_{ij}\label{preq1}
\end{equation}
subject to:~ (I) $\quad \sum_{i\in[1,m]}\sum_{j\in C_i}w_{ij}\cdot x_{ij}= b$, (II) $\quad \sum_{j\in C_i}x_{ij}=1, ~~ i=1,\ldots m$, and  (III) $\quad x_{ij}\in \{0,1\}, ~~ i=1,\ldots,m, ~~~~ j\in C_i$.

In $MCK$, we are given a set of elements subdivided into $m$, mutually exclusive classes, $C_1,...,C_m$, and a knapsack. Each class $C_i$ has $|C_i|$ elements. Each element $j \in C_i$ has a cost $c_{ij} \geq 0$ and a weight $w_{ij}$.
The goal is to minimize the total cost (Eq.~\ref{preq1}) by filling the knapsack with one element from each class (constraint II), such that the weights
of the elements in the knapsack satisfy the constraint I, where $b \geq 0$ is a constant.
The variable $x_{ij}$ takes a value $1$, if the element $j$ is chosen from class $C_i$ and $0$ otherwise (constraint III).

We map a given instance $\mathcal{I}_{MCK}$ to an instance $\mathcal{I}_{SLH}$ of the special case of {\SLH} in polynomial time, as follows: 
\begin{itemize}
	\item[(I)] Each class $C_i$, $i\in[1,m]$, is mapped to a location $L_i\notin L'$ whose count $f(L_i)$ in $H$ is arbitrary.
	\item[(II)] A sensitive location $L_{m+1}\in L'$ (without loss of generality) is considered. The count of $L_{m+1}$ in $H$ is set to $f(L_{m+1})=b$. Thus, $H=(f(L_1),\ldots, f(L_m),b)$.
	\item[(III)] Each element $x_{ij}$ with weight $w_{ij}$ and cost $c_{ij}$ is mapped to an operation on $H$, which decreases $f(L_{m+1})$ by $w_{ij}$ and increases $f(L_i)$ by $w_{ij}$ (i.e., transfers $w_{ij}$ visits
	from $L_{m+1}$ to $L_i$) and incurs $q(H,H'[i])=c_{ij}$. If there are multiple operations such that $q(H,H'[i])=c_{ij}$ (e.g., when $q$ is the $L_1$ distance), we select one arbitrarily.
	When $x_{ij}=1$, its corresponding operation is applied to $H$. The result of applying all operations on $H$ is referred to as the sanitized histogram $H'$.
\end{itemize}

We prove the correspondence between a solution $\mathcal{S}$ to $\mathcal{I}_{MCK}$ and a solution $H'$ to $\mathcal{I}_{SLH}$, as follows:

We first prove that, if $\mathcal{S}$ is a solution to $\mathcal{I}_{MCK}$, then $H'$ is a solution to $\mathcal{I}_{SLH}$.
Since $\sum_{i\in[1,m]}\sum_{j\in C_i}w_{ij}\cdot x_{ij}= b$, $f(L_{m+1})$ is decreased by $b$. Thus, $H'[m+1]=0$ (i.e., all visits to $L_{m+1}$ were transferred to nonsensitive locations) and $\sum_{i\in[1,m+1]}H'[i]=\sum_{i~s.t.~L_i\notin L'}H'[i]=|H|_1$ (i.e., $H'$ has the same size with $H$). By construction, $H'$ has also the same length with $H$.
Since $\sum_{i\in[1,m]}\sum_{j\in C_i}c_{ij}\cdot x_{ij}$ is minimum, $\sum_{i\in[1,m]}q(H,H'[i])$ is minimum. In addition, $q(H,H'[m+1])$ (i.e., the loss for  transferring all visits from $L_{m+1}$ to non-sensitive locations) is constant. Thus, $d_q(H,H')=$ $\sum_{i\in[1,m+1]}q(H,H'[i])$ is minimum. Therefore, $H'$ is a solution to $\mathcal{I}_{SLH}$.

We now prove that, if $H'$ is a solution to $\mathcal{I}_{SLH}$, then $\mathcal{S}$ is a solution to $\mathcal{I}_{MCK}$. Since $\sum_{i~s.t.~L_i\notin L'}H'[i]=|H|_1$, it holds that $H'[m+1]=0$. Thus, $f(L_{m+1})$ is decreased by $b$ (all visits to $L_{m+1}$ were transferred to nonsensitive locations) and $\sum_{i\in[1,m]}\sum_{j\in C_i}w_{ij}\cdot x_{ij}= b$. Since $d_q(H,H')=(\sum_{i\in[1,m]}q(H,H'[i]))+q(H,H'[m+1])$ is minimum and
$q(H,H'[m+1])$ is constant, $\sum_{i\in[1,m]}q(H,H'[i])$ is minimum. This implies that $\sum_{i\in[1,m]}\sum_{j\in C_i}c_{ij}\cdot x_{ij}$ is minimum. Thus, $\mathcal{S}$ is a solution to $\mathcal{I}_{MCK}$.

Therefore, the special case of the {\SLH} problem with $|L'|=1$ is weakly NP-hard, and, clearly, the {\SLH} problem with $|L'|\geq 1$, is also weakly NP-hard.

\subsection{Proof of weak NP-hardness for the {\TR} problem}\label{subsec:TR_NP_hard}
We reduce the weakly NP-hard \emph{Multiple Choice Knapsack} ($MCK_{\geq}$) problem \cite{Kellerer2004,MMCK1} to the {\TR} problem. The $MCK_{\geq}$ problem is defined as follows:

\begin{equation}
	\min \sum_{i\in[1,n]}\sum_{j\in C_i}c_{ij}\cdot x_{ij}\label{preq1b}
\end{equation}
subject to:~ (I) $\quad \sum_{i\in[1,n]}\sum_{j\in C_i}w_{ij}\cdot x_{ij}\geq b$, (II) $\quad \sum_{j\in C_i}x_{ij}=1, ~~ i=1,\ldots n$, and (III) $\quad x_{ij}\in \{0,1\}, ~~ i=1,\ldots,n, ~~~~ j\in C_i$. 
In $MCK_{\geq}$, we are given a set of elements subdivided into $n$, mutually exclusive classes, $C_1,...,C_n$, and a knapsack. Each class $C_i$ has $|C_i|$ elements. Each element $j \in C_i$ has a cost $c_{ij} \geq 0$ and a weight $w_{ij}$.
The goal is to minimize the total cost (Eq.~\ref{preq1b}) by filling the knapsack with one element from each class (constraint II), such that the weights
of the elements in the knapsack satisfy the constraint I, where $b \geq 0$ is a constant.
The variable $x_{ij}$ takes a value $1$, if the element $j$ is chosen from class $C_i$ and $0$ otherwise (constraint III).

We map a given instance $\mathcal{I}_{MCK_{\geq}}$ to an instance $\mathcal{I}_{TR}$ of {\TR} in polynomial time, as follows: 
\begin{itemize}
	\item[(I)] Each class $C_i$, $i\in[1,n]$, is mapped to a location $L_i$, which has an arbitrary count in $H$ and a possibly different, arbitrary count in $H''$.
	\item[(II)] The constant $\epsilon$ is set to $n-\frac{b}{\max_{i\in[1,n],j\in C_i}{w_{ij}}}$.
	\item[(III)] We choose $q()$ and $p()$ to be normalized in $[0,1]$ such that each element $x_{ij}$ with weight $w_{ij}$ and cost $c_{ij}$ is mapped to a value $k_{ij}$ such that the following conditions hold:
	$q(H,H[i]+k_{ij})=1-\frac{w_{ij}}{\max_{i\in[1,n],j\in C_i}{w_{ij}}}$ and $p(H[i]+k_{ij},H'')=\frac{c_{ij}}{\max_{i\in[1,n],j\in C_i}{c_{ij}}}$. The normalization of $q()$ and $p()$ can be  done in polynomial time because $q()$ and $p()$ can take $O(N\cdot n)$ values. If there are multiple values of $k_{ij}$ satisfying the two conditions, one of these values is selected arbitrarily.
	When $x_{ij}=1$, the value $k_{ij}$ is added into $H[i]$ and $H'[i]=H[i]+k_{ij}$ is obtained.
\end{itemize}

We prove the correspondence between a solution $\mathcal{S}$ to $\mathcal{I}_{MCK_{\geq}}$ and a solution $H'$ to $\mathcal{I}_{TR}$, as follows:

We first prove that, if $\mathcal{S}$ is a solution to $\mathcal{I}_{MCK_{\geq}}$, then $H'$ is a solution to $\mathcal{I}_{TR}$.
Since $\sum_{i\in[1,n]}\sum_{j\in C_i}c_{ij}\cdot x_{ij}$ is minimum, $\sum_{i\in[1,n]}\left((\max_{i\in[1,n],j\in C_i}{c_{ij}})\cdot p(H'[i],H'')\right)$ is minimum. Thus,  $d_p(H',H'')=$\newline $\sum_{i\in[1,n]}p(H'[i],H'')$
is minimum.
Since $\sum_{i\in[1,n]}\sum_{j\in C_i}w_{ij}\cdot x_{ij}\geq b$ and $b=\max_{i\in[1,n],j\in C_i}{w_{ij}}\cdot (n-\epsilon)$, it holds that $\sum_{i\in[1,n]}(1-q(H,H[i]+k_{ij}))\geq n-\epsilon$. This implies
$n-d_q(H,H')\geq n -\epsilon$ and $d_q(H,H')\leq \epsilon$. Therefore, $H'$ is a solution to $\mathcal{I}_{TR}$.
We now prove that, if $H'$ is a solution to $\mathcal{I}_{TR}$, then $\mathcal{S}$ is a solution to $\mathcal{I}_{MCK_{\geq}}$. Since $d_p(H',H'')=\sum_{i\in[1,n]}p(H'[i],H'')$ is minimum, $\sum_{i\in[1,n]}\left((\max_{i\in[1,n],j\in C_i}c_{ij})\cdot p(H'[i],H''[i])\right)$ is minimum. This implies that
$\sum_{i\in[1,n]}\sum_{j\in C_i}c_{ij}\cdot x_{ij}$ is minimum. Since $d_q(H,H')=\sum_{i\in[1,n]}q(H,H'[i])\leq \epsilon$ and $\epsilon=n-\frac{b}{\max_{i\in[1,n],j\in C_i}{w_{ij}}}$, it holds that
$n-\sum_{i\in[1,n]}q(H,H'[i])\geq n-\epsilon=\frac{b}{\max_{i\in[1,n],j\in C_i}{w_{ij}}}$. This implies $\sum_{i\in[1,n]}(1-q(H,H[i]+k_{ij}))\geq \frac{b}{\max_{i\in[1,n],j\in C_i}{w_{ij}}}$ and $\sum_{i\in[1,n]}\sum_{j\in C_i}w_{ij}\cdot x_{ij} \geq b$. Thus, $\mathcal{S}$ is a solution to $\mathcal{I}_{MCK_{\geq}}$.

Therefore, {\TR} is weakly NP-hard.

\subsection{Proof of weak NP-hardness for the {\TA} problem}\label{subsec:TA_NP_hard}

We reduce the weakly NP-hard {\TR} problem (see Appendix \ref{subsec:TR_NP_hard}) to the {\TA} problem as follows.
We map a given instance $\mathcal{I}_{TR}$ of the {\TR} problem into an instance $\mathcal{I}_{TA}$ of the {\TA}  problem in polynomial time,
by mapping $H$ and $H''$ to histograms $H_{TA}=H$ and $H''_{TA}=H''$, respectively, defining $d_p(H'_{TA},H''_{TA})=\frac{1+1}{d_p(H',H'')+1}$ and $d_q(H_{TA},H'_{TA})=d_q(H,H')$, and setting $\epsilon_{TA}=\epsilon$. Note that $d_p(H',H'')\geq 0$ by definition, so $d_p(H'_{TA},H''_{TA})\in (0,2]$.

We prove the correspondence between a solution $H'$ to $\mathcal{I}_{TR}$ and a solution $H'_{TA}$ to $\mathcal{I}_{TA}$, as follows:

We first prove that, if $H'$ is a solution to $\mathcal{I}_{TR}$ then $H'_{TA}$ is a solution to $\mathcal{I}_{TA}$.
Since $H'$ is a solution to $\mathcal{I}_{TR}$, $d_p(H',H'')$ is minimum. Thus, $d_p(H'_{TA},H''_{TA})=\frac{1+1}{dp(H',H'')+1}$ is maximum.
In addition, $d_q(H,H')\leq \epsilon$. Thus, $d_q(H_{TA},H'_{TA})=d_q(H,H')\leq \epsilon$. Therefore, $H'_{TA}$ is a solution to {\TA}.

We now prove that, if $H'_{TA}$ is a solution to $\mathcal{I}_{TA}$ then $H'$ is a solution to $\mathcal{I}_{TR}$.
Since $H'_{TA}$ is a solution to $\mathcal{I}_{TA}$, $d_p(H'_{TA},H''_{TA})$ is maximum. Thus, $d_p(H',H'')=\frac{1+1}{d_p(H'_{TA},H''_{TA})}-1$ is minimum.
Also, $d_q(H_{TA},H'_{TA})\leq \epsilon$. Thus, $d_q(H,H')=d_q(H_{TA},H'_{TA})\leq \epsilon$. Therefore, $H'$ is a solution to $\mathcal{I}_{TR}$.

Therefore, {\TA} is weakly NP-hard.

\subsection{Reduction from the {\TA} to the {\TR} problem}\label{subsec:TA_TR_reduction}

The {\TA} problem can be reduced to the {\TR} problem in polynomial time, as follows.
Given an instance $\mathcal{I}_{TA}$ of the {\TA} problem, we can construct an instance $\mathcal{I}_{TR}$ of the {\TR} problem in polynomial time,
by mapping $H$ and $H''$ to histograms $H_{TR}=H$ and $H''_{TR}=H''$, defining $d_p(H'_{TR},H''_{TR})=\frac{1+1}{d_p(H',H'')+1} \in (0,2]$ and $d_q(H_{TR},H'_{TR})=d_q(H,H')$, and setting $\epsilon_{TR}=\epsilon$. Given a feasible solution $H'_{TR}$ of $\mathcal{I}_{TR}$, we can map it back to a feasible solution $H'$ of $\mathcal{I}_{TA}$ with cost
$d_p(H', H'')=\frac{1 + 1}{d_p(H'_{TR}, H''_{TR})} - 1 \geq 0$ in polynomial time. This requires constructing $H'=H'_{TR}$, which clearly is a solution to $\mathcal{I}_{TA}$.

\subsection{Results for $L_2$ distance for the $LHO$ algorithm}\label{L2appendix}

This appendix provides results for the $L_2$ distance. The results in
Figures \ref{appendixLHOn},  \ref{appendixLHOK}, and \ref{appendixLHOL} are qualitatively similar to those in Figures \ref{fig1}, \ref{fig2}, and \ref{fig4}, respectively.
\begin{figure*}[ht!]\hspace{-2mm}
	\begin{subfigure}[b]{.24\textwidth}\centering
		\includegraphics[scale=0.3,clip]{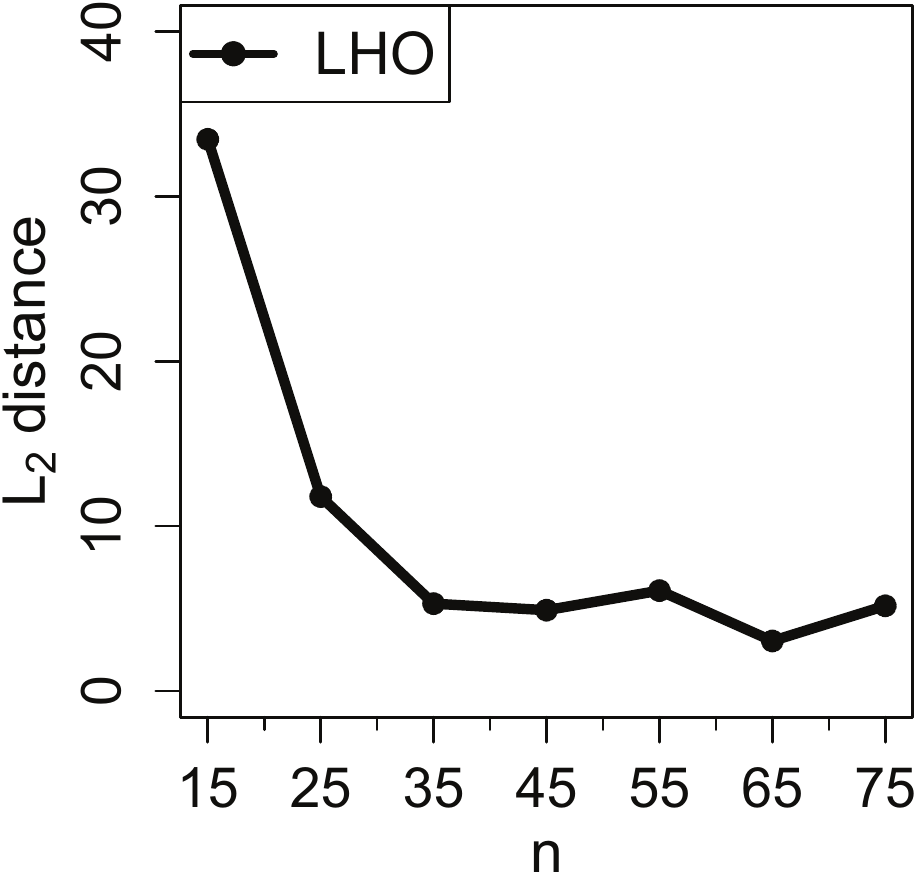}
		\caption{NYC}
	\end{subfigure}\hspace{+1mm}
	\begin{subfigure}[b]{.24\textwidth}\centering
		\includegraphics[scale=0.3,clip]{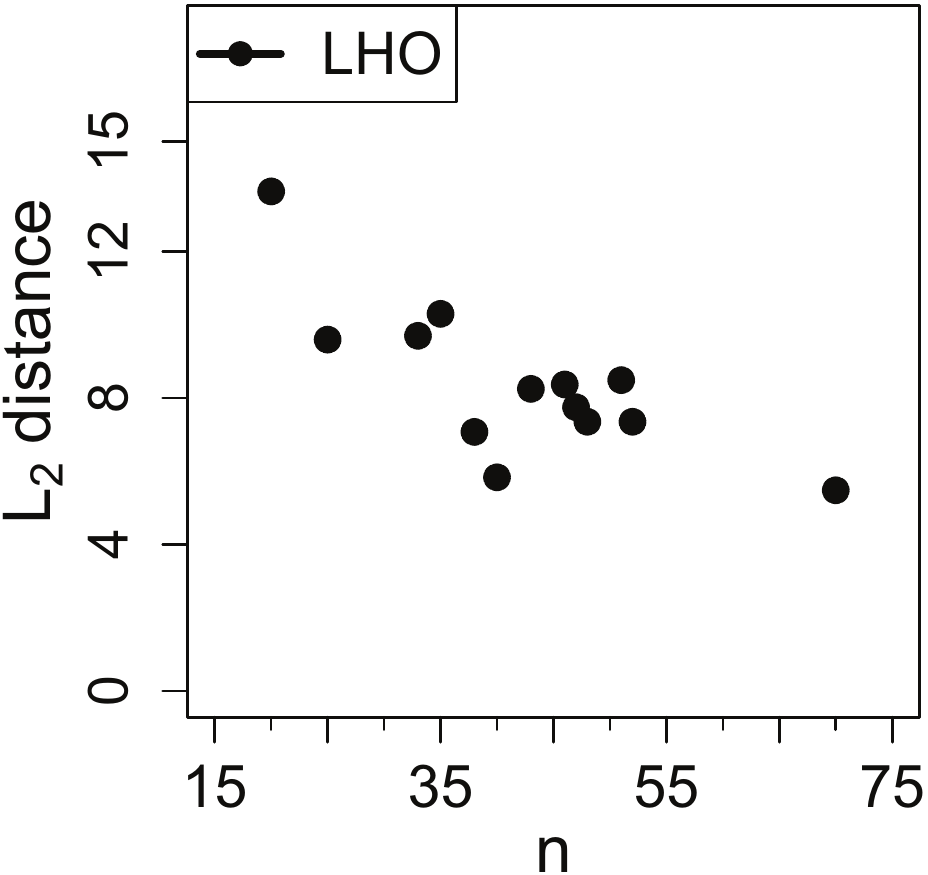}
		\caption{NYC}
	\end{subfigure}\hspace{+1mm}
	\begin{subfigure}[b]{.24\textwidth}\centering
		\includegraphics[scale=0.3,clip]{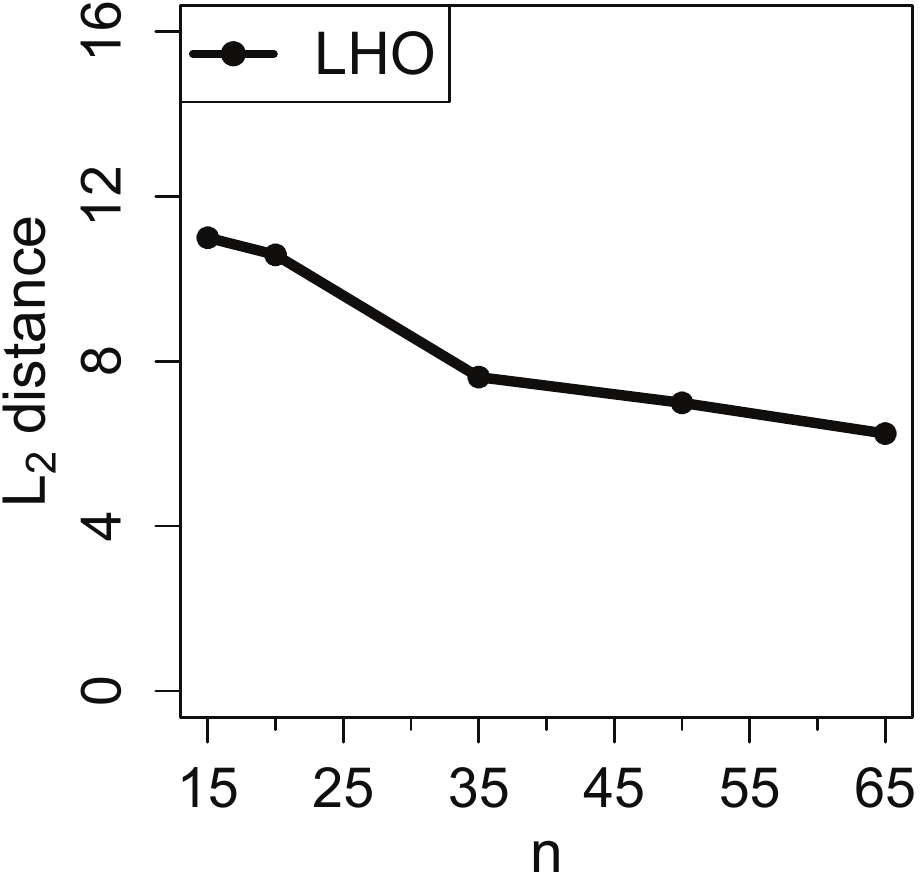}
		\caption{TKY}
	\end{subfigure}\hspace{+1mm}
	\begin{subfigure}[b]{.24\textwidth}\centering
		\includegraphics[scale=0.3,clip]{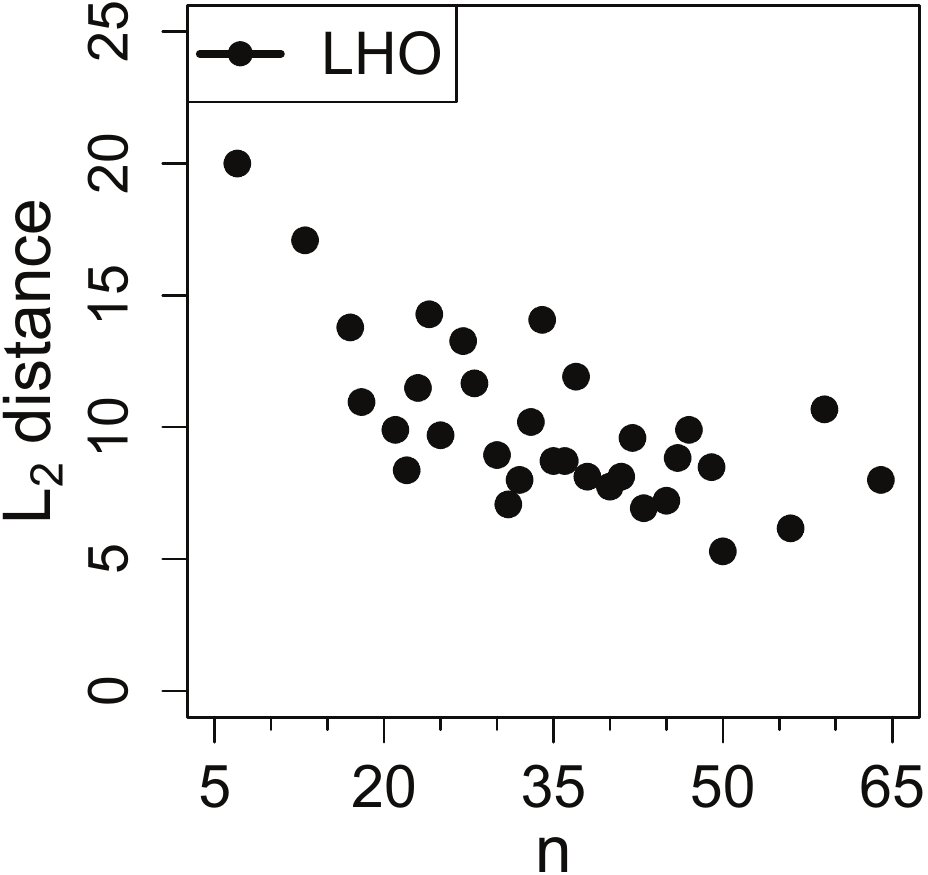}
		\caption{TKY}
	\end{subfigure}
	\caption{$L_2$ distance vs Length $n$ : (a) Median $L_2$ distance for histograms of length $n$ in \emph{NYC}. (b) $L_2$ distance for each histogram with $K=20$ in \emph{NYC}. (c) Median $L_2$ distance for histograms of length $n$ in \emph{TKY}. (d) $L_2$ distance for each histogram with $K=20$ in \emph{TKY}}\label{appendixLHOn}
\end{figure*}
\vspace{+2mm}

\begin{figure*}[ht!]\hspace{-2mm}
	\begin{subfigure}[b]{.24\textwidth}\centering
		\includegraphics[scale=0.3,clip]{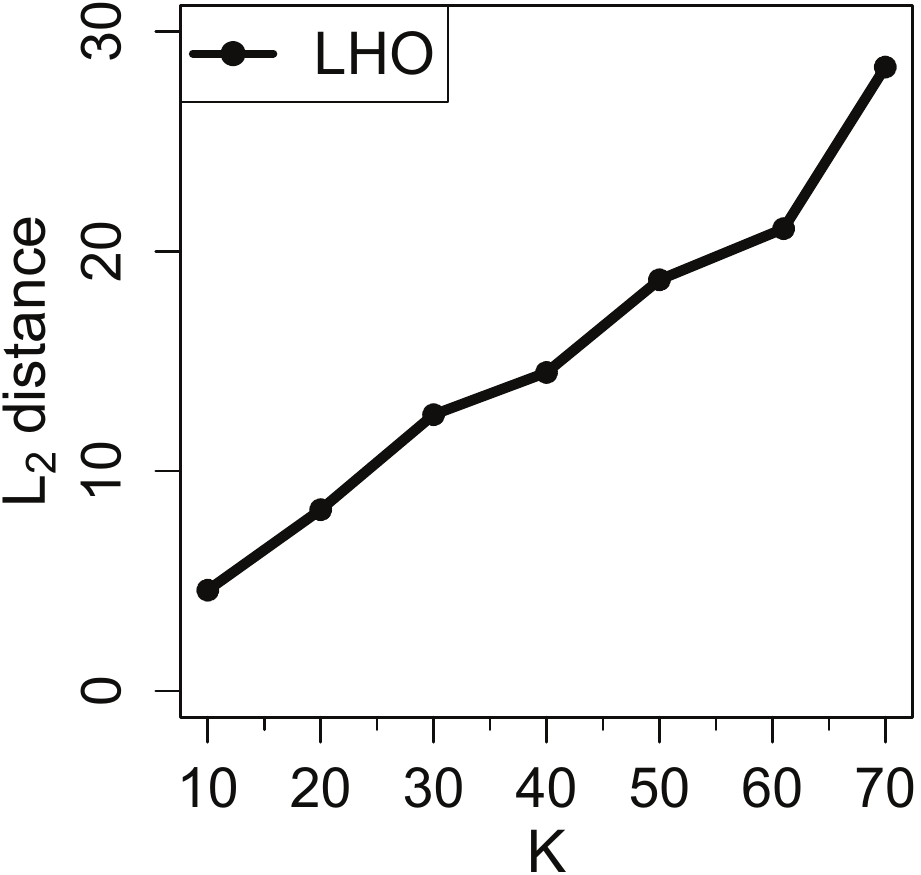}
		\caption{NYC}
	\end{subfigure}\hspace{+1mm}
	\begin{subfigure}[b]{.24\textwidth}\centering
		\includegraphics[scale=0.3,clip]{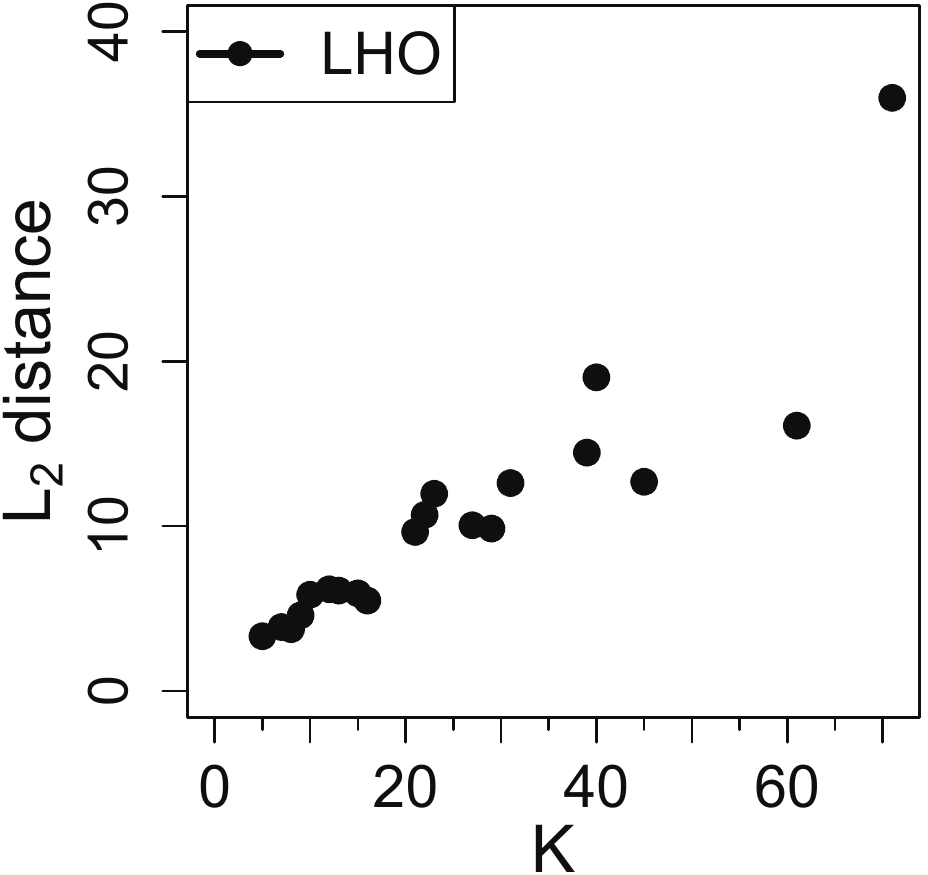}
		\caption{NYC}
	\end{subfigure}\hspace{+1mm}
	\begin{subfigure}[b]{.24\textwidth}\centering
		\includegraphics[scale=0.3,clip]{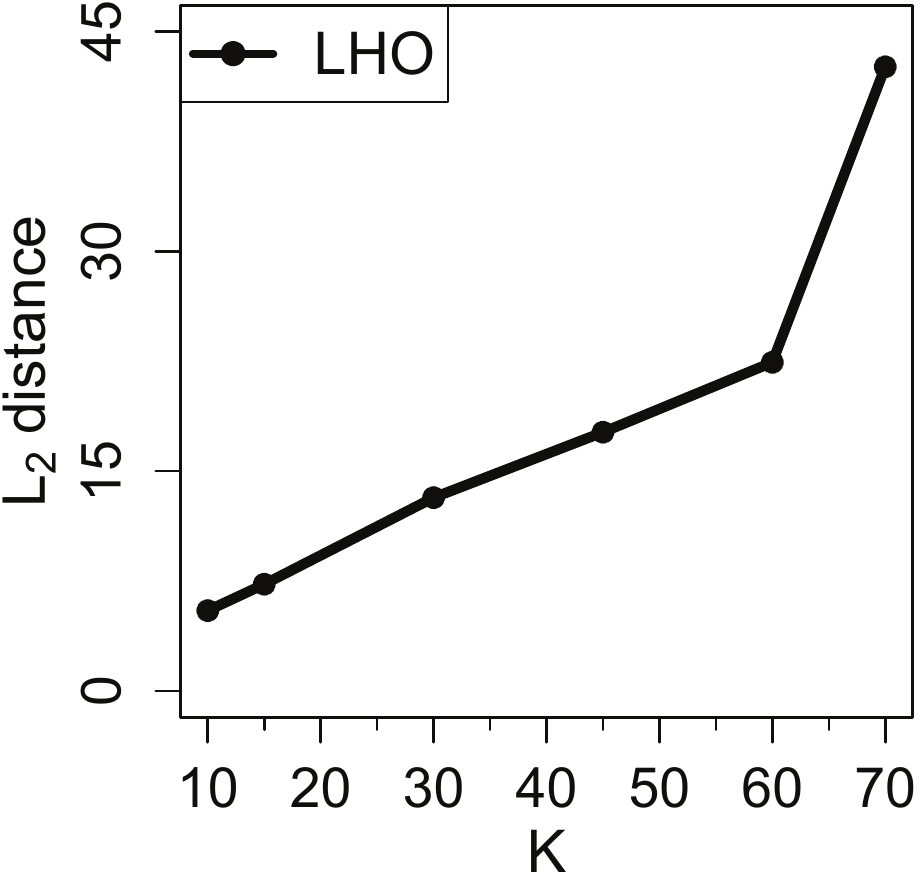}
		\caption{TKY}
	\end{subfigure}\hspace{+1mm}
	\begin{subfigure}[b]{.24\textwidth}\centering
		\includegraphics[scale=0.3,clip]{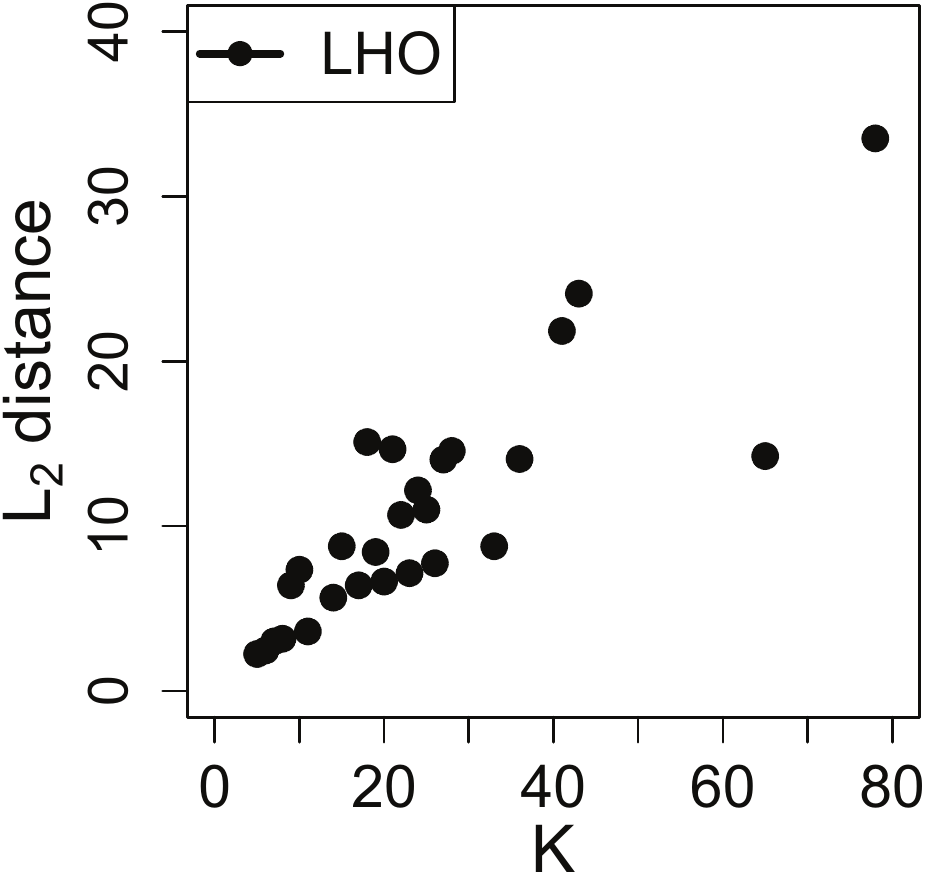}
		\caption{TKY}
	\end{subfigure}
	\caption{$L_2$ distance vs Total frequency of sensitive locations $K$: (a) Median $L_2$ distance for histograms of varying $K$ in \emph{NYC}. (b) $L_2$ distance for each histogram of length $n=30$ in \emph{NYC}. (c) Median $L_2$ distance for histograms of varying $K$ in \emph{TKY}. (d) $L_2$ distance for each histogram of length $n=40$ in \emph{TKY}}\label{appendixLHOK}
\end{figure*}
\vspace{+2mm}

\begin{center}
	\begin{figure*}[ht!]\hspace{-2mm}
		\begin{subfigure}[b]{.45\textwidth}\centering
			\includegraphics[scale=0.3,clip]{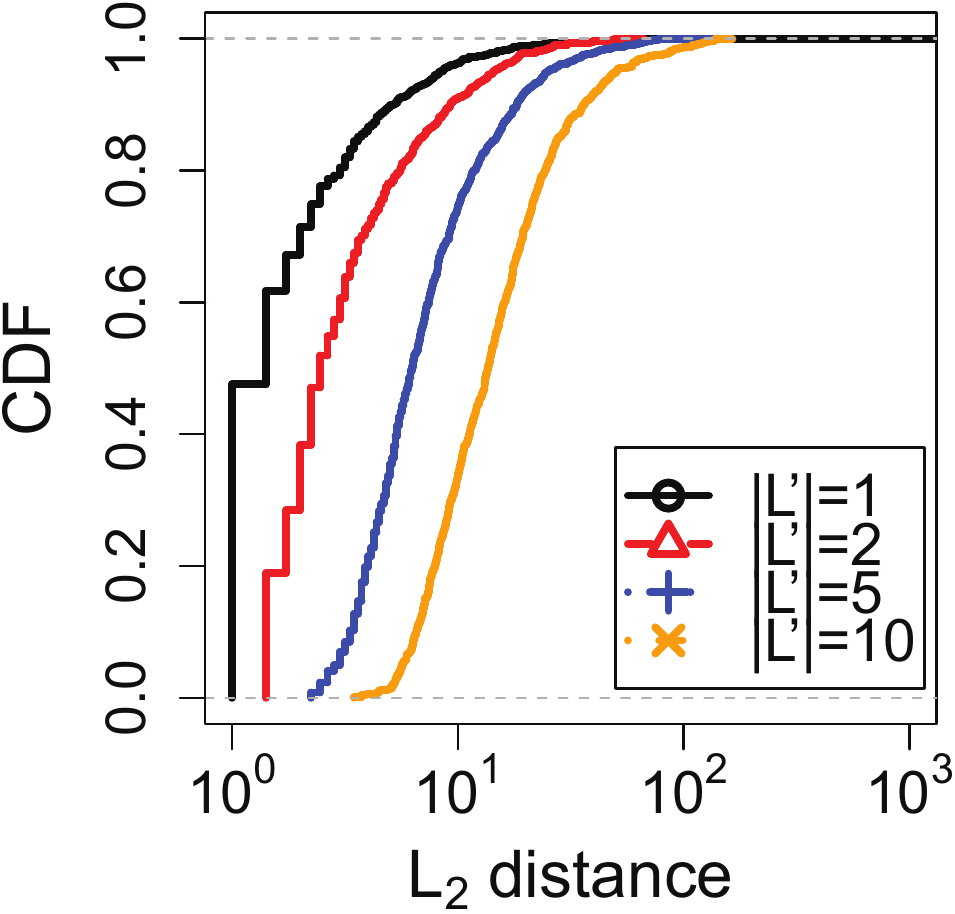}
			\caption{NYC}
		\end{subfigure}\hspace{+1mm}
		\begin{subfigure}[b]{.45\textwidth}\centering
			\includegraphics[scale=0.3,clip]{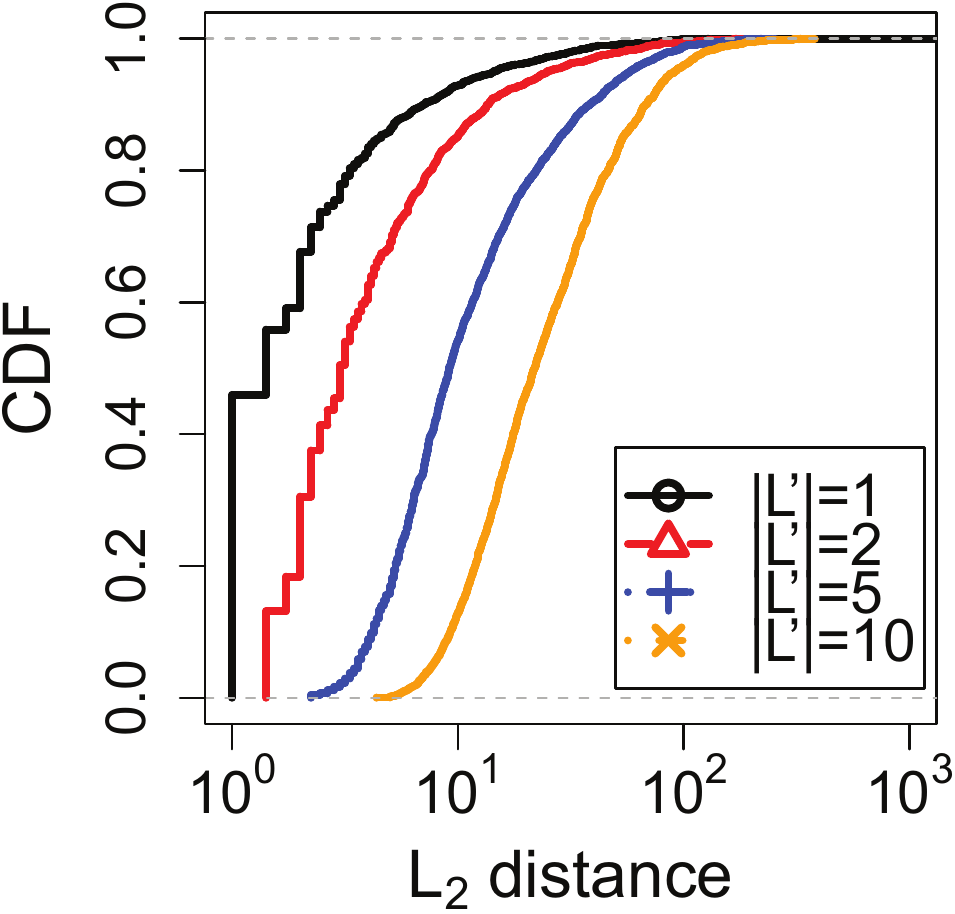}
			\caption{TKY}
		\end{subfigure}
		\caption{Cumulative Distribution Function of $L_2$ distance (i.e., a point $(x,y)$ means that a fraction $y$ of histograms have $L_2$ distance between $0$ and $x$) for varying $|L'|$ for: (a) \emph{NYC}, and (b) \emph{TKY}}\label{appendixLHOL}
\end{figure*}\end{center}
\subsection{Results for $L_2$ distance for the $RO$ algorithm and the $RH$ heuristic}\label{L2appendixRORH}

This appendix provides results for the $L_2$ distance. The results in Figures \ref{appendixROa}, \ref{appendixROb}, \ref{appendixROc}, and \ref{appendixROd}
are qualitatively similar to those in Figures \ref{fig7a}, \ref{fig7b}, \ref{fig8c}, and \ref{fig9a}, respectively.

\begin{figure*}[ht!]\hspace{-2mm}
	\begin{subfigure}[b]{.22\textwidth}\centering
		\includegraphics[scale=0.25,clip]{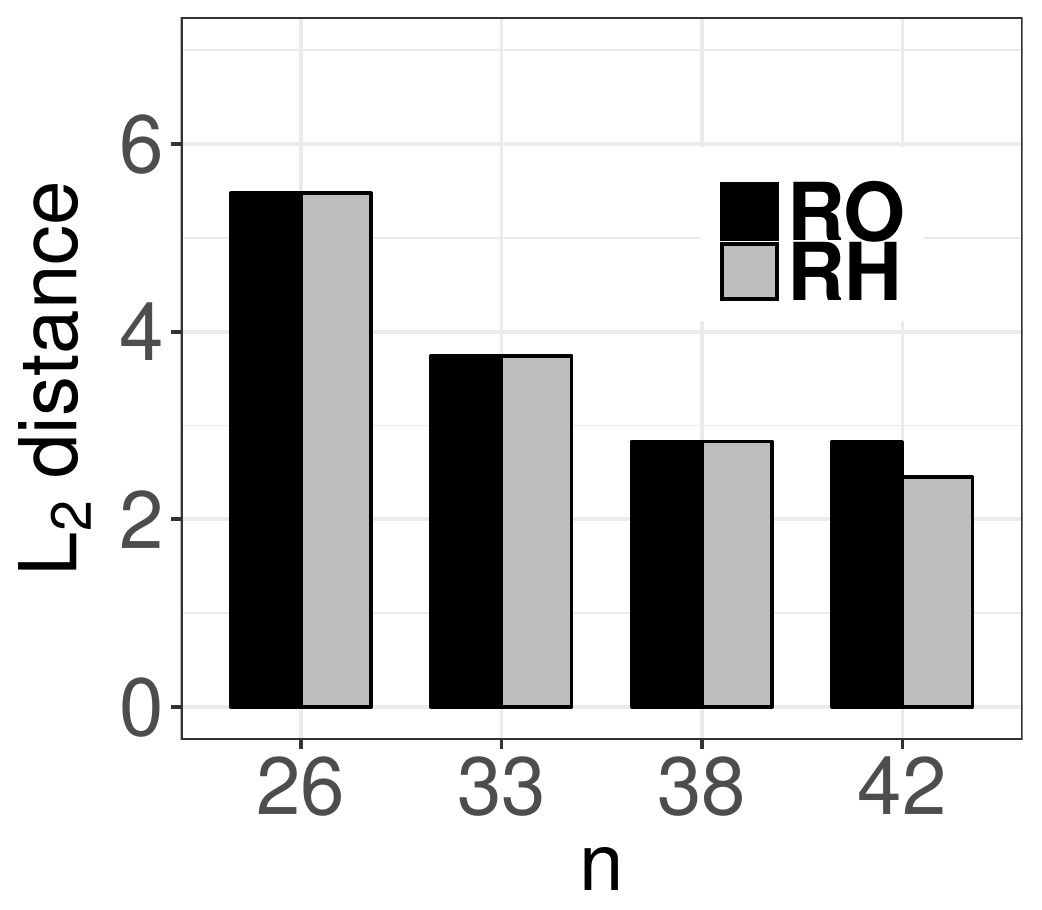}
		\caption{NYC}\label{appendixROa}
	\end{subfigure}\hspace{+1mm}
	\begin{subfigure}[b]{.22\textwidth}\centering
		\includegraphics[scale=0.25,clip]{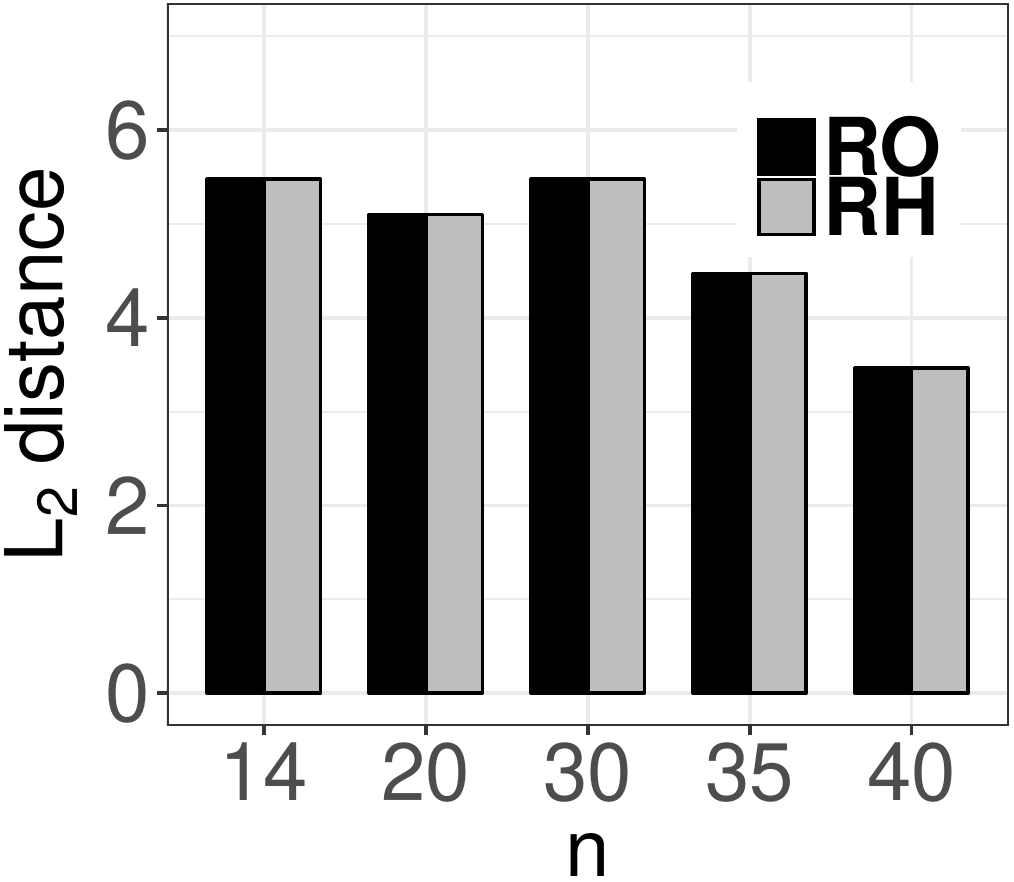}
		\caption{TKY}\label{appendixROb}
	\end{subfigure}\hspace{+1mm}
	\begin{subfigure}[b]{.24\textwidth}\centering
		\includegraphics[scale=0.25,clip]{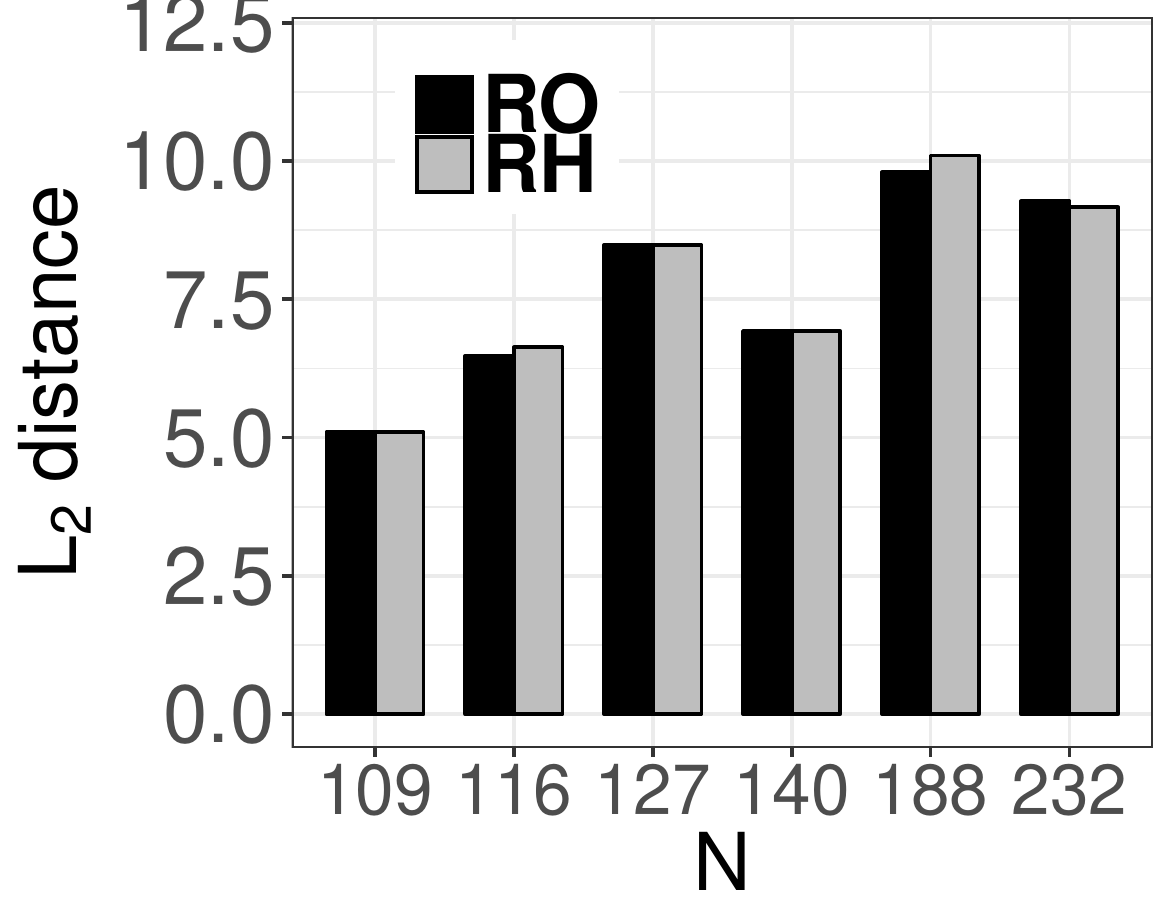}
		\caption{NYC}\label{appendixROc}
	\end{subfigure}\hspace{+1mm}
	\begin{subfigure}[b]{.28\textwidth}\centering
		\includegraphics[scale=0.25,clip]{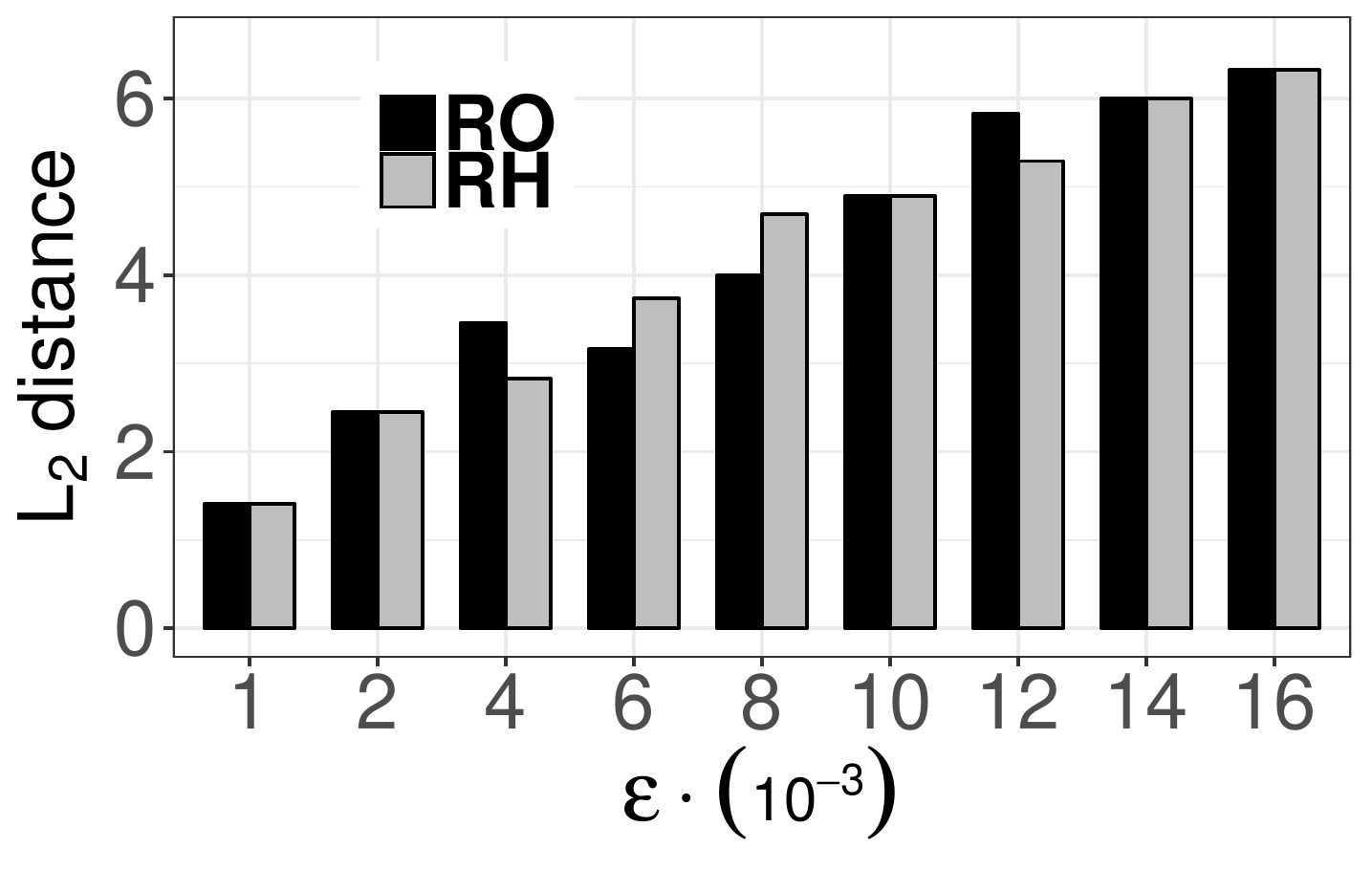}
		\caption{NYC}\label{appendixROd}
	\end{subfigure}
	\caption{$L_2$ distance vs (a) length $n$ for histograms with $N=100$ in \emph{NYC}, (b) length $n$ for histograms with $N=100$ in \emph{TKY},
		(c) size $N$ for histograms with $n=25$ in \emph{NYC}, and (d) threshold $\epsilon$ for a histogram with length $n=40$ and size $N=100$ in \emph{NYC} }\label{appendixRO}
\end{figure*}

%
%
\end{document}